\documentclass[11pt]{article}
\usepackage{amsmath,amsfonts,amssymb,amsthm,amstext,array}
\usepackage{commath}
\usepackage{bm,color,setspace,graphicx,times,subcaption}
\usepackage{mathtools}
\usepackage{hyperref}
\usepackage{algorithm}
\usepackage[noend]{algpseudocode}
\usepackage{comment}
\usepackage[margin=1in]{geometry}

\graphicspath{{Figs/}}

\usepackage[normalem]{ulem}

\definecolor{darkgrn}{rgb}{0, 0.8, 0}

\theoremstyle{plain}
\newtheorem{thm}{Theorem}[section]
\newtheorem{lem}[thm]{Lemma}
\newtheorem{prop}[thm]{Proposition}
\newtheorem{cor}[thm]{Corollary}
\newtheorem{defn}[thm]{Definition}
\newtheorem{claim}[thm]{Claim}
\newtheorem{asmn}[thm]{Assumption}
\newtheorem{con}[thm]{Conjecture}
\newtheorem{rem}[thm]{Remark}
\newtheorem{case}{Case}

\usepackage[nameinlink,noabbrev,capitalize]{cleveref} 
\crefalias{subequation}{equation}
\crefalias{thm}{theorem}

\let\oldlemma\lem
\renewcommand{\lem}{%
  \crefalias{thm}{lem}
  \oldlemma}
\Crefname{lem}{Lemma}{Lemmas}

\let\olddefn\defn
\renewcommand{\defn}{%
  \crefalias{thm}{defn}
  \olddefn}
\Crefname{defn}{Definition}{Definitions}

\let\oldrem\rem
\renewcommand{\rem}{%
  \crefalias{thm}{rem}
  \oldrem}
\Crefname{rem}{Remark}{Remarks}

\let\oldcor\cor
\renewcommand{\cor}{%
  \crefalias{thm}{cor}
  \oldcor}
\Crefname{cor}{Corollary}{Corollaries}

\let\oldclaim\claim
\renewcommand{\claim}{%
  \crefalias{thm}{claim}
  \oldclaim}
\Crefname{claim}{Claim}{Claims}

\let\oldprop\prop
\renewcommand{\prop}{%
  \crefalias{thm}{prop}
  \oldprop}
\Crefname{prop}{Proposition}{Propositions}

\let\oldcon\con
\renewcommand{\con}{%
  \crefalias{thm}{con}
  \oldcon}
\Crefname{con}{Conjecture}{Conjectures}

\let\oldasmn\asmn
\renewcommand{\asmn}{%
  \crefalias{thm}{asmn}
  \oldasmn}
\Crefname{asmn}{Assumption}{Assumptions}

\newcommand{\R}{ {\mathbb R} }

\newcommand{\hK}{\hat{K}}
\newcommand{\hk}{\hat{k}}
\newcommand{\hV}{\hat{V}}

\newcommand{\hE}{\hat{E}}
\newcommand{\he}{\hat{e}}
\newcommand{\hF}{\hat{F}}

\newcommand{\hG}{\hat{G}}
\newcommand{\hu}{\hat{u}}
\newcommand{\hv}{\hat{v}}
\newcommand{\hf}{\hat{f}}

\newcommand{\CH}{\mathcal{C}^H}
\newcommand{\CO}{\mathcal{C}^O}
\newcommand{\hCH}{\hat{\mathcal{C}}^H}
\newcommand{\hCO}{\hat{\mathcal{C}}^O}

\newcommand{\tK}{\tilde{K}}
\newcommand{\tG}{\tilde{G}}
\newcommand{\tR}{\tilde{R}}

\newcommand{\te}{\tilde{e}}
\newcommand{\tf}{\tilde{f}}
\newcommand{\tv}{\tilde{v}}
\newcommand{\tP}{\tilde{P}}

\newcommand{\vx}{\mathbf{x}}
\newcommand{\vy}{\mathbf{y}}

\newcommand{\Ps}{\mathcal{P}} 

\title{Continuous Toolpath Planning in a Graphical Framework for Sparse Infill Additive Manufacturing}

\author{
  Prashant Gupta$^\dag$ \hspace*{0.5in}   Bala Krishnamoorthy$^\dag$ \hspace*{0.5in} Gregory Dreifus$^\ddag$\\
  \vspace*{-0.05in}\\
	\dag: Department of Mathematics and Statistics, Washington State University \\
	\ddag: Department of Mechanical Engineering, Massachusetts Institute of Technology\\
	\{prashant.gupta,kbala\}@wsu.edu, gdreifus@mit.edu
}
\begin{document}
\maketitle

\begin{abstract} 
  We develop a framework that creates a new polygonal mesh representation of the sparse infill domain of a layer-by-layer 3D printing job.
  We guarantee the existence of a single, \emph{continuous} tool path covering each connected piece of the domain in every layer in this graphical model.
  We also present a tool path algorithm that traverses each such continuous tool path with \emph{no crossovers}.
                                                                        
  The key construction at the heart of our framework is a novel \emph{Euler transformation} which converts a $2$-dimensional cell complex $K$ into a new $2$-complex $\hK$ such that every vertex in the $1$-skeleton $\hG$ of $\hK$ has even degree.
  Hence $\hG$ is Eulerian, and an Eulerian tour can be followed to print all edges in a continuous fashion without stops.

  We start with a mesh $K$ of the union of polygons obtained by projecting all layers to the plane.
  First we compute its Euler transformation $\hK$.
  In the \emph{slicing} step, we clip $\hK$ at each layer using its polygon to obtain a complex that may not necessarily be Euler.
  We then \emph{patch} this complex by adding edges such that any odd-degree nodes created by slicing are transformed to have even degrees again.
  We print extra \emph{support} edges in place of any segments left out to ensure there are no edges without support in the next layer above.
  These support edges maintain the Euler nature of the complex.
  Finally, we describe a tree-based search algorithm that builds the continuous tool path by traversing ``concentric'' cycles in the Euler complex.
  Our algorithm produces a tool path that avoids material collisions and crossovers, and can be printed in a continuous fashion irrespective of complex geometry or topology of the domain (e.g., holes).

  We implement our test our framework on several 3D objects.
  Apart from standard geometric shapes including a nonconvex star, we demonstrate the framework on the Stanford bunny.
  Several intermediate layers in the bunny have multiple components as well as complicated geometries.
\end{abstract}

\section{Introduction} \label{sec:intro}
\emph{Additive manufacturing} refers to any process that adds material  to create a 3D object. %
3D printing is a popular form of additive manufacturing that deposits material (plastic, metal, biomaterial, polymer, etc.) in layer by layer fashion. 
We focus on extrusion based 3D printing, in which material is pushed out of an extruder that follows some tool path while depositing material in beads that meld together upon contact.
In this paper, we will refer to this process simply as 3D printing.

In \emph{sparse infill} 3D printing, we first print the outer ``shell'' or boundary of the 3D object in each layer.
We then cover the interior space by printing an \emph{infill lattice} \cite{WuAaWeSi2018}, which is typically a standard mesh. 
In an arbitrary infill lattice, one is not guaranteed to find a \emph{continuous} tool path, i.e., an entire layer being printed by non-stop extrusion of material.
Non-continuous tool paths typically have multiple starts and stops, which could reduce quality of the print, cause print failures (e.g., delamination), and could increase print time.
To ensure existence of a continuous tool path, we need to choose the mesh modeling the infill lattice carefully.
Subsequently, we need to develop algorithms that ensure a continuous tool path can be obtained for each layer with arbitrary geometry.
Further, we need to identify a traversel of this tool path that avoids crossovers.

\subsection{Our contributions} \label{ssec:contrib}
We recently proposed a method that transforms a given $2$-dimensional cell complex $K$ into a new $2$-complex $\hK$ in which every vertex is part of a uniform even number of edges \cite{GuKr2018}.
Hence every vertex in the graph $\hG$ that is the $1$-skeleton of $\hK$ has an even degree, which makes $\hG$ Eulerian, i.e., it is guaranteed to contain an Eulerian tour.
We refer to this method as an \emph{Euler transformation} of a polygonal (or cell) complex.
For $2$-complexes in $\R^2$ under mild assumptions (that no two adjacent edges of a polygon in $K$ are boundary edges), we show that the Euler transformed $2$-complex $\hK$ has a geometric realization in $\R^2$, and that each vertex in its $1$-skeleton has degree $4$.
We bound the numbers of vertices, edges, and polygons in $\hK$ as small scalar multiples of the corresponding numbers in $K$.

We present a computational framework for 3D printing that identifies continuous tool paths for printing the infill lattice in each layer. We illustrate the steps in our framework in Figure \ref{infill_pattern_work_flow} (on a 3D pyramid with a square base).
First we find the polygons for each layer of the input 3D domain (typically presented as an STL file, e.g., Figure \ref{pyramid})) using a slicing software.
Let $\Ps$ be the union of all of these polygons in 2D (Figure \ref{union_of_projections}). 
We fill the space in $\Ps$ with some infill lattice $K$, using any meshing algorithm (Figure \ref{infill_pattern}).
We then apply Euler transformation to obtain a new infill lattice $\hK$ that is guaranteed to be Euler (Figure \ref{infill_pattern_transformed}).
In the next step, we \emph{clip} $\hK$ using $P_i$, a polygon in layer $i$ (Figure \ref{infill_pattern_transformed_polygon_cut}).
Depending on the shape of $P_i$, this step could create terminal vertices in the infill lattice for layer $i$, making it no longer Euler.
In the last step, we \emph{patch} the clipped infill lattice by adding new edges such that the resulting infill lattice is Euler again (Figure \ref{infill_pattern_transformed_polygon_patch_simple}).
An application of this framework is illustrated in Figures \ref{pyramidplan} and \ref{printedpyramid}.
Finally, we propose a tool path algorithm (Section \ref{sec:toolpathalgo}) that identifies the actual print tool path from the patched Euler infill lattice that avoids crossovers and material collisions.
We address all geometric/computational challenges that arise along the way to ensure the proposed framework is complete.
Since each layer can have multiple polygons in general, our framework can generate continuous tool path for each polygon in a given layer (see Section \ref{sec:boundaryedges} for an exception arising in certain cases with extreme geometries).
As we might not print every boundary edge after the patching step, we also print \emph{support} edges (see Section \ref{sec:slicing}).
The overall goal of our framework is to create an Euler infill lattice in each layer, and also prevent printing in free space so as to avoid print failures.

\begin{figure}[htp!]
  \centering
  \begin{subfigure}[t]{3in}
    \centering
    \includegraphics[scale=0.30]{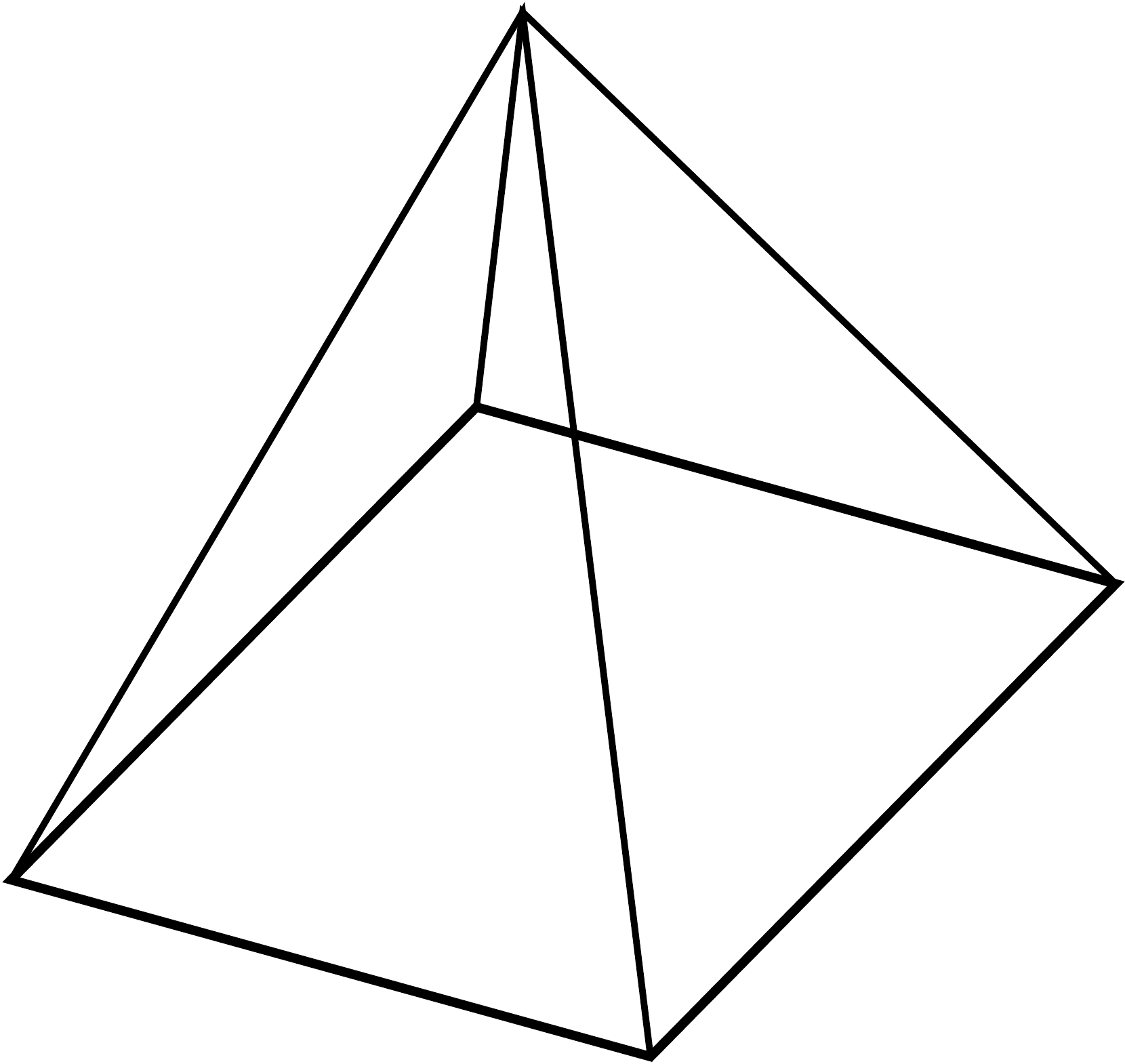}
    \caption{\small \label{pyramid}}
  \end{subfigure}
  \begin{subfigure}[t]{3in}
    \centering
    \includegraphics[scale=0.30]{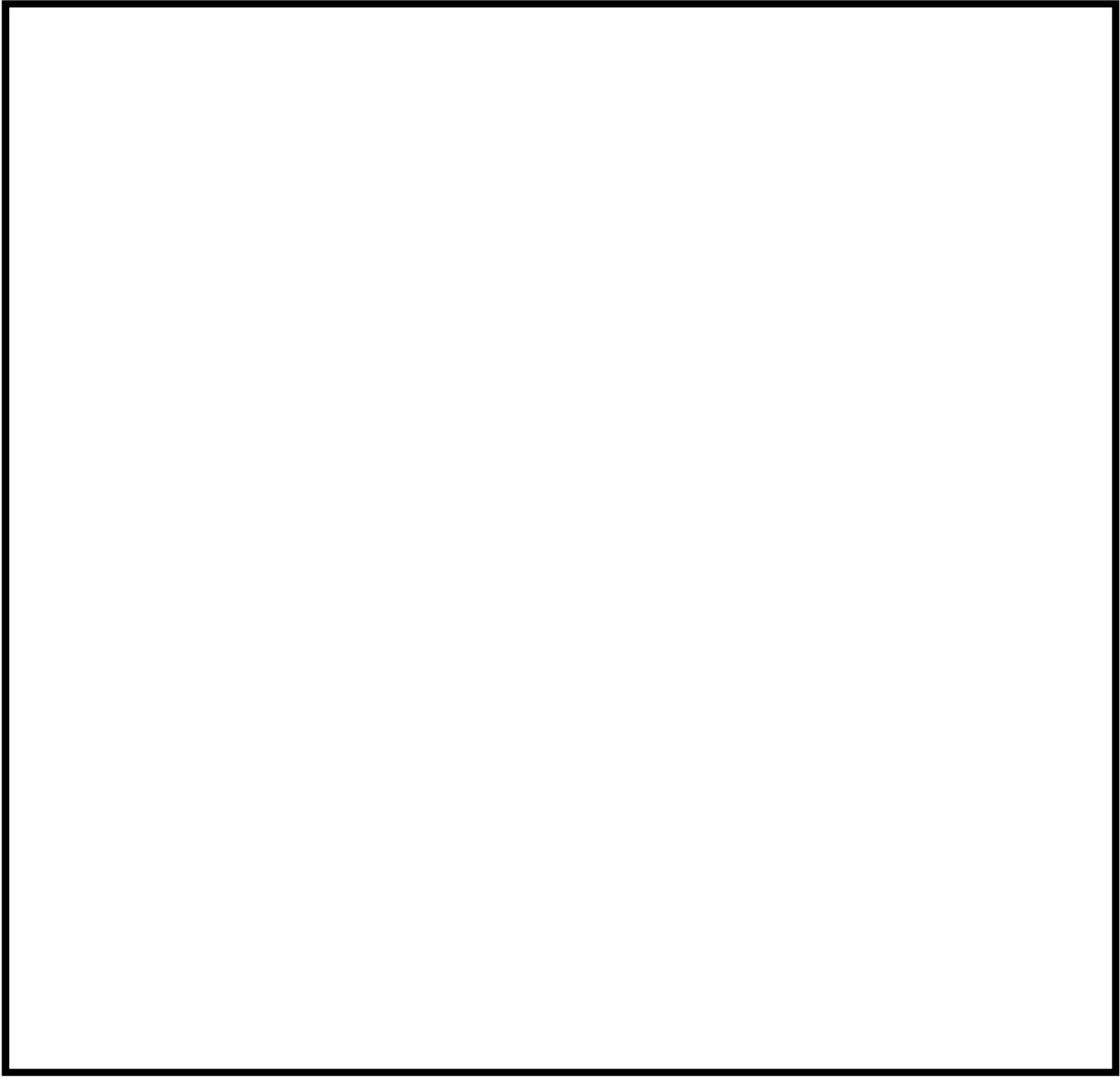}
    \caption{\small \label{union_of_projections}}
  \end{subfigure}
  \begin{subfigure}[t]{3in}
    \centering
    \includegraphics[scale=0.30]{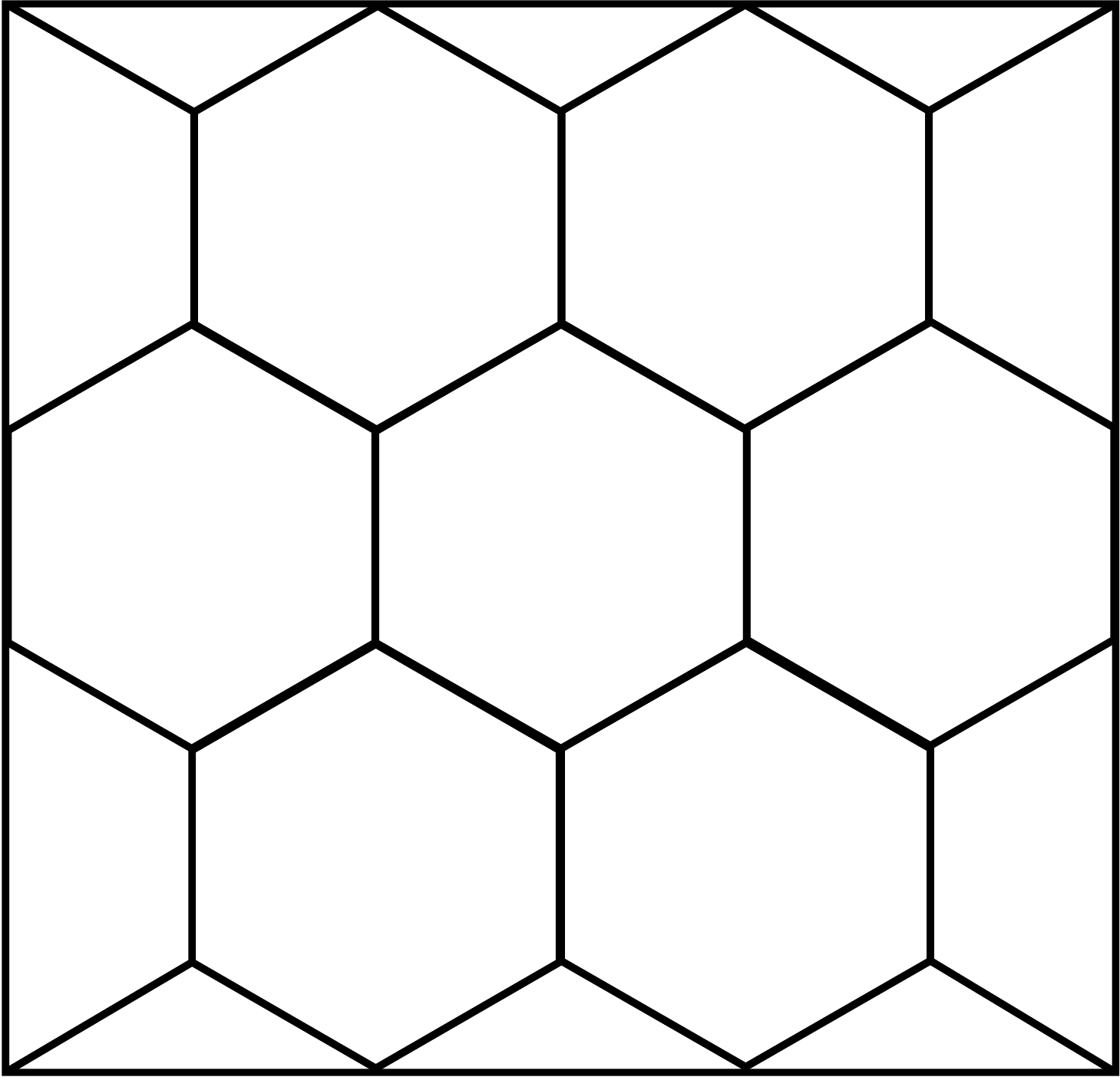}
    \caption{\small \label{infill_pattern}}
  \end{subfigure}
  \begin{subfigure}[t]{3in}
    \centering
    \includegraphics[scale=0.30]{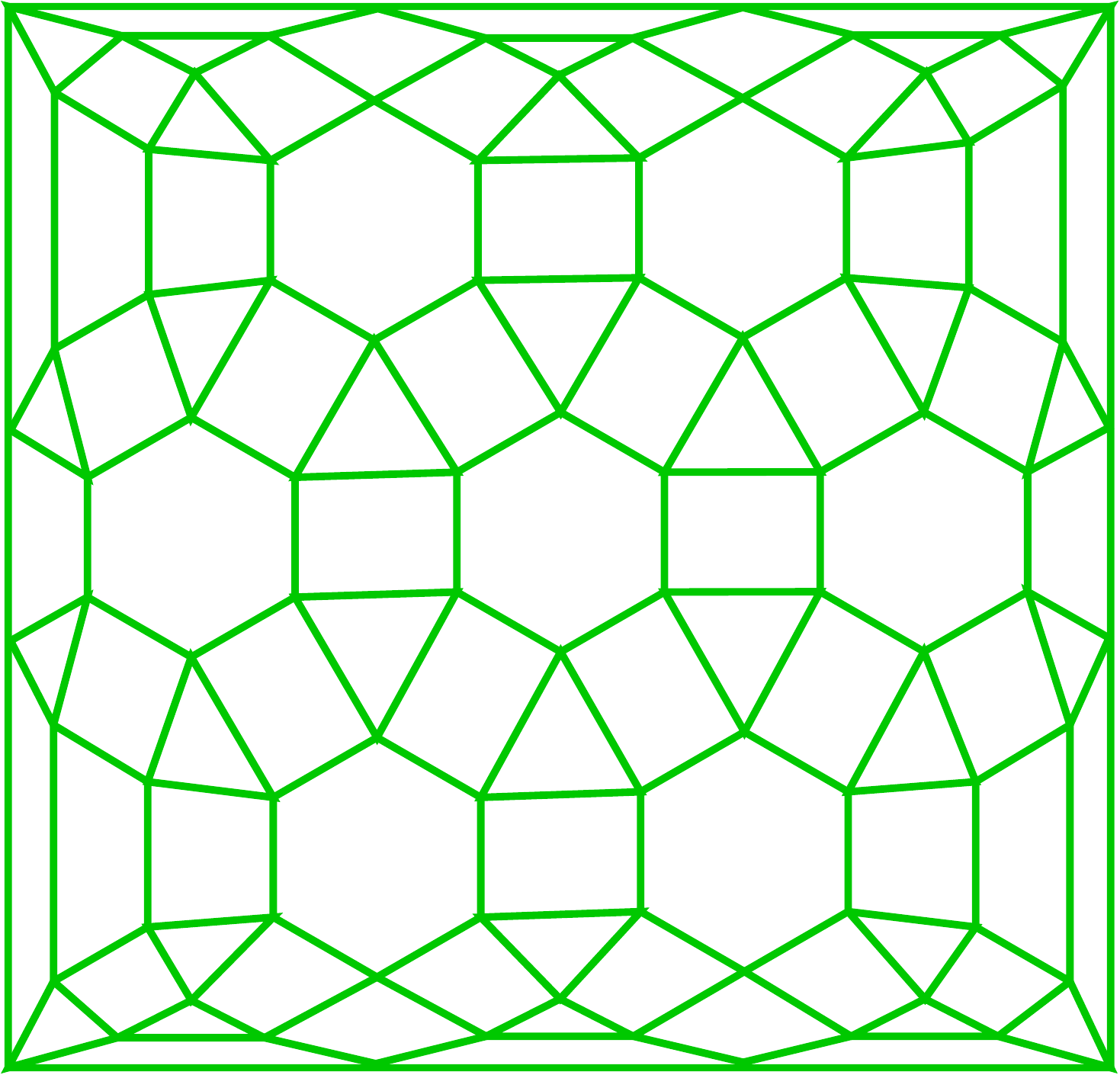}
    \caption{\small \label{infill_pattern_transformed}}
  \end{subfigure}
  \vspace*{0.05in}
  \begin{subfigure}[t]{3in}
    \centering
    \includegraphics[scale=0.30]{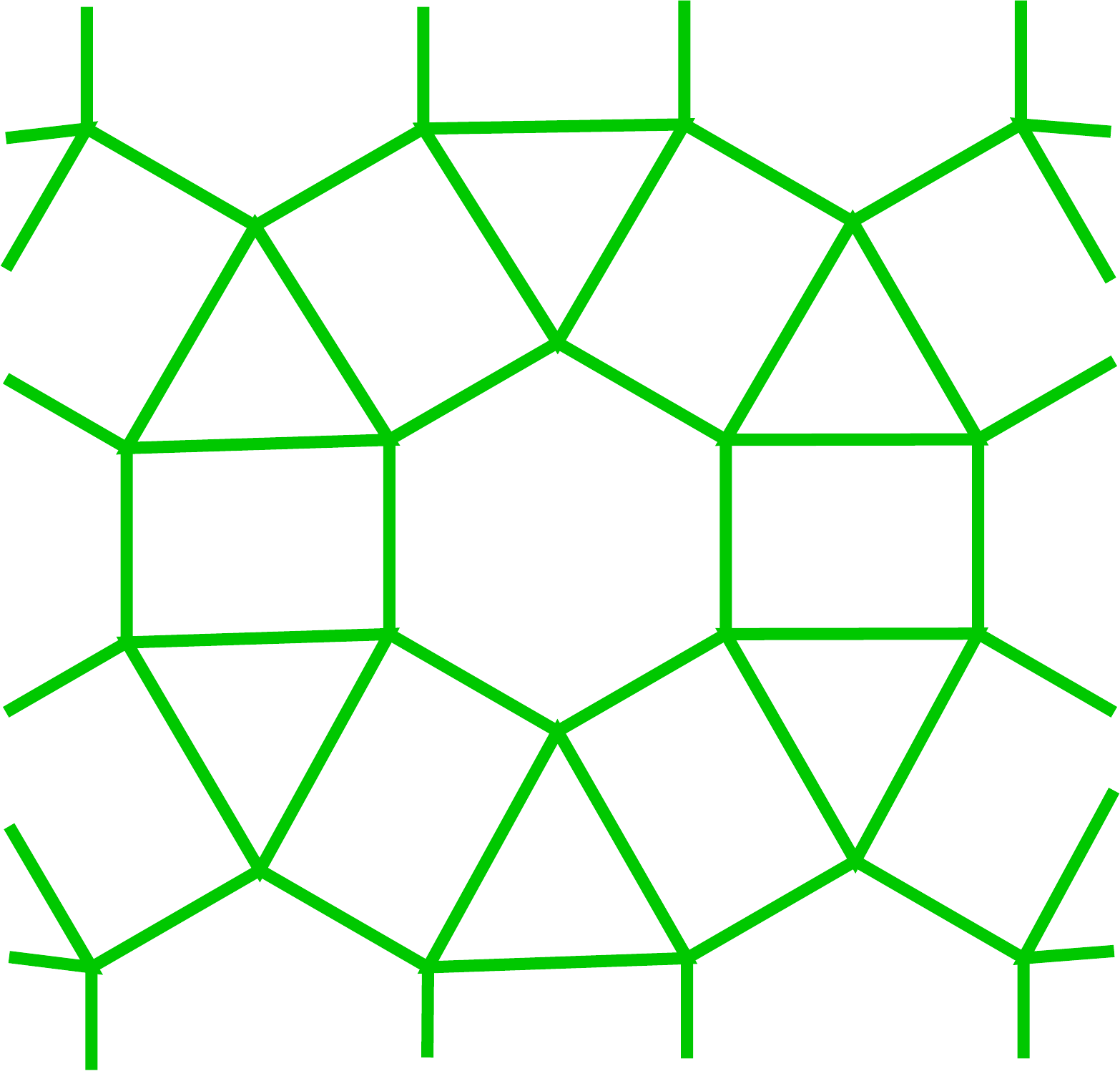}
    \caption{\small \label{infill_pattern_transformed_polygon_cut}}
  \end{subfigure}
  \begin{subfigure}[t]{3in}
    \centering
    \includegraphics[scale=0.30]{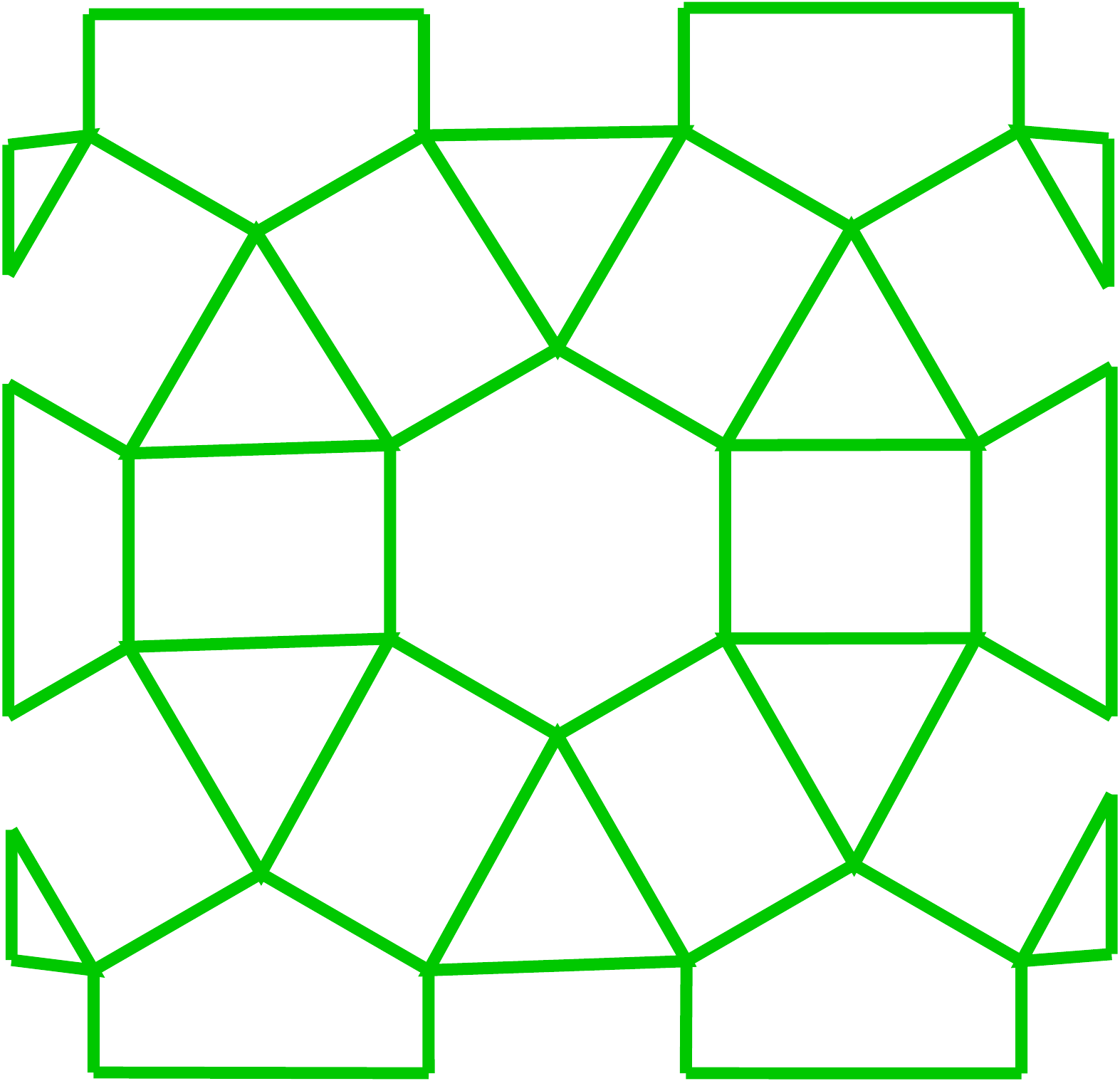}
    \caption{\small \label{infill_pattern_transformed_polygon_patch_simple}}
  \end{subfigure}
  \begin{subfigure}[t]{3in}
    \centering
    \includegraphics[scale=0.15, bb=0 0 1154 859]{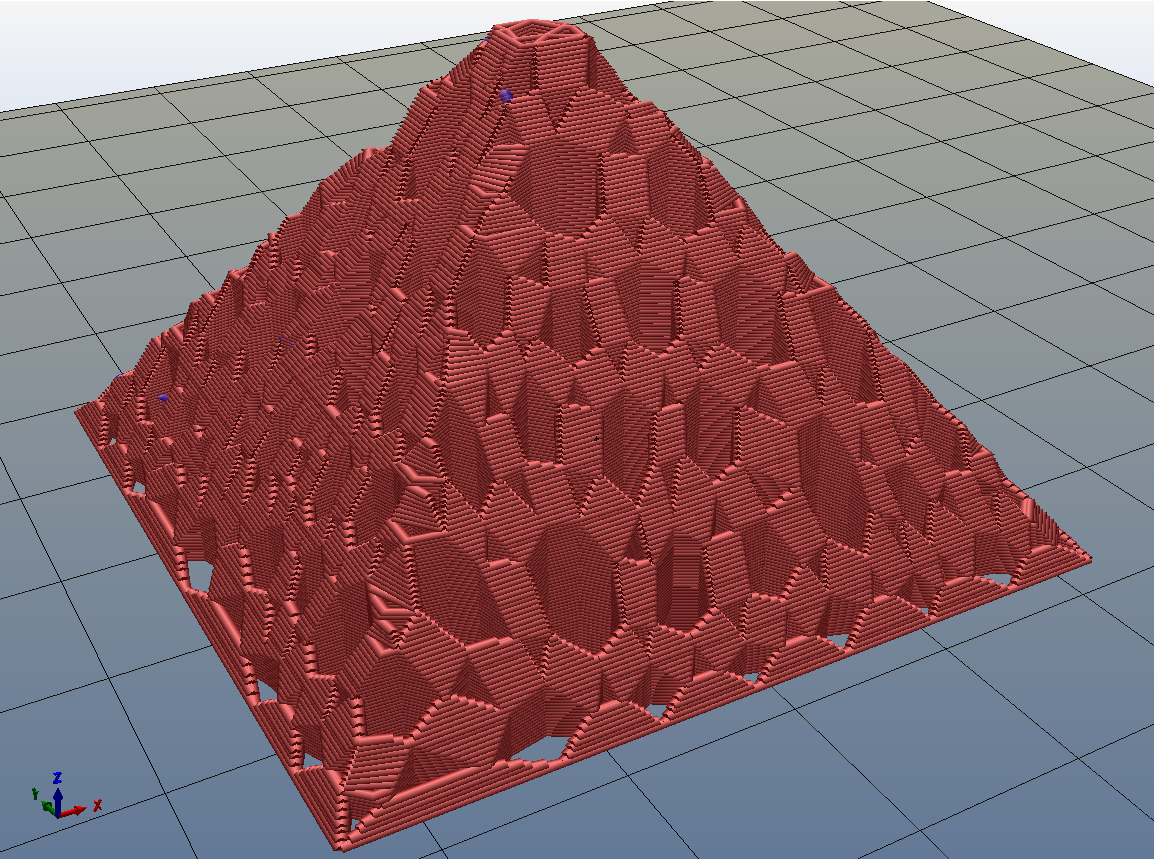} 
    \caption{\small \label{pyramidplan}}
  \end{subfigure}
  \begin{subfigure}[t]{3in}
    \centering
    \includegraphics[scale=0.075, bb=0 0 2448 1744]{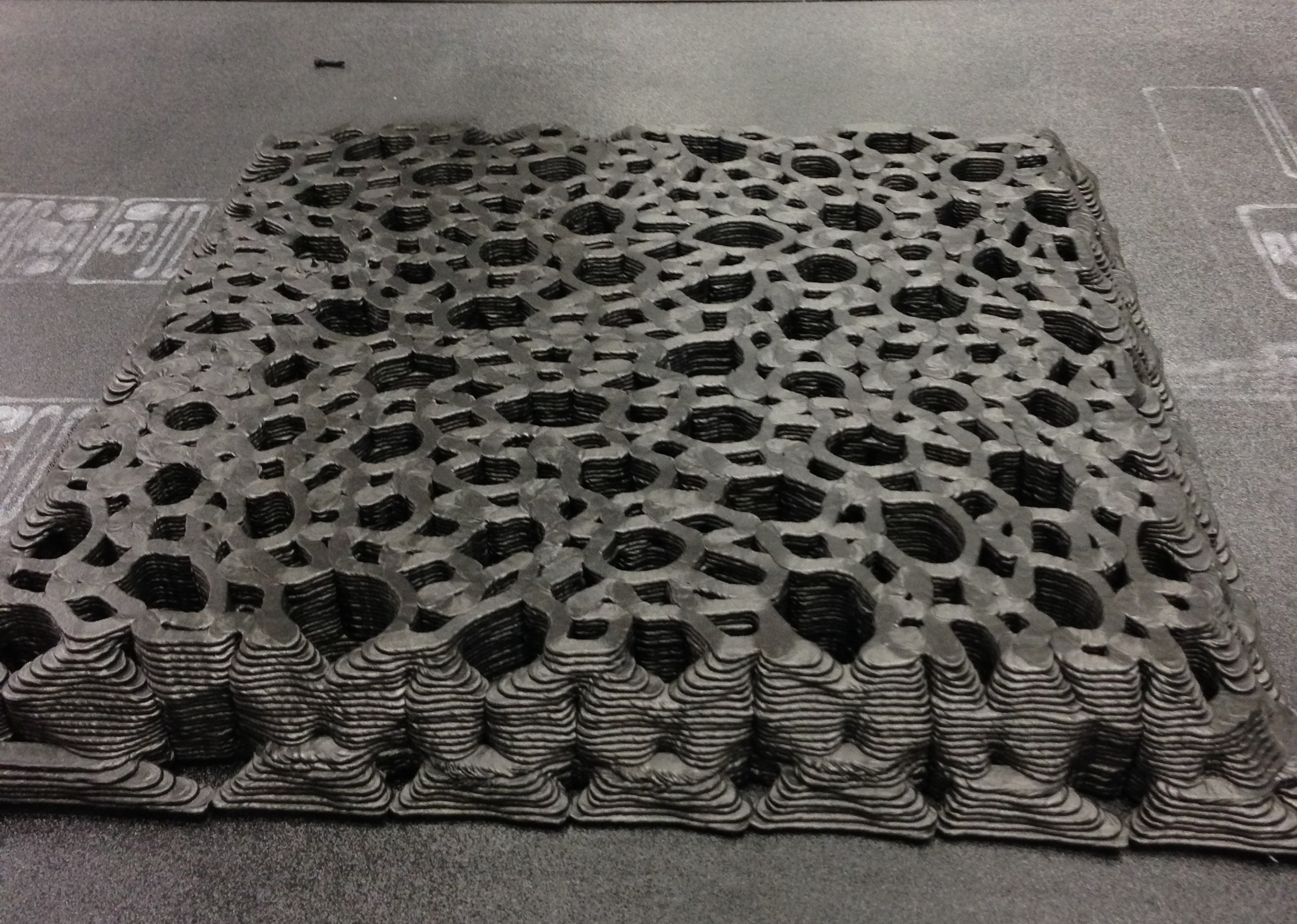}
    \caption{\small \label{printedpyramid}}
  \end{subfigure}
  \caption{ \label{infill_pattern_work_flow}
    Illustration of our framework (see Section \ref{ssec:contrib} for details).
    Subfigure \ref{pyramidplan} shows the plan produced by our framework for printing the infill lattice of a  $610\,$mm $\times$ $610\,$mm $\times$ $610\,$mm pyramid with $143$ layers.
    Subfigure \ref{printedpyramid} shows details of the print after $23$ layers.
    See \cref{sec:implem} for other prints produced by our framework.
  }
\end{figure}

\subsection{Related work}

\paragraph{Euler Graph} \label{sec:previouseulerwork}
The $1$-skeleton of $K$ is an undirected planar graph ($G$).
One approach to make $G$ Eulerian is to delete a minimal number of vertices and/or edges.
But Cai and Yang~\cite{CaYa2011} showed that the Euler vertex deletion and Euler edge deletion problems are NP-Complete.
More importantly, removing edges and/or vertices could create gaps in the coverage of the domain, potentially affecting the mechanical properties of the final product.
Another approach to make $G$ Eulerian is to add a minimum number of edges that pair odd degree vertices in $G$, which can be cast as a \emph{graph augmentation} problem.
Boesch~\cite{TCR1977} presented a polynomial time algorithm for this problem, but planarity of $G$, which is necessary to avoid material collision, is not guaranteed after the augmentation.
Another approach for augmentation is to use the Chinese postman problem to double the edges along shortest paths connecting odd degree vertices in $G$.
But printing the resulting multigraph could be challenging due to non-uniform thickness of (multiple copies of) edges and non-uniform degrees of nodes.
In fact, the degree of some of the nodes could be multiplied by up to $n^o/2$, where $n^o$ is the number of odd degree vertices in $G$.
Visiting a vertex a large number of times could reduce quality of the print.

Jin et al.~\cite{JiHeFuZhDu2017} showed that subpaths in this setting could be generated using a geometric approach and then joined optimally to reduce the amount of jumps.
Zhao et al.~\cite{ZhGuHuGaYoChBeZhCoDaBa2016} proposed a method that finds global continuous paths using Fermat spirals.
But both of these approaches are designed for completely filling the region (with or without holes), and not for sparse infilling.
   
The Catmull--Clark subdivision \cite{CatmullClark1978} creates quadrilateral polygons from any input $2$-complex.
But some vertices in the resulting mesh may have odd degrees.

Our Euler transformation appears similar to the Doo--Sabin subdivision \cite{DoSa1978}, which creates new vertices for each polygon after each iteration.
But  Doo--Sabin subdivision does not preserve the underlying space, and creates boundary vertices with odd degree.
It will add $O(b)$ jumps, where $b$ is the number of boundary vertices after subdivision.
Also, new vertices are created in the Doo--Sabin subdivision based on vertices and the number of edges in the original polygon.
But in our Euler transformation, new vertices are created through mitered offsets (parallel offset of neighboring edges) of polygons.
The length of edges get scaled roughly by a factor of $0.25$ by Doo--Sabin subdivision when the number of vertices is large.
In contrast, one gets much greater control on edge lengths in our Euler transformation by choosing mitered offsets suitably, independent of the number of vertices. 
Further, original angles are not preserved in the Doo--Sabin subdivision. 
Finally, combinatorial changes in general and topological changes for concave polygons are not allowed the the Doo--Sabin subdivision.
But our Euler transformation allows combinatorial as well as topological changes, thus avoiding small edges after transformation (see Section \ref{sec:eulertransformwithcombin}).

\paragraph{Tool path}\label{sssec:toolpath}

While the perimeter and inset can typically be printed as continuous loops, tool paths for the infill lattice commonly resemble a back and forth pattern, a spiral, or a fractal-like path for an arbitrary geometry.
The tool path consists of \textit{print paths} (material is pushed through the nozzle) and \textit{travel paths} (extruder moves from one location to other without pushing material).
Galceran and Carreras \cite{GaCa2013} described coverage path planning as assignment of finding a path that passes through all the points without obstacles.
Xu \cite{Xu2011} presented the use of graph-based algorithms in coverage problems.
General requirements for graph-based coverage problems such as all vertex coverage, non-overlapping and continuous paths, etc.~\cite{CaHuHa1988} are applicable in 3D printing as well, including the requirement that each edge should be printed.
One of the major steps in path planning is the identification of the tool path trajectory~\cite{DiPaCuLi2015}.
This tool path generation step involves filling interior regions and joining sub paths \cite{JiHeDu2017}.
While attempts have been made to join sub paths into a continuous path, they are all limited by increasing complexity of geometry.
Fewer sub paths in the tool path trajectory implies better quality of print.

Wang et al.~\cite{KuWuWa2019} developed $3$-dimensional infill (crossfill) using space filling curves, whose layer by layer cross-section is a continuous curve.
Crossfill curves are fit into the infill region by intersecting with boundary of the polygon in a given layer, but this step can create multiple components.
Later these components are connected into one continuous curve through an outer perimeter.
New edges added to create these components are not guaranteed to have a support below it.
We can still have material collision in each individual component if their boundary is too skewed or component is too thin.
 
Use of graphical models in additive manufacturing was demonstrated by Dreifus et al.~\cite{Dretal2017}.
They mesh each layer of the print as a graph, and find an Eulerian cycle over all edges of the graph.
If the infill lattice is not Eulerian, they add "phantom edges" to the odd-degree vertices of the graph.
When the extruder reaches an odd degree vertex, it stops printing, lifts above the printed material, moves to its matched vertex, and resumes printing.
However, these stops and starts leave \textit{teardrops} of material in their wake, as the extruder drags excess material behind it.
Such teardrops weaken the print.
Also, stopping and starting repeatedly increases print times.
Further, their approach to identify the Eulerian cycle created \emph{crossovers} when pairs of sub-paths of the tour cross each other.

\section{Euler Transformation} \label{sec:eulertransformation}

We recently introduced the Euler transformation of polyhedral complexes in a general setting, with details provided for 2D and 3D cases \cite{GuKr2018}.
Our framework for continuous tool path planning depends crucially on the Euler transformation in 2D, and uses extensively the related notation and definitions.
We present the main results for the 2D case here.

\subsection{Definitions on Polygonal Complexes} \label{sec:defns}

We present definitions that we use to specify properties of the input mesh $K$ as well as the Euler transformed mesh $\hK$ in 2D.
See standard books on algebraic topology \cite{Hatcher2002,Munkres1984} for details.

\begin{defn} \label{def:plyhdcplx}
	\emph{\bfseries (Polygonal complex)}
	A polygonal complex $K$ is a collection of \emph{polygons} in $\R^2$ such that every face of a polygon in $K$ is also included in $K$, and the nonempty intersection of any two polygons in $K$ is a face of both.
	Polygons in $K$ are referred to also as its \emph{$2$-cells}.
        We refer to $K$ as a $2$-complex.
\end{defn}

We will work with \emph{finite} polygonal complexes, i.e., where the set of cells in $K$ is finite.
The cells of interest in this work are vertices, edges, and polygons.
Our definition of Euler transformation (in \cref{sec:eulertsfm}) as well as geometric realization results (in \cref{sec:geomrlzn}) do not require polygons in $K$ to be convex.
In general, some cells in the Euler transformed complex $\hK$ may not be convex.
But if we assume cells in $K$ are convex, then we can guarantee a large majority of cells in $\hK$ are so as well.

\begin{defn} \label{def:purecplx}
	\emph{\bfseries (Pure  complex)}
	A polygonal complex is \emph{pure} if every vertex and every edge is a face of some polygon in $K$.
\end{defn}
\noindent Thus in a pure $2$-complex, there are no ``isolated'' edges or vertices, and all top-dimensional cells are polygons.

We assume the input mesh $K$ is a finite, connected, pure $2$-complex in $\R^2$.
Along with $K$, we assume we are given a collection $\CH$ of polygons that capture $2$-dimensional ``holes'', and a singleton set $\CO$ that consists of a polygon capturing the ``outside''.
Note that vertices and edges in the intersection of a polygon in $K$ and a polygon in $\CH$ or $\CO$ are precisely the boundary cells of $K$.
We make the following assumptions about intersections of polygons in $K$, $\CH$, and $\CO$.
We denote the \emph{underlying spaces} of these objects as $|K|, |\CH|$, and $|\CO|$, respectively.
To be precise, $|\CH| = \cup_{c_i \in \CH} |c_i|$.
\begin{asmn} \label{asmn:Kholesoutside}
The following conditions hold for the input complex $K$, the collection of holes $\CH$, and the outside polygon $\CO$.
\begin{enumerate}
  \item $|K| \cup |\CH| \cup |\CO| = \R^2$.
  \item \label{asmn:sepholes} Polygons in $\CH$ are pairwise disjoint, and are also disjoint from the polygon that is $\CO$. 
  \item \label{asmn:holeintr} Any polygon in $K$ and a polygon in $\CH$ intersect in at most \emph{one} facet (edge) of both.
  \item  \label{asmn:outintr} No two edges that are \emph{adjacent facets} of a polygon in $K$  
  intersect the polygon that is $\CO$.
\end{enumerate}
\end{asmn}
\noindent Hence polygons in $K, \CH$, $\CO$ cover $\R^2$, and each polygon in $\CH$ captures a separate hole that is also separate from the outside.

We point out that \emph{articulation} (or cut) vertices are allowed in $K$, i.e., vertices whose removal disconnects the complex (we assume $K$ is connected to start with).
Conditions specified in \cref{asmn:Kholesoutside} ensure such vertices are boundary vertices of $K$.
For instance, $K$ could consist of two copies of the complex in \cref{infill_pattern} that meet at one of the four corner points.

\subsection{Definition of Euler Transformation} \label{sec:eulertsfm}

We define the Euler transformation $\hK$ of the input $2$-complex $K$ by listing the polygons included in $\hK$.
We denote vertices as $v$ (or $u, v_i$), edges as $e$ (or $e_i$), and polygons as $f$ (or $f_i$).
The corresponding cells in $\hK$ are denoted $\hv, \he, \hf$, and so on.
We first define cells in $\hK$ \emph{abstractly}, and discuss aspects of geometric realization in \cref{sec:geomrlzn}.

We start by \emph{duplicating} every polygon in $K \cup \CH \cup \CO$.
Since we do not want to alter the domain in $\R^2$ captured by $K$, we set $\hCH = \CH$ and $\hCO = \CO$.
But we ``shrink'' each polygon in $K$ when duplicating (\cref{sec:geomrlzn}).
By \cref{asmn:Kholesoutside}, this duplication represents each edge $e \in K$  by two copies in $\hK$.

The polygons in $\hK$ belong to three classes, and correspond to the polygons, edges, and vertices in $K$ (see Figure \ref{fig:3types2cells2d}).
\begin{enumerate}
        \item \label{2dETcls1}
	For each polygon $f \in K$, we include $\hf \in \hK$ as its copy.
	
	\item \label{2dETcls2}
	Each edge $e \in K$ generates the $4$-gon $\hf_e$ in $\hK$ specified as follows.
	Let $e = \{u,v\} \in f,f'$, where $f \in K$ and $f' \in K \cup \CH \cup \CO$.
	Then $\hf_e$ is the polygon whose facets are the four edges $\{\hu,\hv\}, \, \{\hv,\hv'\}, \, \{\hu',\hv'\}$, and $\{\hu,\hu'\}$.
	Here, $\hv, \hv'$ are the two copies of $v$ in $\hK$.
	Edges $\he = \{\hu,\hv\}$ and $\he' =  \{\hu',\hv'\}$ are facets of the Class \ref{2dETcls1} polygons $\hf$ added to $\hK$ or of polygons $\hf'$ in $\hCH$ or $\hCO$.
	Edges $\{\hu,\hu'\}$ and $\{\hv,\hv'\}$ are added new.
	
	\item \label{2dETcls3}
	Each vertex $v \in K$ that is part of $p$ polygons in $K$ generates a $p$-gon (polygon with $p$ sides) $\hf_v$ in $\hK$ whose vertices and edges are specified as follows.
	Let $v \in f_k$ for $k=1,\dots,p$ in $K$.
	Then $\hf_v$ has vertices $\hv_k$, $k=1,\dots,p$, where $\hv_k$ is the copy of $v$ in $\hf_k$ (in $\hK$).
	For every pair of polygons $f_i,f_j \in \{f_k\}_1^p$ that intersect in an edge $e_{ij} \in K$, the edge $\he_{ij} = \{\hv_i, \hv_j\}$ is included as a facet of $\hf_v$.
	Edges $\he_{ij}$ are ones added new as facets of Class \ref{2dETcls2} polygons (see above).
\end{enumerate}

\vspace*{0.25in}
\begin{figure*}[hbp!]
	\centering
	\includegraphics[scale=0.24]{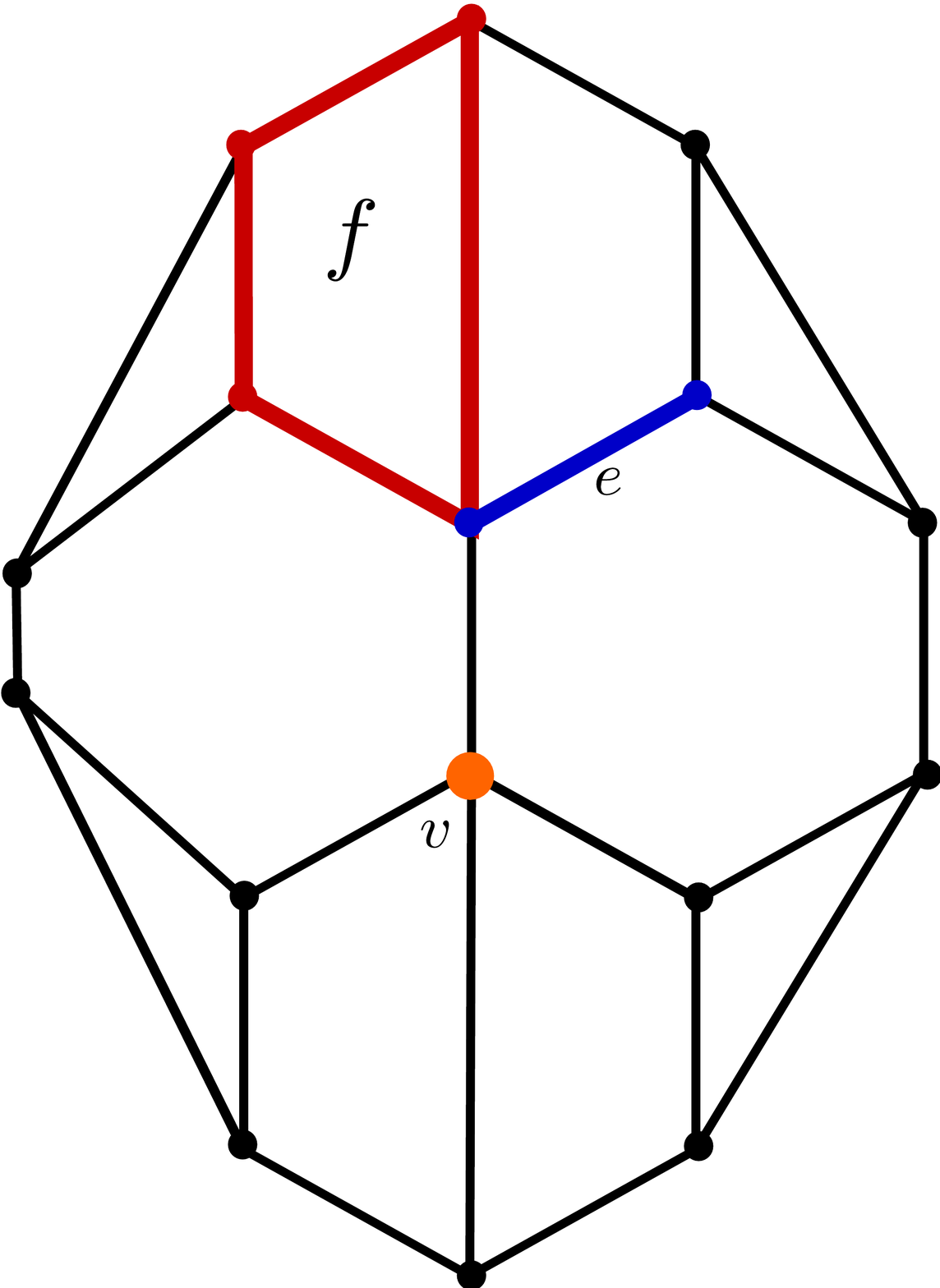}
	\hspace*{0.1in}
	\includegraphics[scale=0.24]{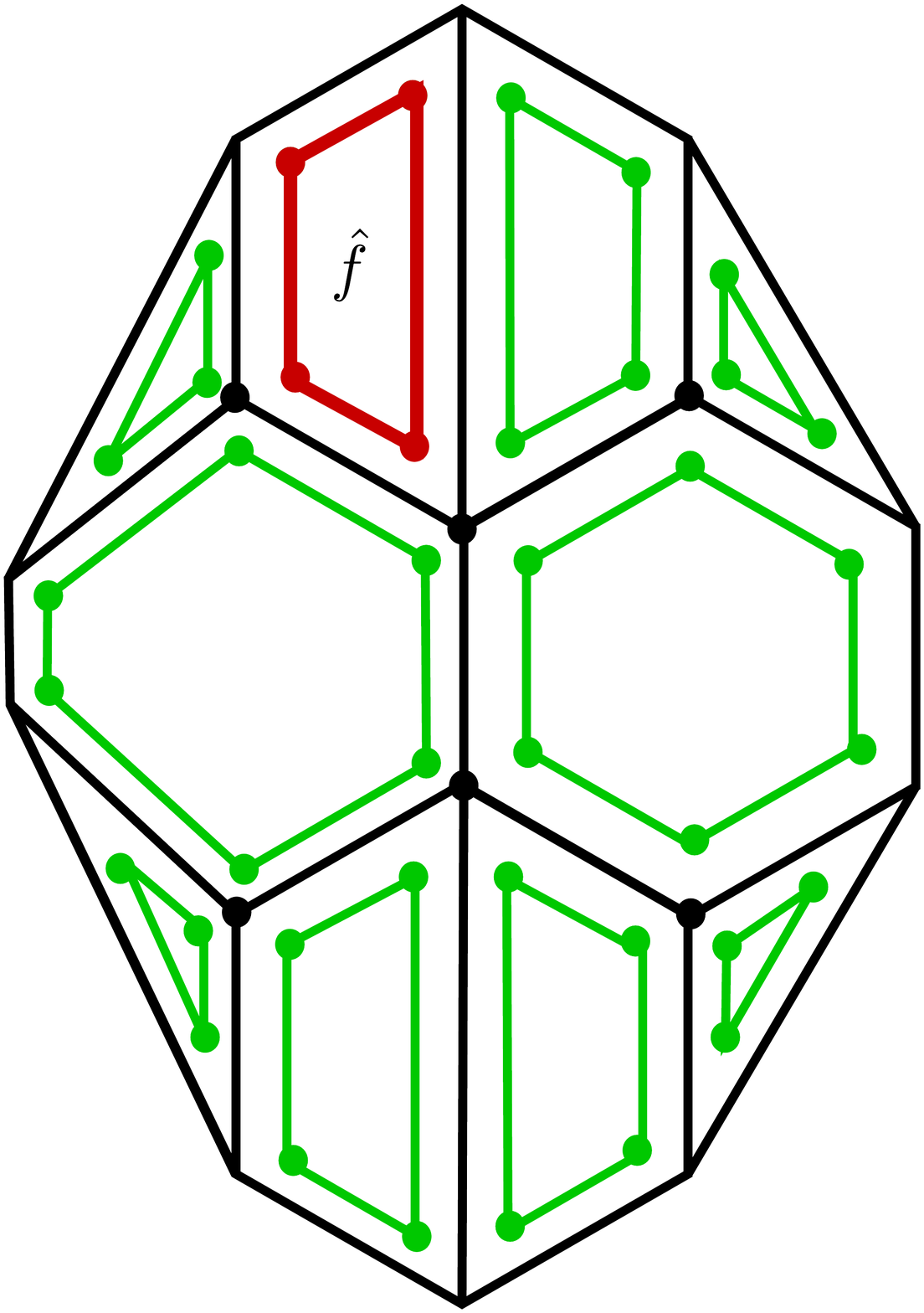}
	\hspace*{0.1in}
	\includegraphics[scale=0.24]{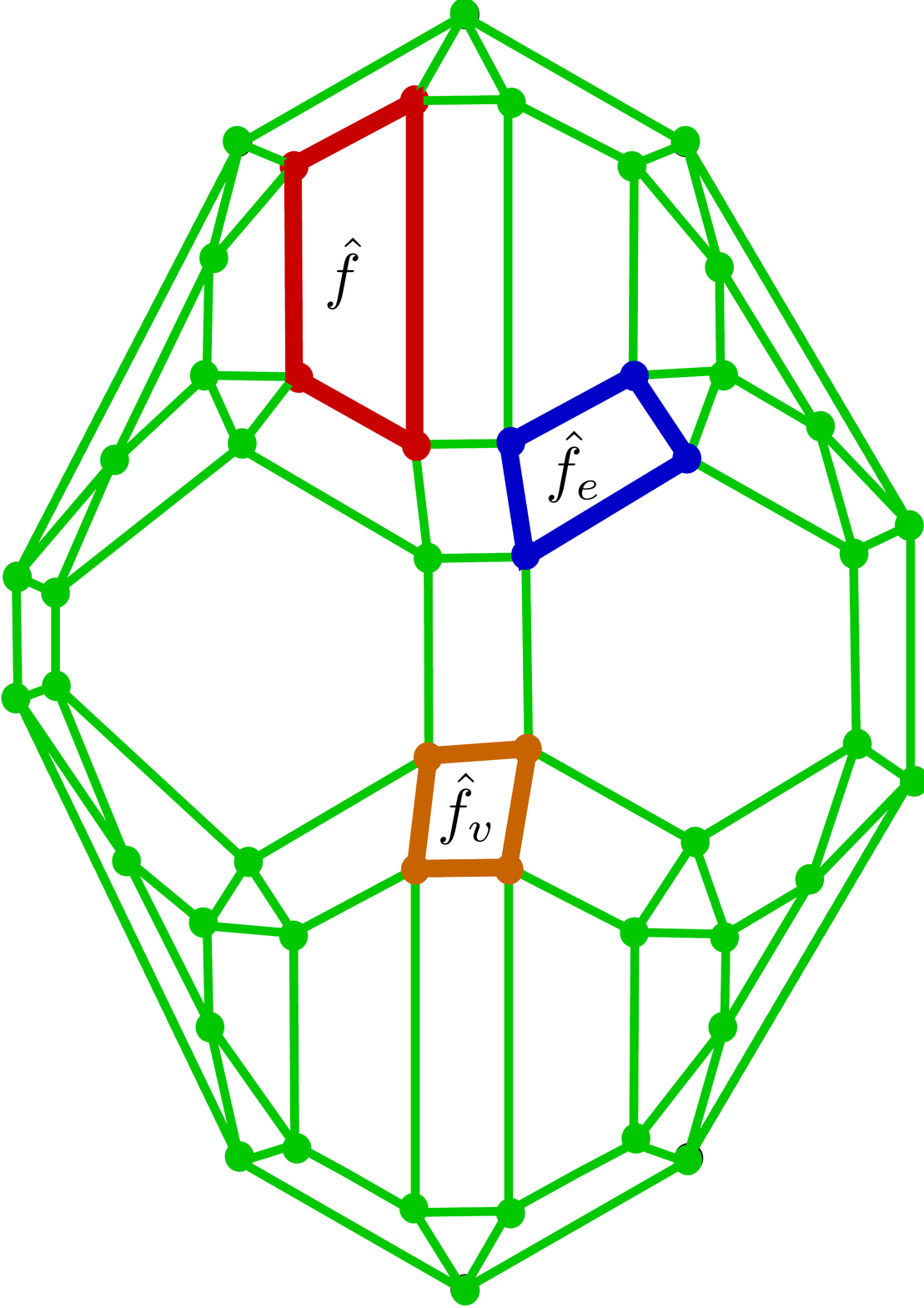}
	\caption{A polygon $f$, edge $e$, and a vertex $v$ highlighted in an input complex $K$ (left),
		an intermediate complex showing only the copies of original polygons in $K$ that are included in $\hK$, i.e., of Class \ref{2dETcls1} (middle),
		and the final Euler transformation $\hK$ (right).}
	\label{fig:3types2cells2d}
\end{figure*}

\section{Geometric Properties of Euler Transformed Complex} \label{sec:geomrlzn}

As the Euler transformation adds new polygons corresponding to input polygons, edges, as well as vertices, we \emph{offset} the Class \ref{2dETcls1} polygons in $\hK$ in order to generate enough space to add extra polygons.
Intuitively, we ``shrink'' each of the polygons in $K$.
We use standard techniques for producing offset polygons in 2D, e.g., \emph{mitered} offset generated using the straight skeleton (SK) of the input polygon \cite{AiAuAlGa1995}.
We define the polygon offset as a mitered offset of the input polygon that creates no \emph{combinatorial} or \emph{topological changes}---i.e., no edges are shrunk to points, and the polygon is not split into multiple polygons.
Later on, we generalize the definition of Euler transformation to allow combinatorial or topological changes (\cref{sec:eulertransformwithcombin,sec:eulertransformwithTopolog}).
Naturally, we do not want to alter the print domain $|K|$.
Hence we include the polygons in $\CH$ and $\CO$ in $\hK$ without any changes.

\subsection{Geometric Realization} \label{ssec:rlzn2d}

\begin{thm} \label{thm:deg4ET2d}
  Every vertex in the $1$-skeleton of $\hK$ has degree $4$.
\end{thm}
\begin{proof}
  Consider a vertex $v$ shared by adjacent edges $e_1, e_2 \in f$, where $f \in K \cup \CH \cup \CO$ is a polygon.
  Following \cref{asmn:Kholesoutside}, the edges $e_1$ and $e_2$ are shared by exactly two polygons each from the input complex,  holes, or the outside cell.
  Let $f'_1, f'_2$ be the other polygons containing edges $e_1, e_2$, respectively (with $f$ being the first polygon).

  Consider the vertex $\hv \in \hK$ generated as part of $\hf$.
  The polygon $\hf$ is a mitered offset of $f$ when $f \in K$, or is identical to $f$ when it belongs to $\CH \cup \CO$.
  Hence $\hf$ is a simple polygon in both cases, and $\hv$ is part of two edges $\he_1, \he_2 \in \hf$.
  Further, $\hv$ will be part of two more edges $\{\hv,\hv'_1\}$ and  $\{\hv,\hv'_2\}$  added as facets of the Class \ref{2dETcls2} polygons generated by $e_1, e_2$.
  Here $\hv'_i \in \he'_i \in \hf'_i$ for $i=1,2$.
  Hence $\hv$ has degree $4$ in the $1$-skeleton of $\hK$.
\end{proof}

\begin{rem}
	\label{rem:adjbdyedges}
	{\rm 
		We show why we require the input complex to satisfy Conditions \ref{asmn:holeintr} and \ref{asmn:outintr} in \cref{asmn:Kholesoutside}, which require that no two adjacent edges of a polygon in $K$ can be boundary edges.
		Consider the input complex $K$ consisting of a single square, whose four edges are shared with the outside cell $\CO$.
		Then every vertex in the Euler transformation  $\tilde{K}$ will have the odd degree of $3$ (see \cref{fig:outintr}).
		But if we apply the Euler transformation \emph{once more} to $\tilde{K}$, we do get a valid complex $\hK$ with each vertex having degree $4$.
		Note that $\tilde{K}$ satisfies Condition \ref{asmn:outintr}, and hence becomes a valid input.
		\begin{figure}[htp!]
			\centering
			\includegraphics[scale=0.22]{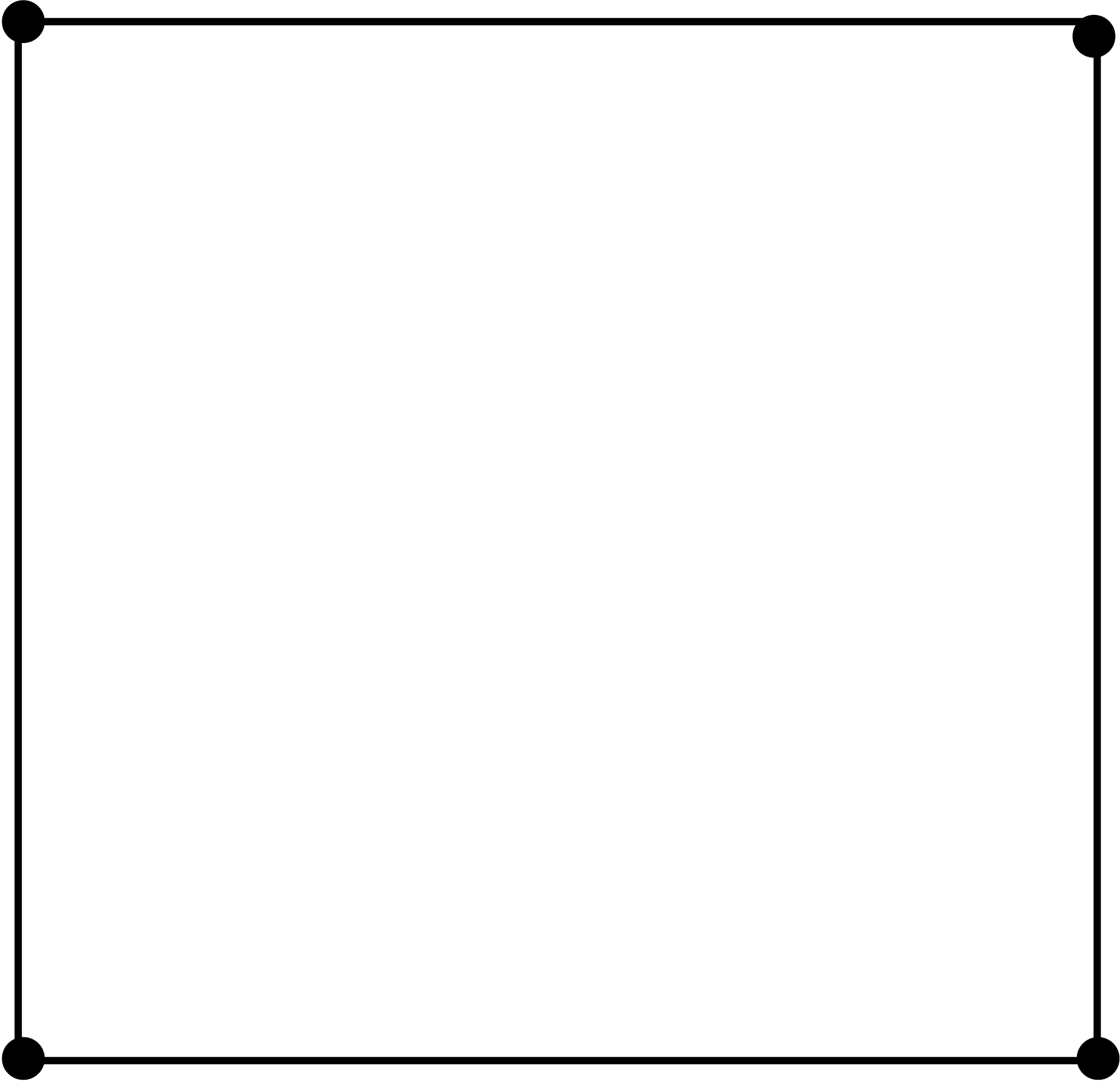}
			\quad
			\includegraphics[scale=0.22]{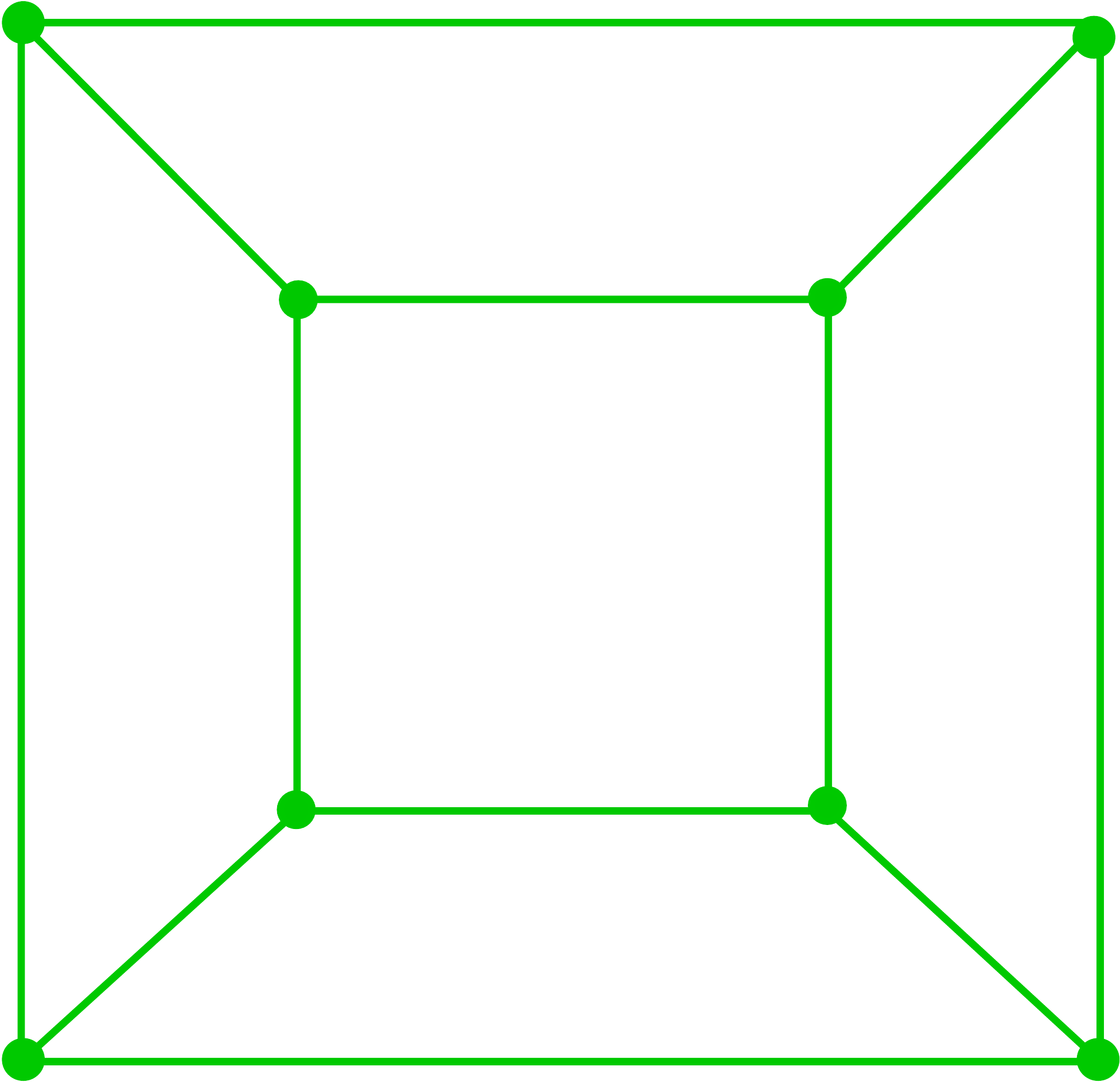}
			\quad
			\includegraphics[scale=0.22]{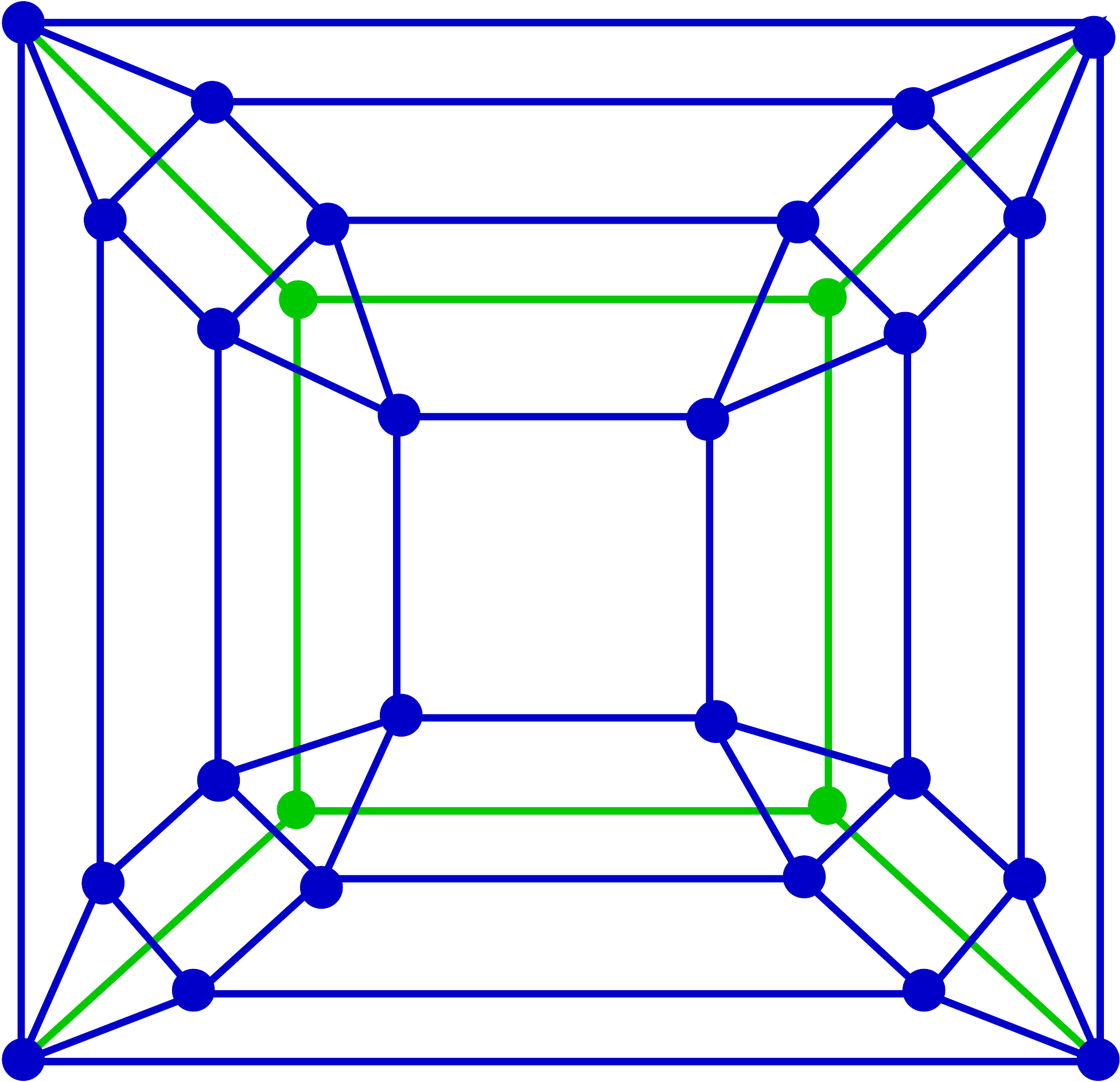}
			\caption{Applying the Euler transformation to the square $K$ whose more than one adjacent edge is shared with the outside (left) produces a complex $\tilde{K}$ in which every vertex has odd degree (middle).
				Applying the Euler transformation again to $\tilde{K}$ produces a valid complex $\hK$ where every vertex has degree $4$ (right).}
			\label{fig:outintr}
		\end{figure}
	}
\end{rem}

\begin{lem}
	\label{lem:cntshVhEhF2d}
	Let $V, E, F$ denote the sets of vertices, edges, and polygons (faces) in $K$, and let $\hV, \hE, \hF$ denote the corresponding sets in $\hK$.
	The following relations hold for the cardinalities of these sets: $|\hV| = 2|E|,~|\hE| = 4|E|$, and $|\hF| = |V| + |E| + |F|$.
\end{lem}

\begin{proof}
	Let $\delta(v)$ denote the degree of vertex $v \in K$, and let $\hf_v$ be the polygon generated by $v$ in $\hK$.
	This is a Class \ref{2dETcls3} polygon.
	Following \cref{asmn:Kholesoutside} about $K, \CH, \CO$, it is clear that when $v$ belongs to $p$ polygons in $K$, we must have $\delta(v)=p$ and $\hf_v$ has $p$ vertices.
	Since each cell $\hf$ corresponding to polygon $f \in K$ is a mitered offset, and since each vertex $\hv$ is part of one such offset polygon, it follows that $\hf_u \cap \hf_v = \emptyset$ for any two vertices $u,v \in K$.
	Hence we get
	\[
	|\hV| = \sum_{v \in K} \delta(v) = 2|E|.
	\]
	By \cref{thm:deg4ET2d}, each vertex $\hv \in \hK$ has degree $\hat{\delta}(\hv) = 4$ in $\hK$.
	Combined with the result above on $|\hV|$, we get that
	\[
	|\hE| = \frac{1}{2} \sum_{\hv \in \hK} \hat{\delta}(\hv) = \frac{1}{2} \cdot 2|E| \cdot 4 = 4|E|.
	\]
	%
        Further, each polygon, edge, and vertex in $K$ generate corresponding unique polygons in $\hK$ belonging to three classes.
	Hence we get $|\hF| = |F| + |E| + |V|$.
\end{proof}

\begin{rem} \label{rem:prtlgth}
  {\rm 
    While the number of edges in $\hK$ is quadrupled, the total length of all edges gets roughly doubled (see our previous manuscript \cite{GuKr2018} for details). 
    If we want to limit the total print length, we could start with a much sparser input complex $K$, and choose the mitered offsets of its polygons such that $|K|$ is covered adequately while limiting the total length of edges in $\hE$.
  }
\end{rem}

\begin{lem}
	\label{lem:hGplanar}
	The graph $\hG$, the $1$-skeleton of $\hK$, is planar.
\end{lem}
\begin{proof}
	By the definition of Euler transformation, each polygon $\hf \in \hK$ that is a mitered offset of  polygon $f \in K$ is a simple closed polygon.
        Any two polygons $\hf, \hf' \in \hK$ of Class \ref{2dETcls1} generated by polygons $f,f' \in K$ satisfy $\hf \cap \hf' = \emptyset$.
	
	Class \ref{2dETcls2} polygons $\hf_e, \hf_{e'} \in \hK$ generated by edges $e,e' \in K$ intersect at a vertex $\hv$ if and only if $e$ and $e'$ are adjacent edges of a polygon $f \in K$ meeting at the vertex $v$.
	Since each $\hf \in \hK$ is a mitered offset of some polygon $f \in K$, at least one of the two copies $\he, \he'$ of edges in $\hK$ corresponding to the edge $e \in K$ is shorter in length than $e$.
	If $e$ is not a boundary edge then both $\he$ and $\he'$ are shorter than $e$.
	If $e$ is a boundary edge then one edge out of $\he, \he'$ has the same length as $e$ while the other is shorter.
	Hence each polygon $\hf_e$ of Class \ref{2dETcls2} is a (convex) trapezium.
	
	Since all edges of the polygon $\hf_v$ of Class \ref{2dETcls3} generated by vertex $v \in K$ are precisely the \emph{new} edges added to define the Class \ref{2dETcls2} polygons, each $\hf_v$ is a simple closed polygon.
	Further, by the properties of Class \ref{2dETcls2} polygons specified above, $\hf_v \cap \hf_{v'} = \emptyset$ for any two vertices $v, v' \in K$.
	
	Thus every polygon in $\hK$ is simple and closed.
	Any two such polygons intersect at most in an edge or a vertex, and any two edges in $\hK$ intersect at most in a vertex.
	Hence $\hG$, the $1$-skeleton of $\hK$, is a planar graph.
\end{proof}

\begin{rem}
	\label{rem:disjholes}
	{\rm
		We illustrate why we require holes in the domain to be disjoint (Condition \ref{asmn:sepholes} in \cref{asmn:Kholesoutside}).
		Consider the input complex $K$ with two holes $h, \bar{h} \in \CH$ that intersect at a vertex $v$.
		The corresponding vertex $\hv$ in the transformed complex $\hK$ will not have a degree of $4$.
		There will also be other vertices in $\hK$ that have odd degree, which are circled in \cref{fig:disjholes}.
		Let these odd-degree vertices be labeled $\hv',\hv''$.
		Technically, there are two \emph{identical} copies of the edge $\{\hv,\hv'\}$ and similarly of $\{\hv,\hv''\}$.
		But such duplicate edges make the graph $\hG$ ($1$-skeleton of $\hK$) non-planar.
		If we include only one copy of each pair of duplicate edges, we get odd degree vertices in $\hG$.
                \begin{figure}[htp!]
			\centering
			\includegraphics[scale=0.35]{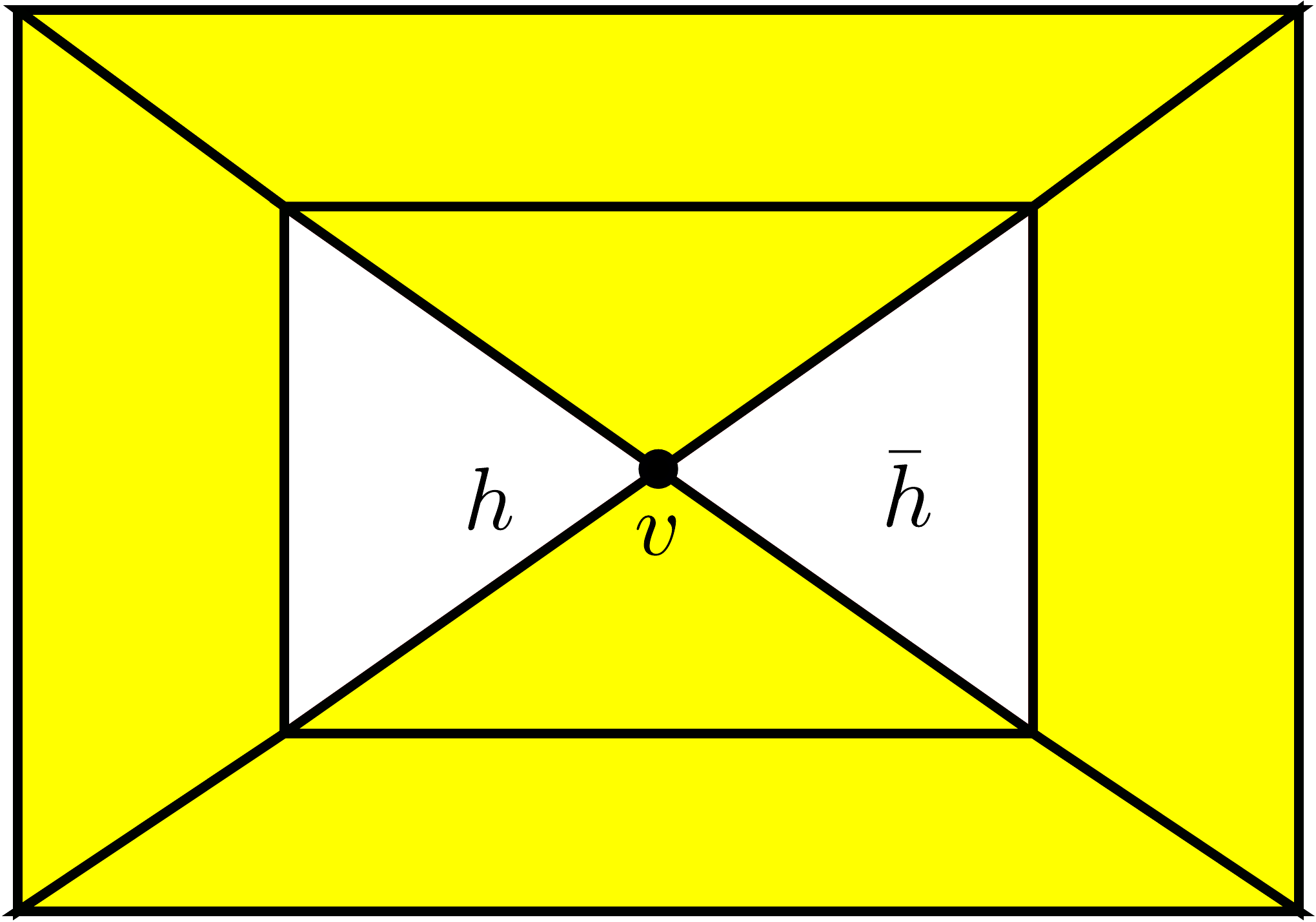}
			\quad\quad
			\includegraphics[scale=0.35]{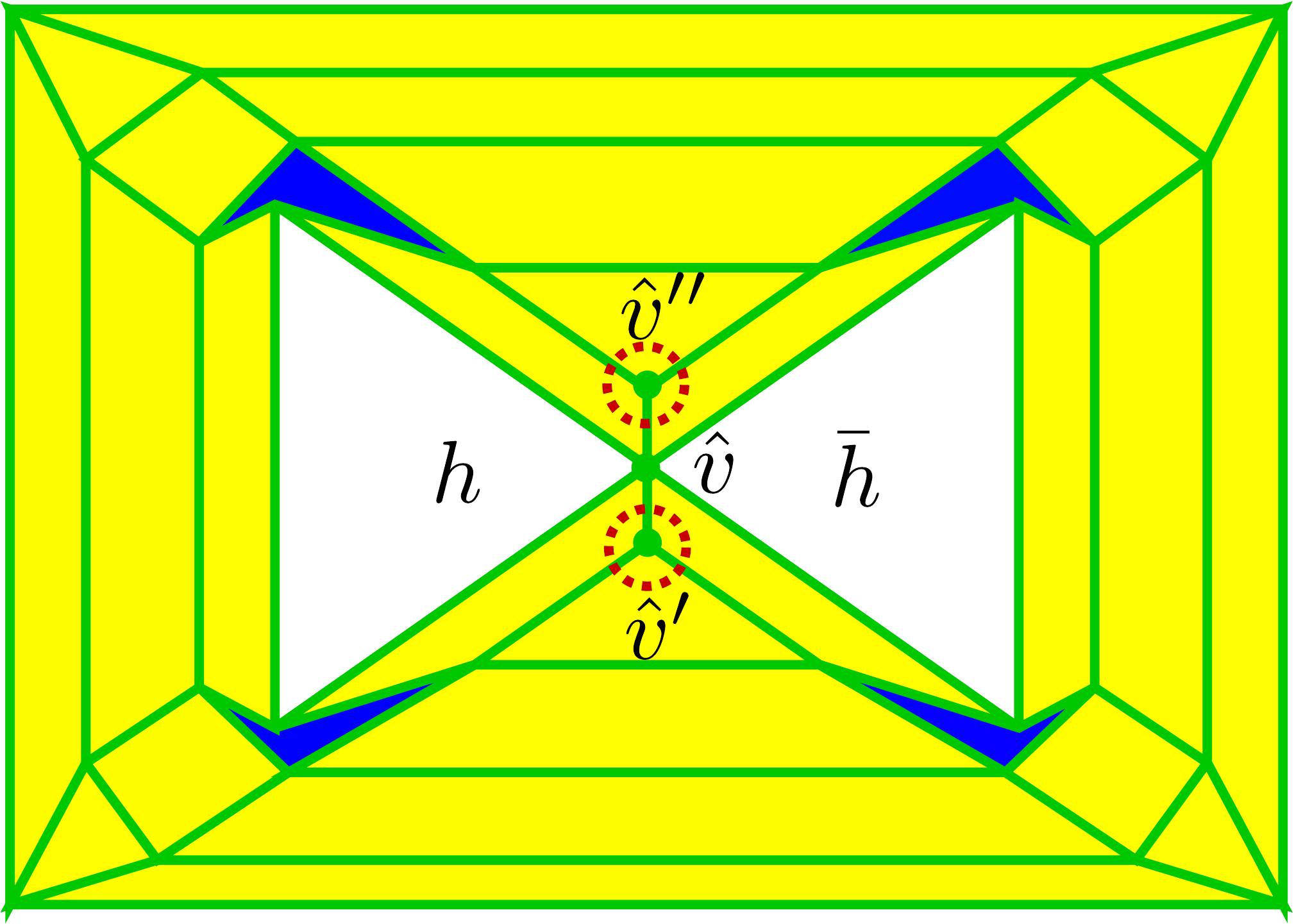}
			\caption{Two holes touching at a vertex (left), and the result of applying Euler transformation (right).
				Vertices with odd degree in the result are circled.
				The four polygons shaded in blue are of Class \ref{2dETcls3} generated by vertices in the input complex (see \cref{sec:eulertransformation}).
				These cells could be nonconvex.
			}
			\label{fig:disjholes}
		\end{figure}
		
	}
\end{rem}

We pointed out in the Proof of \cref{lem:hGplanar} that the polygons of Class \ref{2dETcls2} in $\hK$ generated by edges are convex $4$-gons.
Each polygon $\hf \in \hK$  of Class \ref{2dETcls1} is geometrically similar to the polygon $f \in K$ generating it.
Hence if $f$ is convex, so is $\hf$.
But polygons of Class \ref{2dETcls3} generated by vertices are not guaranteed to be convex.
In fact, when $v \in K$ is a boundary vertex where $K$ has a notch, or an ``incut corner'', $\hf_v \in \hK$ could be nonconvex---see \cref{fig:disjholes} for illustrations.
We finish with the result on $\hK$ remaining connected.

\begin{prop}
	\label{prop:hKcnctd2d}
	If $K$ is connected, then so is $\hK$.
\end{prop}

\begin{proof}
  We noted in the proof of \cref{lem:hGplanar} that the mitered offset polygons in $\hK$ are pairwise disjoint.
  But we show that when polygons $f,f' \in K$ are connected, so are the corresponding offset polygons $\hf, \hf' \in \hK$.
  By \cref{asmn:Kholesoutside} on the input complex, when polygons $f,f' \in K$ intersect, they do so either in an edge $e$ or in a vertex $v$.
  If $f \cap f' = e$, then by the definition of Euler transformation (\cref{sec:eulertransformation}), the corresponding offset polygons $\hf,\hf' \in \hK$ are connected by the pair of new edges defining $\hf_e$, the $4$-gon of Class \ref{2dETcls2} generated by edge $e$.
  If $f \cap f' = v$ and $v$ is not an articulation vertex, then the corresponding offset polygons $\hf,\hf' \in \hK$ are similarly connected by the Class \ref{2dETcls3} polygon $\hf_v$ generated by $v$, with the corresponding copies $\hv, \hv'$ of $v$ in $\hf,\hf'$, respectively, being vertices of $\hf_v$.
  If $f \cap f' = v$ that is an articulation vertex, then $\hv=v$ is the identical copy of this vertex in $\hK$.
  There will be two Class \ref{2dETcls3} polygons $\hf_v, \hf_v'$ generated by $v$ in the two biconnected components joined at $v$, with $\hf_v \cap \hf_v' = \hv$.
  Further, $\hf_v$ is connected to $\hf$ and $\hf_v'$ to $\hf'$, ensuring that $\hf$ and $\hf'$ are connected.
  It follows that $\hK$ is connected when the input complex $K$ is connected.
\end{proof}

\section{Generalized Euler Transformation} \label{sssec:genET}

We now consider generalizations of the Euler transformation where we could relax parts of \cref{asmn:Kholesoutside}.
The goal is to allow combinatorial and topological changes in the polygons undergoing transformation.

Consider a $2$-complex $K$ consisting of a single polygon $f$.
$K$ does not satisfy the input condition for Euler transformation, since adjacent edges are shared with the outside (Figure \ref{fig:2cellpartition}).
Nevertheless, we apply the transformation to $K$.
In the resulting $\hK$, all vertices will have odd degree.
But this $\hK$ satisfies the input condition.
Hence if we apply the Euler transformation again to $\hK$, i.e., we apply it \emph{twice} on $K$, the resulting complex has a $1$-skeleton that is Euler in the default setting.
We define this process as the generalized Euler transformation.

\begin{figure}[htp!] 
  \centering
  \includegraphics[scale=0.33]{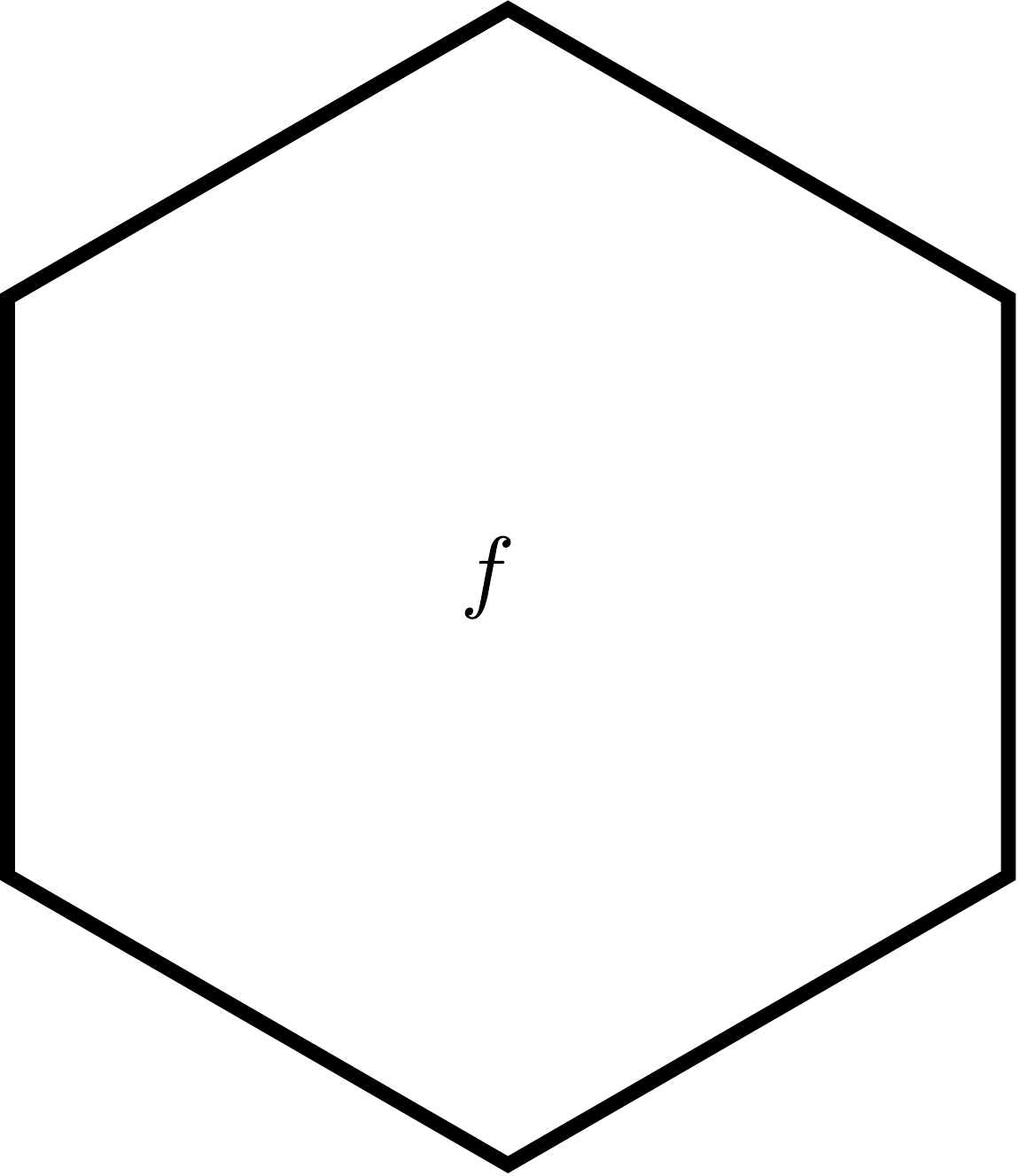}
  \quad\quad
  \includegraphics[scale=0.33]{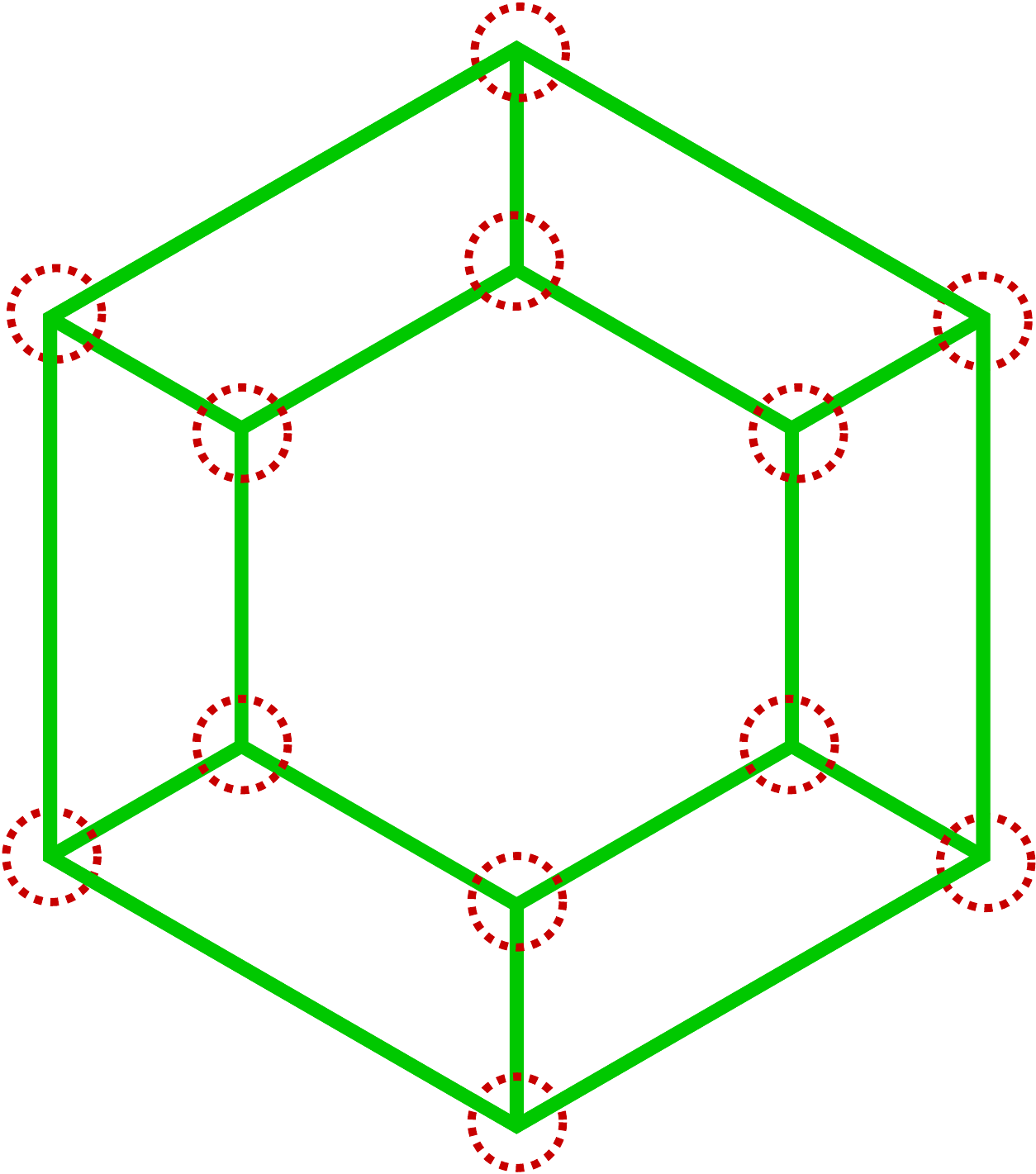}
  \quad\quad
  \includegraphics[scale=0.33]{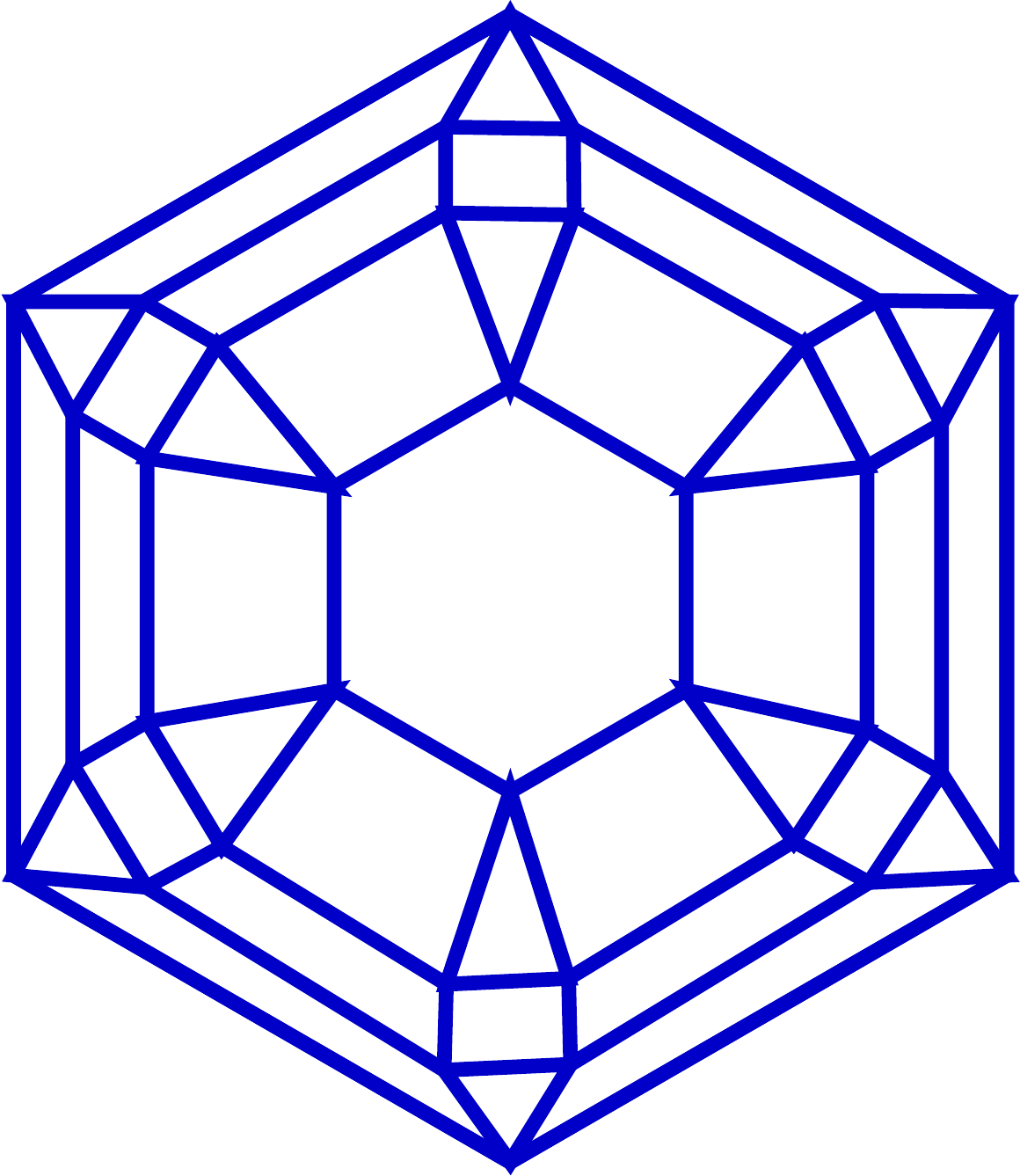}
  \caption{\label{fig:2cellpartition}
    $\hK$ consisting of a single polygon $f$ (left) and its Euler transformation $\hK$ in green (middle).
    Vertices circled in red have odd degrees.
    The complex in blue (right) is the Euler transformation of $\hK$, and its $1$-skeleton is Euler.}
\end{figure}

\begin{defn}
	\label{def:genET}
	(Generalized Euler Transformation in $d=2$)
	Let $K$ be a $2$-dimensional cell complex in $\R^2$ with polygons possibly having adjacent boundary edges.
	Apply the Euler transformation on $K$ to obtain $\hK$, which will always satisfy \cref{asmn:Kholesoutside}.
        Now apply Euler transformation on $\hK$.
\end{defn}
We could use the generalized Euler transformation to improve mechanical properties of the design (by increasing the density of the mesh) in some regions while still guaranteeing that the $1$-skeleton of the $2$-complex is Euler.
Note that the density of the mesh could increase significantly by this process.
\begin{lem}
  \label{lem:trgrrate} 
  After $m$ transformations of complex $K$, the number of vertices $|\hat{V}^m| = 2\cdot 4^{m-1}|E|$ and number of edges $|\hat{E}^m|=4^m|E|$.	
\end{lem}
\begin{proof}
  After the first transformation, we get $|\hat{E}^{1}| = 4 |E|$ (by \cref{lem:cntshVhEhF2d}).
  Extending the argument, after $m$ transformations we get $|\hat{E}^m| = 4^m |E|$.
  Similarly, we get $|\hat{V}^1| = 2 |E|$, then $|\hat{V}^m| = 2 |\hat{E}^{m-1}| = 2\cdot4^{m-1}|E|$.     	
\end{proof}

We present Euler transformation with \emph{combinatorial} and \emph{topological} changes to a polygon resulting from mitered offset. 

\subsection{Euler Transformation with Combinatorial changes}\label{sec:eulertransformwithcombin}

Suppose $\hf, \hf_e, \hf_v$ are $3$ classes of polygons corresponding to a polygon ($f$), edge ($e$), and a vertex($v$) in $K$.
Maximum mitered offset in a polygon of the input complex $K$ is limited by the smallest edge length, as the mitered offset in our original Euler transformation assumes no combinatorial and topological changes in $\hf$.
\textit{Combinatorial change in $\hK$ is a change in number of edges of some Class \ref{2dETcls1} polygon $\hf$ in $\hK$ from corresponding polygon in $f$ in $K$.} 
Suppose polygon $f$ has $n$ edges and is now permitted to have combinatorial changes when generating its mitered offset.
Then we can reduce at most $n-3$ edges as we want $\hf$ to still be a polygon, and a polygon has at least $3$ edges.
If $\hf_e$ is sharing edges with two Class \ref{2dETcls1} polygons and if both of those edges are reduced, then $\hf_e$ will collapse into an edge (see \cref{fig:eulertransfecollapse}).
Since Class \ref{2dETcls3} polygons ($\hf_v$) do not share edges with any Class \ref{2dETcls1} polygons, combinatorial changes in the Euler transformation will not affect the number of edges in $\hf_v$.

\begin{lem}
  \label{lem:combinatorialdegree}
  Let $\hv$ is a vertex in a Class \ref{2dETcls1} polygon $\hf$ of $\hK$ created after collapsing $\pi$ adjacent edges in $\hf$, where combinatorial changes are allowed.
  Let $\hf_e$ be a Class \ref{2dETcls2} polygon that contains one of these collapsed edges.
  If no $\hf_e$ is collapsed to an edge, then degree of $\hv$ is $~2\pi + 4~$ else $2\pi + 4 - m$ where $m$ is number of polygons similar to $\hf_e$ collapsed to an edge. 
\end{lem}
\begin{proof}
  Since the polygon $\hf$ is allowed to have combinatorial changes in $\hK$, it will change degree of vertices in $\hf$.
  $\pi$ adjacent edges in $\hf$ have $\pi + 1$ vertices.
  Each end vertex of the path created by $\pi$ adjacent edges adds $3$ edges to $\hv$, and each interior vertex of the path adds $2$ edges to $\hv$ (\cref{fig:lemmaedgecontraction}).
  Hence $\hv$ has degree $2(\pi -1) + 3 + 3 = 2\pi + 4$ and $1$-skeleton of $\hK$ is still Euler.
  If $m > 0$ Class \ref{2dETcls2} polygons sharing one of these adjacent edges are collapsed into edges, then two edges sharing $\hv$ of each collapsed Class \ref{2dETcls2} polygon is replaced by one edge.
  Also, $m \leq \pi$ since each distinct edge in any Class \ref{2dETcls1} polygon is shared by a unique Class \ref{2dETcls2} polygon in the Euler transformation.
  This implies $\hv$ has degree $2\pi + 4 -2m + m = 2\pi + 4 - m$ and $1$-skeleton of $\hK$ is Euler depending on $m$ is even or odd (see \cref{fig:lemmaedgecontraction}).
\end{proof}

\begin{figure}[htp!] 
  \centering
  \includegraphics[scale=0.30]{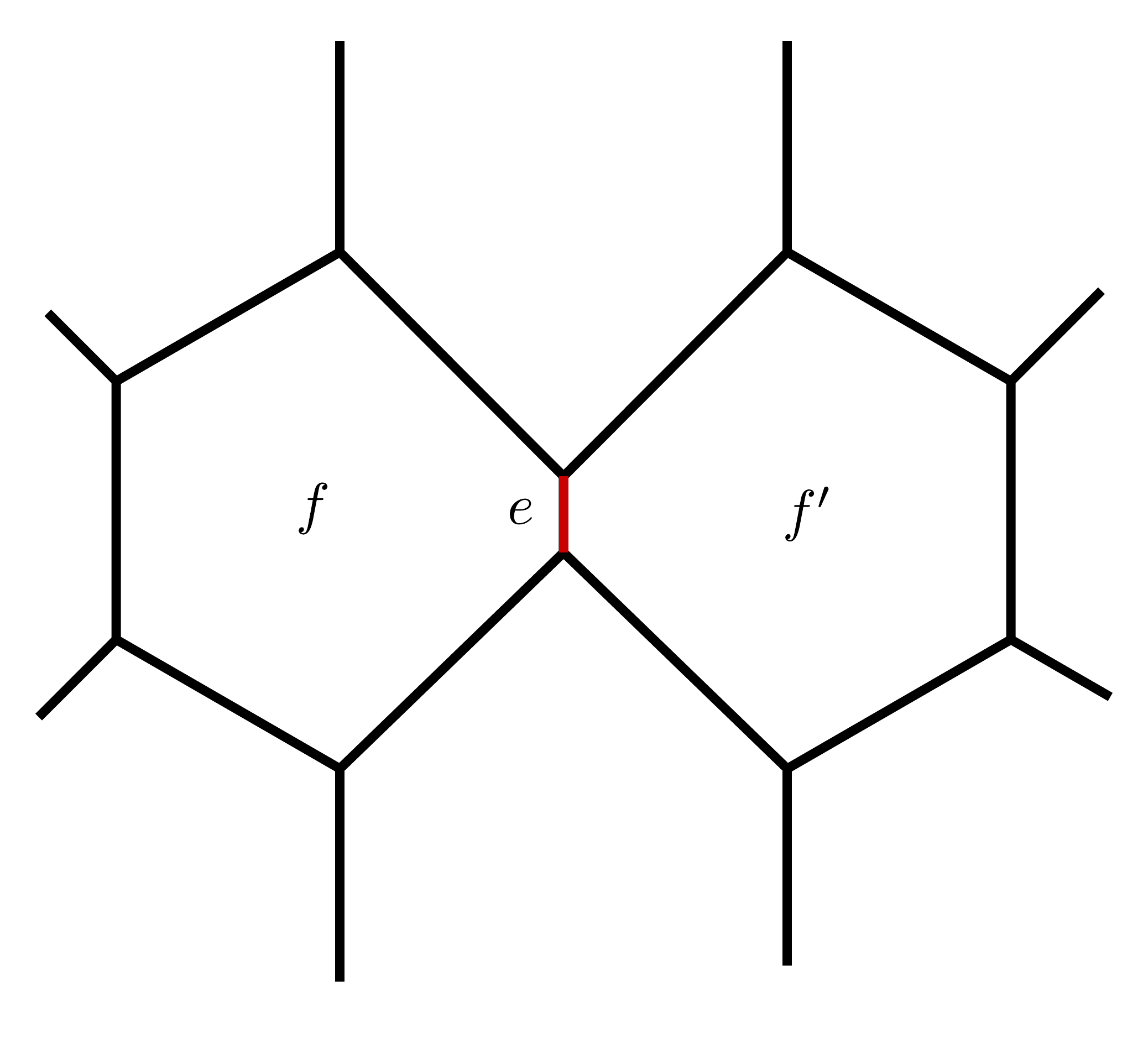}
  \quad
  \includegraphics[scale=0.35]{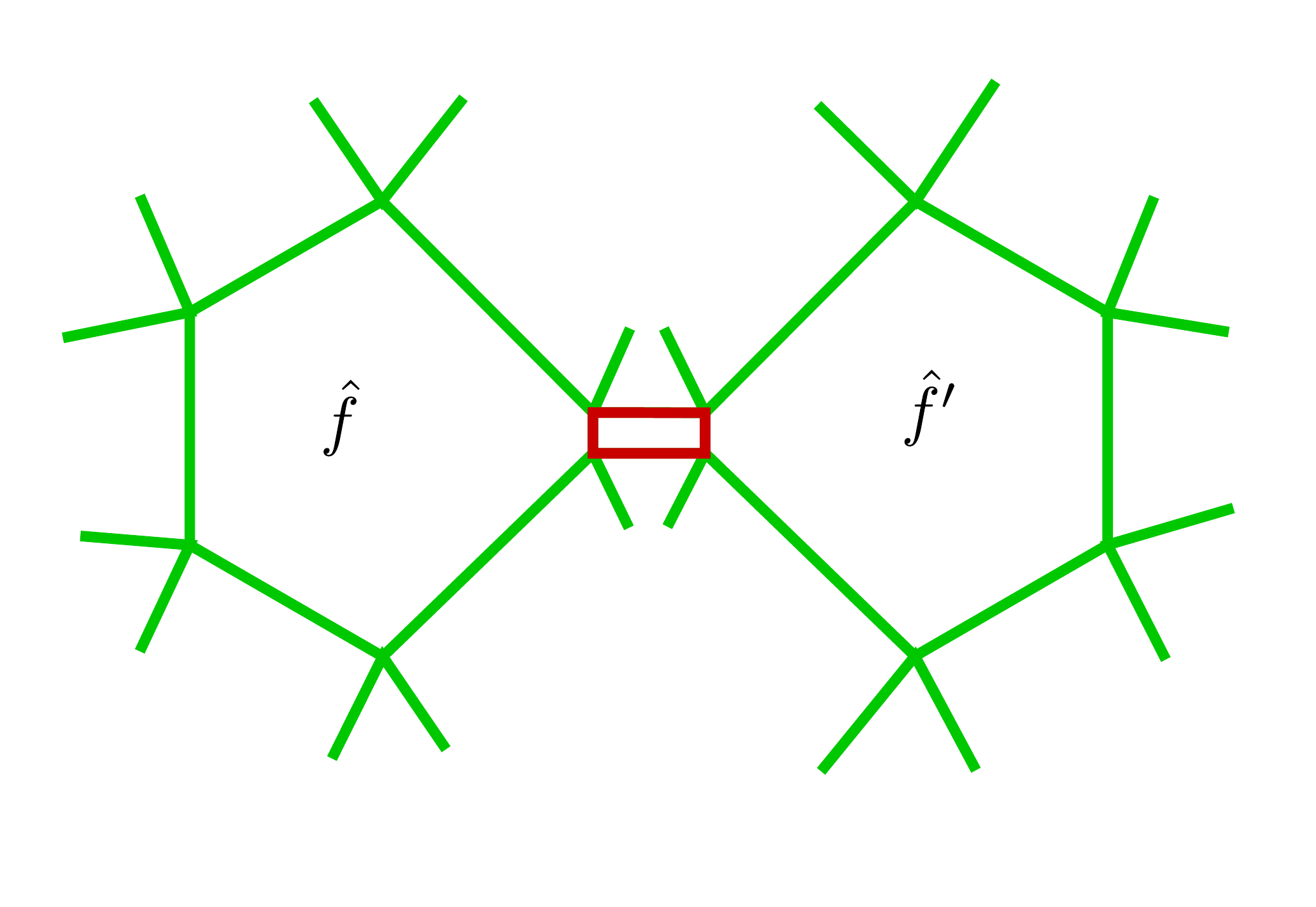} \\
  \includegraphics[scale=0.33]{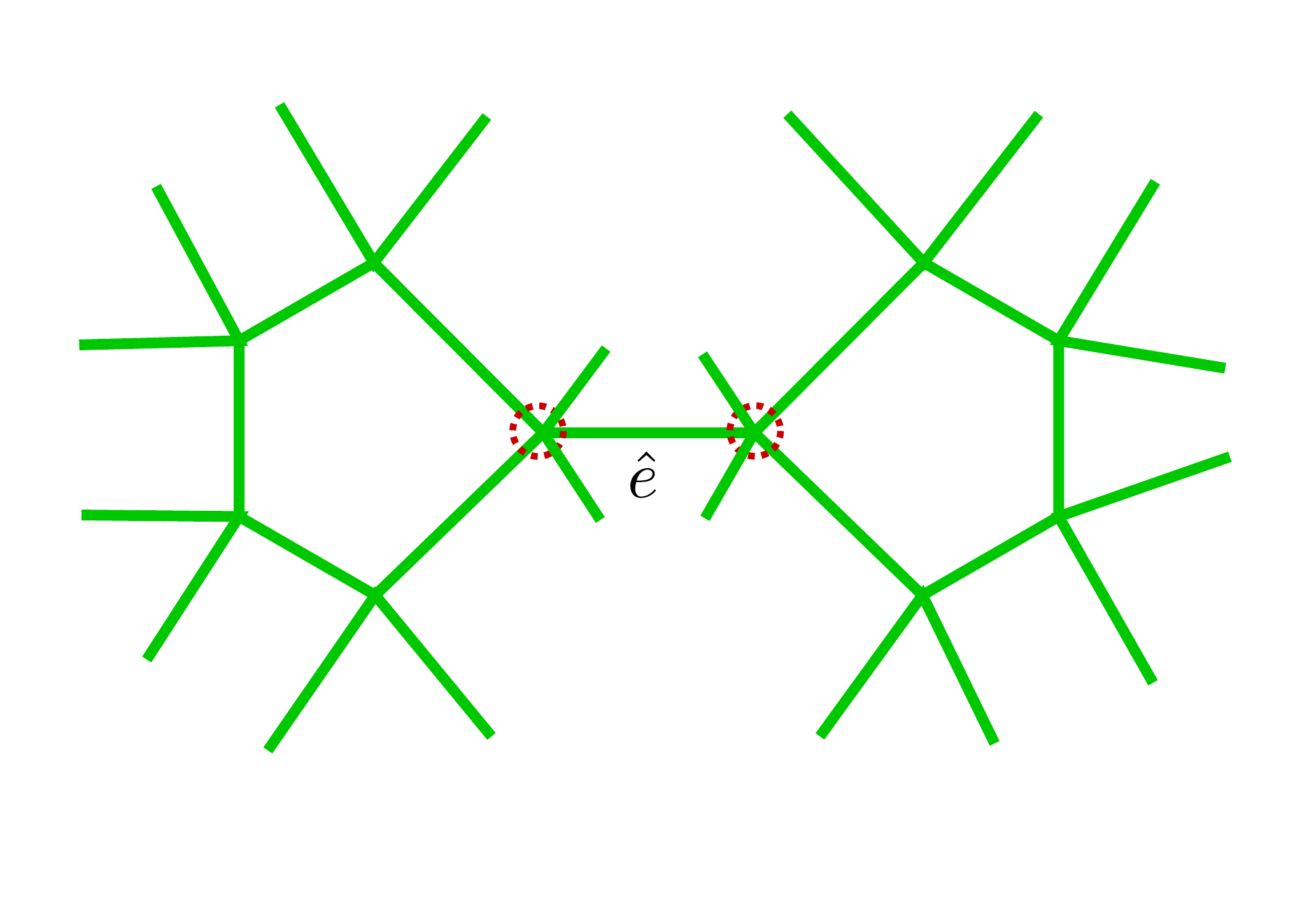}
  \caption{\label{fig:eulertransfecollapse}
    Two polygons of a $2$-complex $K$ in the plane (top left).
    Euler transformation of $K$ into $\hK$ (top right), where $\hf, \hf'$ are Class \ref{2dETcls1} polygons and $\hf_e$ (red) is a Class \ref{2dETcls2} polygon corresponding to edge $e$ (red) in $K$.
    With a higher mitered offset, $\hf_e$ is collapsed to edge $\he$; red circled vertices have odd degree (bottom).
  }
\end{figure}

If combinatorial changes are allowed in Euler transformation, we should apply Euler transformation to a local complex for any odd degree vertices created.
In the following lemma, we use the generalized Euler transformation to address the issue of odd degree vertices that may be created by combinatorial changes in $\hK$.

\begin{figure}[htp!] 
	\centering
	\includegraphics[scale=0.35]{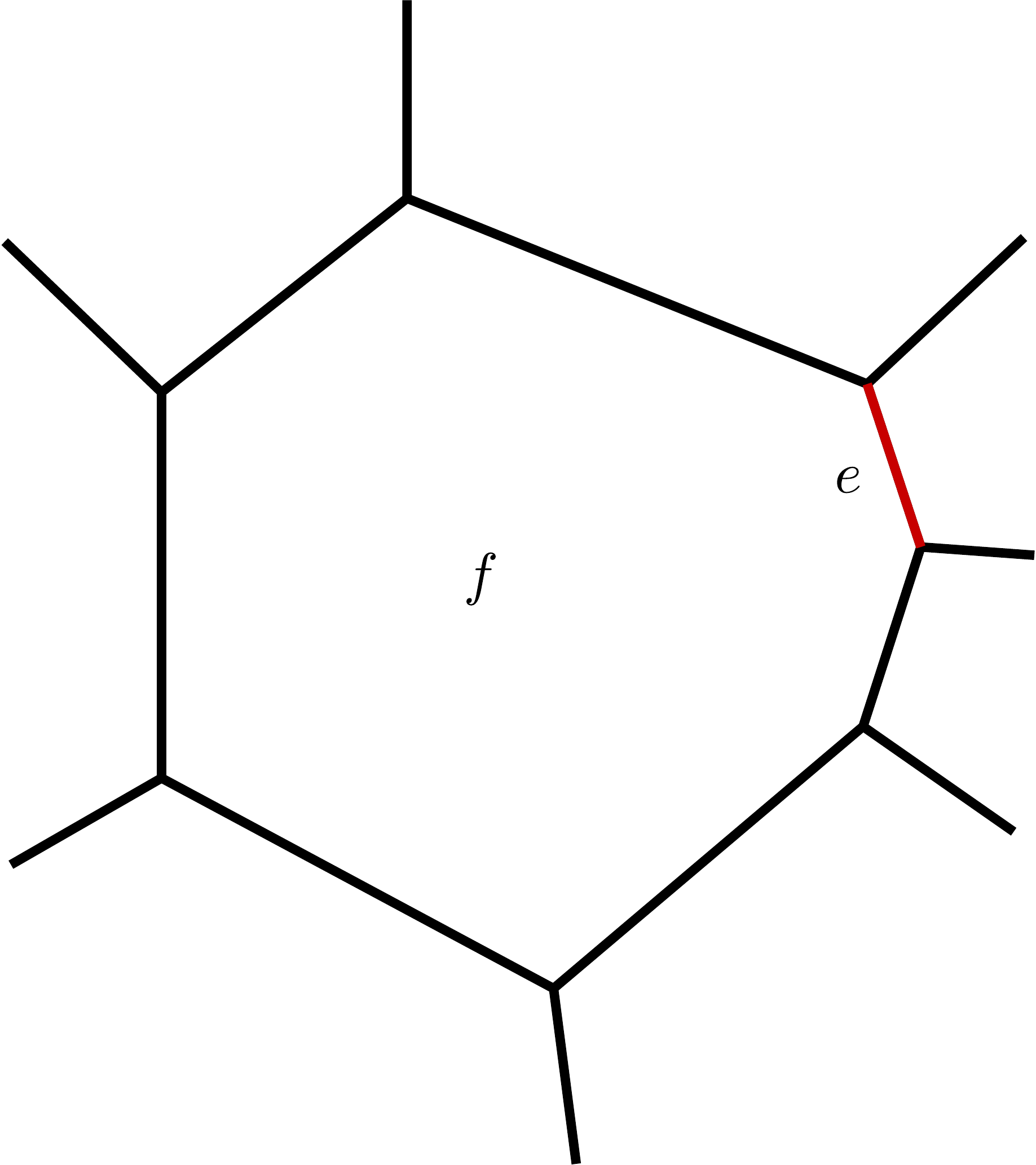}
	\hspace*{0.3in}
	\includegraphics[scale=0.40]{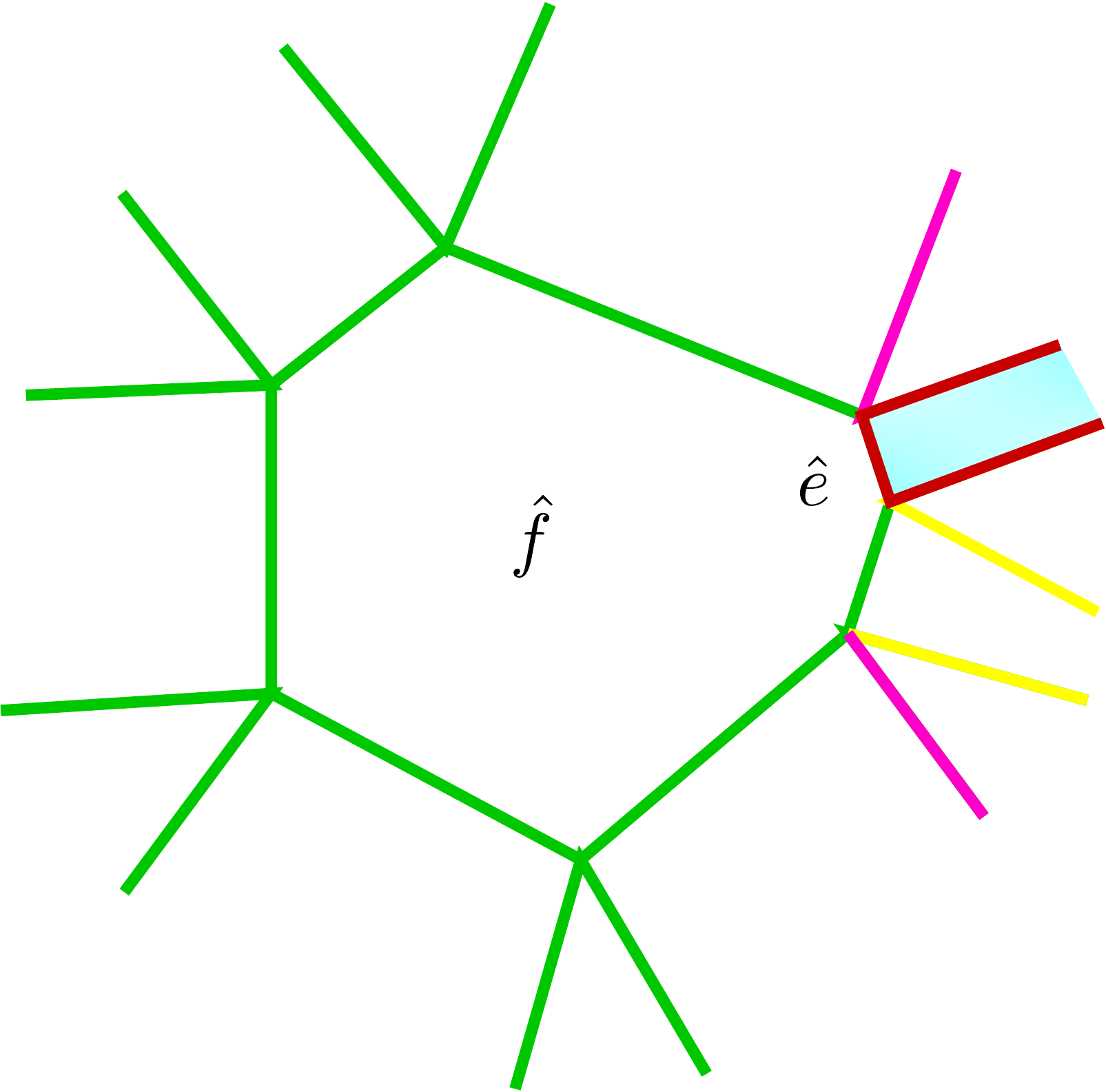}
        \vspace*{0.3in}\\
	\includegraphics[scale=0.40]{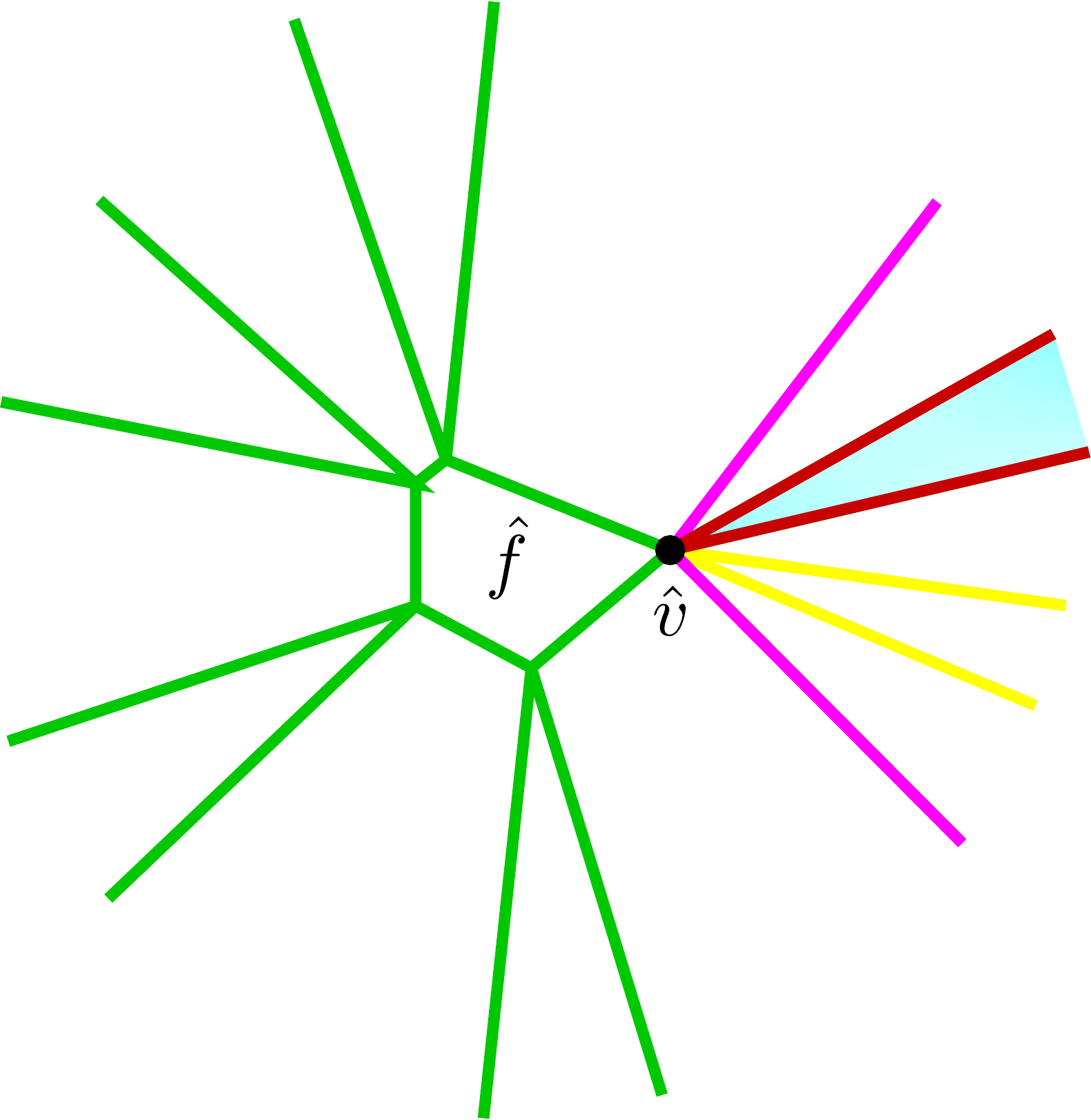}
	\hspace*{0.3in}
	\includegraphics[scale=0.40]{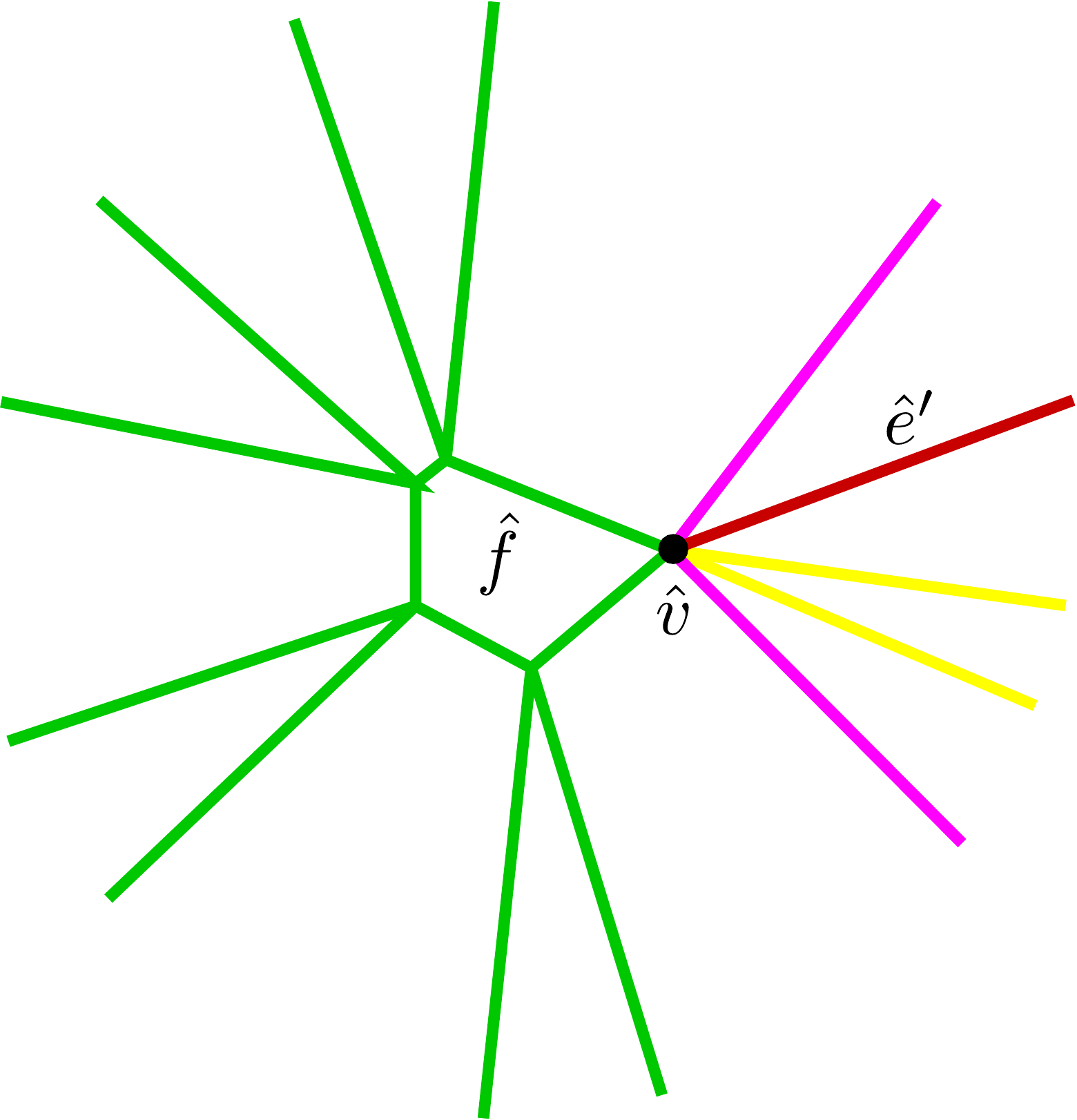}
	\caption{\label{fig:lemmaedgecontraction}
		$f$ is a polygon in $K$ (top left).
		Euler transformation of $K$ into $\hK$ (top right), where $\hf$ is the Class \ref{2dETcls1} polygon corresponding to $f$, and $\hf_e$ (blue) is the Class \ref{2dETcls2} polygon corresponding to edge $e$.
		Combinatorial changes are allowed in $\hK$ in the bottom left figure, and edge $\he$ of $\hf$ is collapsed to a point $\hv$ where $\hf_e$ (blue) is a triangle.
		As $\he \in \hf_e$ is collapsed to a point, the degree of $\hv$ in the $1$-skeleton of $\hK$ is $2(2) + 4 = 8$.
		In the version of $\hK$ shown in the bottom right figure, an edge of $\hf_e$ shared with some Class \ref{2dETcls1} polygon other than $\hf$ is also collapsed to a point.
		Here, $\hf_e$ is collapsed to an edge $\he'$ (red) and the degree of $\hv$ in the $1$-skeleton of $\hK$ is $2(2) + 4 - 1= 7$.
	}
\end{figure}

\begin{lem}
  \label{lem:localeulertransformation}
  Suppose $\hf_e$ in $\hK$ is collapsed to an edge $\he$, since combinatorial changes are allowed in $\hK$.
  Suppose $\hk$ is a sub $2$-complex of $\hK$ consisting of Class \ref{2dETcls3} polygons sharing any edge $\he$ in $\hK$,
  and let $\hk_j$ be a single component contained in $\hk$, since $\hk$ can have multiple components.
  Suppose $\hk_j = \cup \tilde{k}_i$, where $\tilde{k}_i$ is sub $2$-complex in $\hk_j, \hK$, and let $\tilde{k}_i, \tilde {k}_j$ do not share any edge in $\hK$, if $i \neq j$.
  If $\check{k}_i$ is the generalized Euler transformation of $\tilde{k}_i$ and $\bar{k}_j = \cup \check{k}_i$ is single component in $\bar{k}$, then the $1$-skeleton of $(\hK \smallsetminus \hk) \cup \bar{k}$ is Euler.  	
\end{lem}

\begin{proof}
  Since combinatorial changes are allowed in $\hK$, let $\hv$ be a vertex in $\hk_j$ created after $\pi$ adjacent edges of $\hf$ are collapsed to a vertex.
  $\hk_j$ is single component in $\hk$, and only Class \ref{2dETcls3} polygons share an edge with any collapsed Class \ref{2dETcls2} polygons $\hf_e$ is in $\hk_j$.
  Hence $\hv$ can be shared by some Class \ref{2dETcls3} polygon not in $\hk_j$.
  If any Class \ref{2dETcls3} polygon sharing $\hv$ is not contained in $\hk_j$, this Class \ref{2dETcls3} polygon does not share an edge with any Class \ref{2dETcls2} polygon.
  Hence such cells do not belong to any component $\hk_j$ in $\hk$.
  Since Class \ref{2dETcls1} polygons are edge-disjoint from any Class \ref{2dETcls3} polygons, and there are some Class \ref{2dETcls3} polygons and one Class \ref{2dETcls1} polygon ($\hf$) sharing vertex $\hv$ in $\hK$ but not contained in $\hk_j$, we get that $\hv$ is shared by an even number of additional edges not in $\hk_j$.
  Since Class \ref{2dETcls1} and Class \ref{2dETcls3} polygons are edge-disjoint in $\hK$, then any vertex ($\hv'$) in some Class \ref{2dETcls3} polygon in $\hk_j$ not similar to $\hv$ has two more edges, not in $\hk_j$ sharing $\hv'$.
  Hence all the vertices in $\hk_j$ have even number of additional edges in $\hK$ not contained in $\hk_j$ as shown in Figure \ref{fig:lemmalocaleulertransformation}.
  
  Let $\check{k}_i$ be the generalized Euler transformation of each sub $2$-complex $\tilde{k}_i \in \hk_j$.
  Then any vertex in $\hK \smallsetminus \hk \cup \bar{k}$ has an even number of edges connected to it, since each $\bar{k}_j = \cup \check{k}_i$ contributes even number of edges to any vertex shared by $\bar{k}_j$ and each vertex in $\hk$ has two more edges not contained in $\hk$.
  Hence the $1$-skeleton of $(\hK \smallsetminus \hk) \cup \bar{k}$ is Euler. 
\end{proof}

\begin{figure}[htp!] 
  \centering
  \begin{subfigure}[t]{3in}
    \centering
    \includegraphics[scale=0.33]{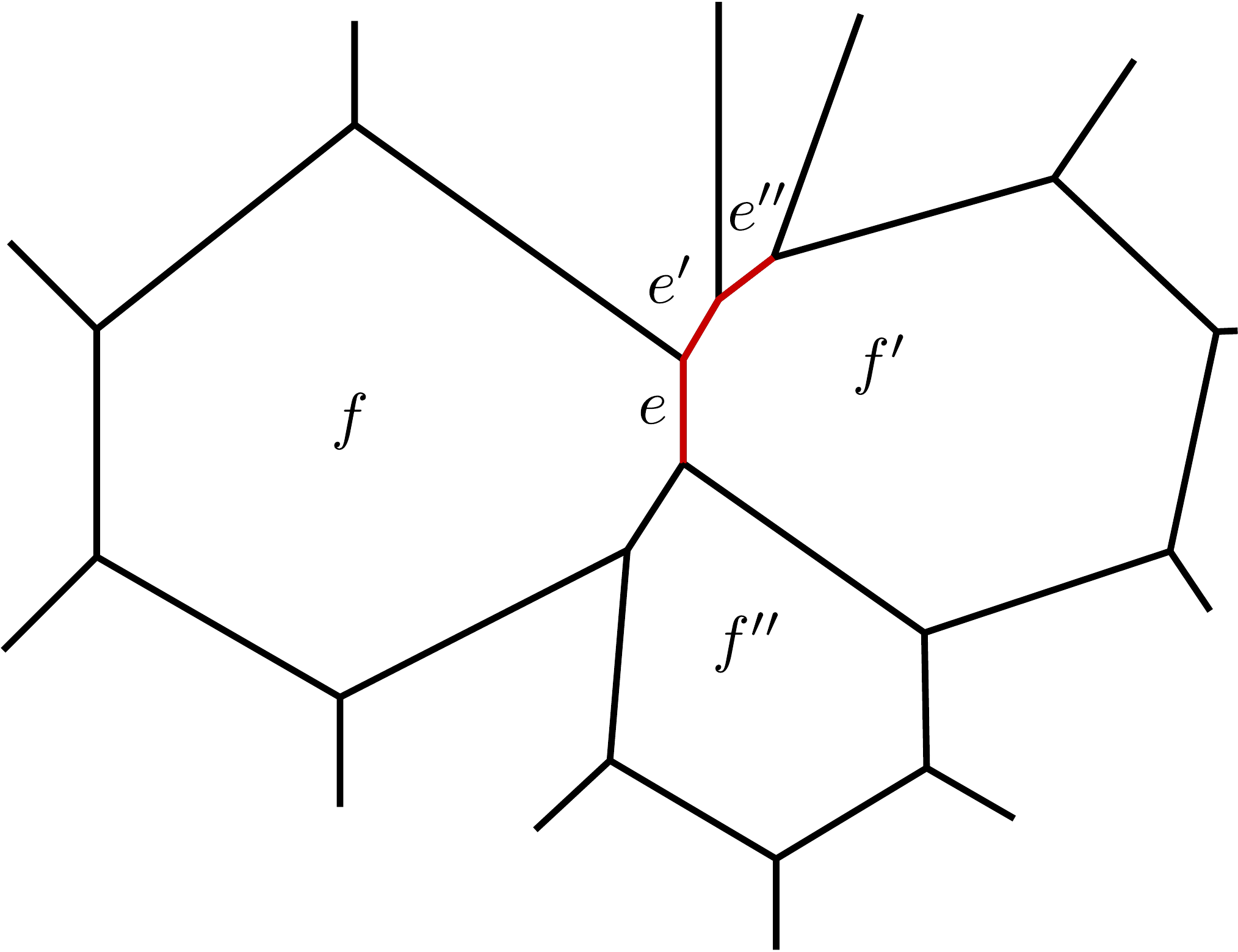}
    \caption{\label{fig:lemmalocaleulertransformationa}}
  \end{subfigure}
  \begin{subfigure}[t]{3in}
    \centering
    \includegraphics[scale=0.33]{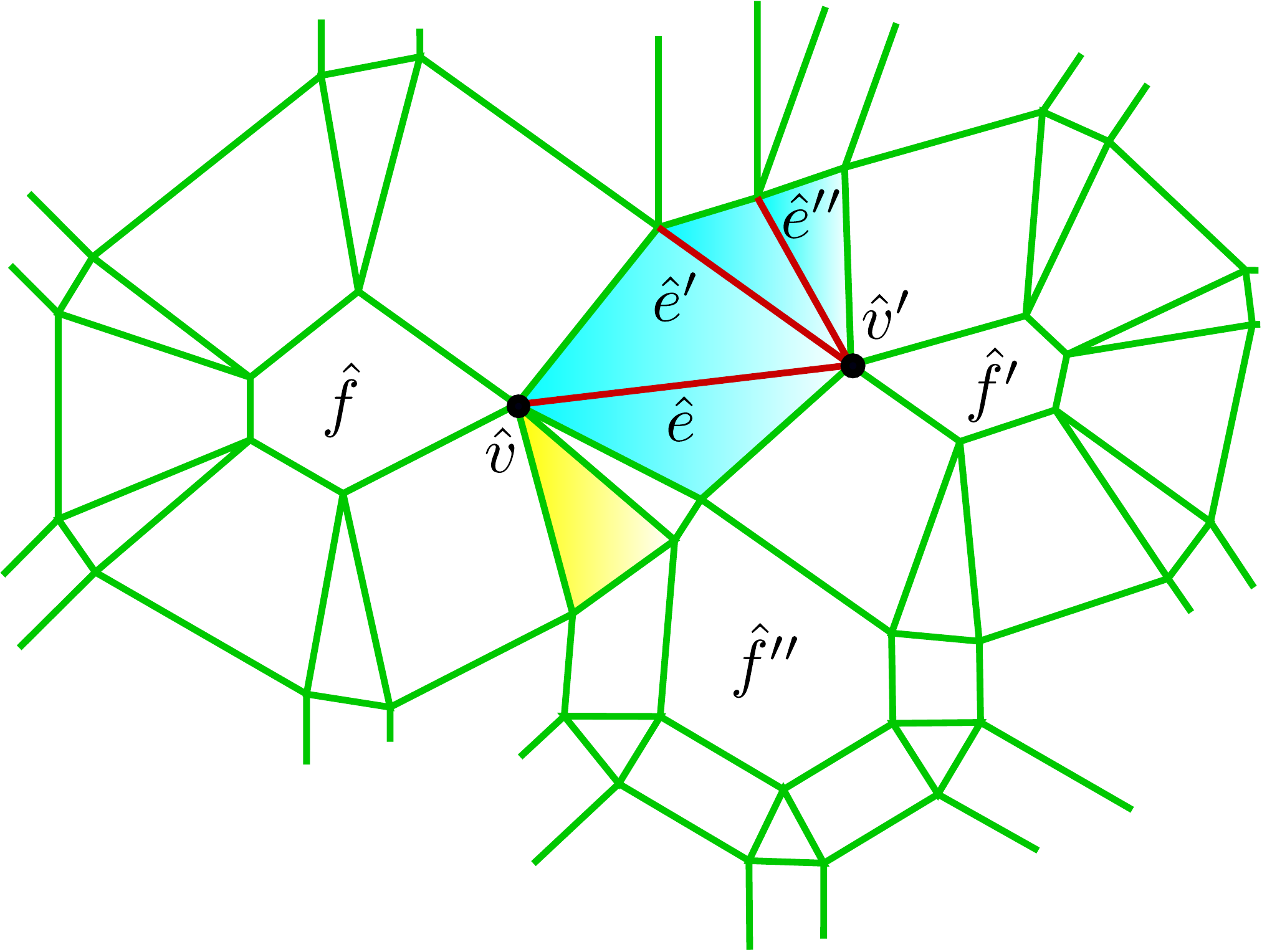}
    \caption{\label{fig:lemmalocaleulertransformationb}}
  \end{subfigure}
  \vspace*{0.2in}\\
  \begin{subfigure}[t]{3in}
    \centering
    \includegraphics[scale=0.33]{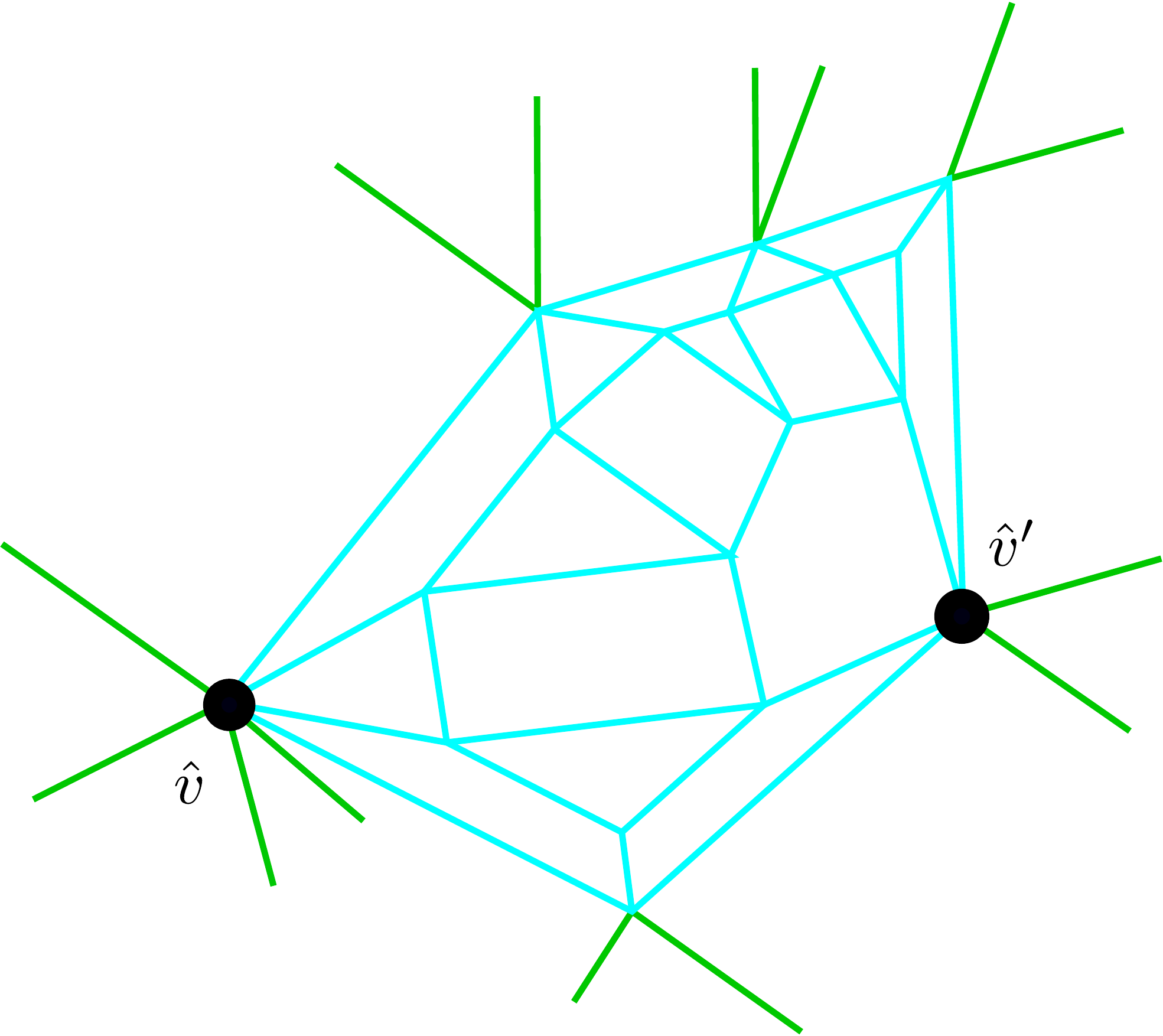}
    \caption{\label{fig:lemmalocaleulertransformationc}}
  \end{subfigure}
  \begin{subfigure}[t]{3in}
    \centering
    \includegraphics[scale=0.33]{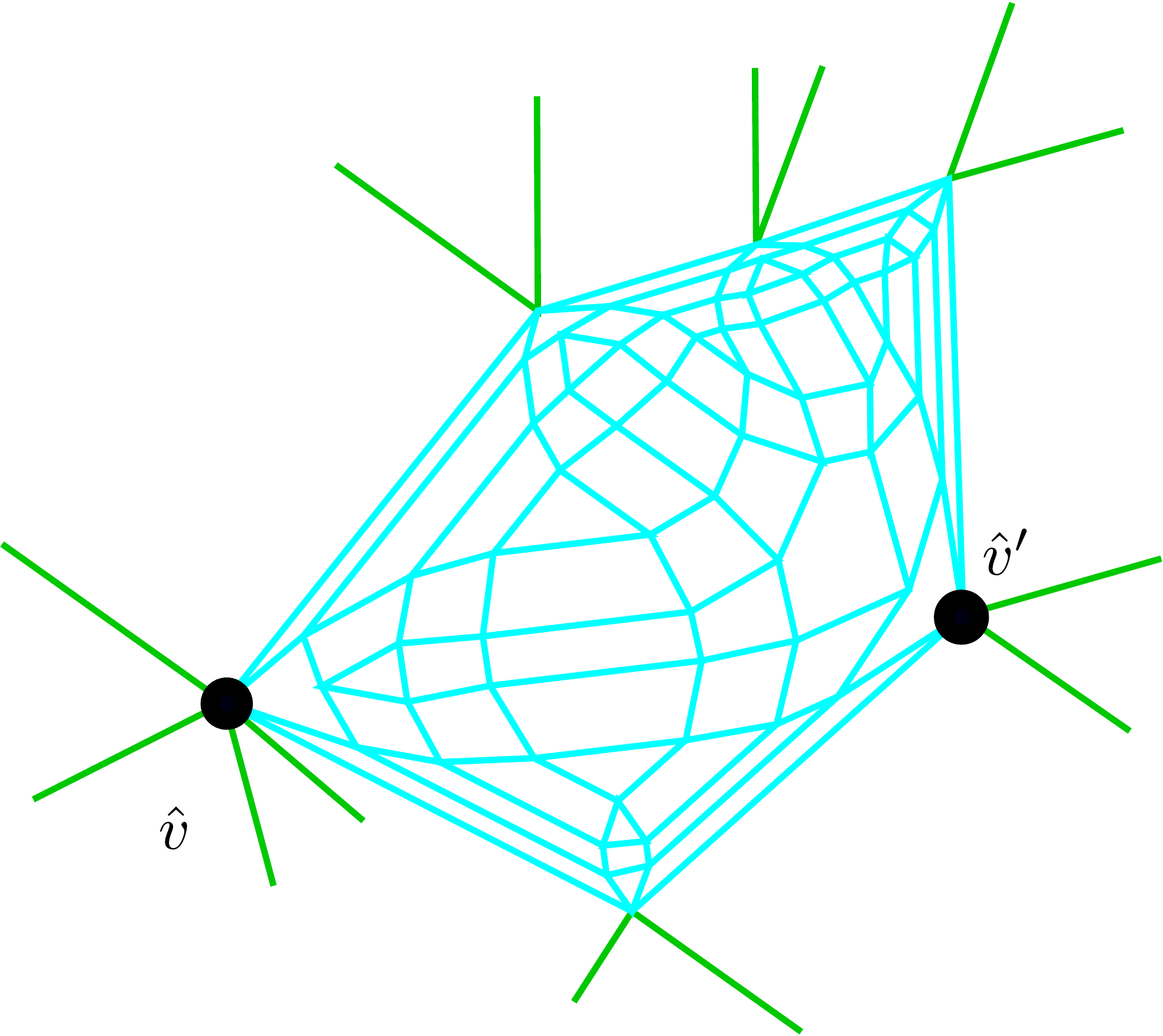}
    \caption{\label{fig:lemmalocaleulertransformationd}}
  \end{subfigure}
  \caption{  \label{fig:lemmalocaleulertransformation}
    (a) $f, f', f''$ are polygons in $2$-complex $K$.
    (b) $2$-complex $\hK$ after Euler transformation with combinatorial changes to some polygons in $\hK$. Class \ref{2dETcls2} polygons $\hf_e,\hf_{e'}, \hf_{e''}$ corresponding to $e, e', e''$ in $K$ are collapsed to edges $\he, \he', \he''$.
    Sub $2$-complex $\tilde{k}_i$ (in blue) of some $\hk_j$ in $\hK$ consists of Class \ref{2dETcls3} polygons sharing edges $\he, \he', \he''$ in the $1$-skeleton of $\tilde{k}_i$, and has vertices $\hv, \hv'$ with odd degree $7$.
    The number of edges (green) at each vertex of $\tilde{k}_i$ not in $\tilde{k}_i$ are even, and $\hv$ has one Class \ref{2dETcls3} polygon (yellow) not contained in $\tilde{k}_i$.
    (c) and (d) show generalized Euler transformation of $\tilde{k}_i$.
    (d) shows $\check{k}_i$ (blue), the generalized Euler transformation of $\tilde{k}_i$.
    Now $\hv, \hv'$ and other vertices of $\check{k}_i$ have even degrees in the $1$-skeleton of $\hK \smallsetminus \tilde{k}_i \cup \check{k}_i$.
  }
\end{figure}

Let vertex $\hv$ in the sub $2$-complex $\tilde{k}_i$ of single component $\hk_j$ in $\hk$ have odd  degree in $\tilde{k}_j$.
We apply generalized Euler transformation to any such sub $2$-complex $\tilde{k}_i$.
Then by Lemma \ref{lem:localeulertransformation}, the new $2$-complex is Euler.

\subsection{Euler Transformation with Topological Changes}\label{sec:eulertransformwithTopolog}

When a polygon $f$ in $K$ is split into multiple Class \ref{2dETcls1} polygons in $\hK$ after mitered offset, it is said to undergo a \emph{topological change} (see \cref{fig:topologychangelemma}). 
If a polygon in $K$ is concave, then its mitered offset could create topological changes.
Without loss of generality we assume the $2$-complex $K$ consists of polygons satisfying \cref{asmn:Kholesoutside}, and its polygons can be convex or concave.

\begin{lem}
  Let the complex $\hK$ be created by Euler transformation with some polygons undergoing only topological (no combinatorial) changes.
  Then the degree of vertices in $\hK$ is even.
  Furthermore, if $K$ is connected, then so is $\hK$. 
\end{lem}

\begin{figure}[htp!] 
	\centering
	\includegraphics[scale=0.30]{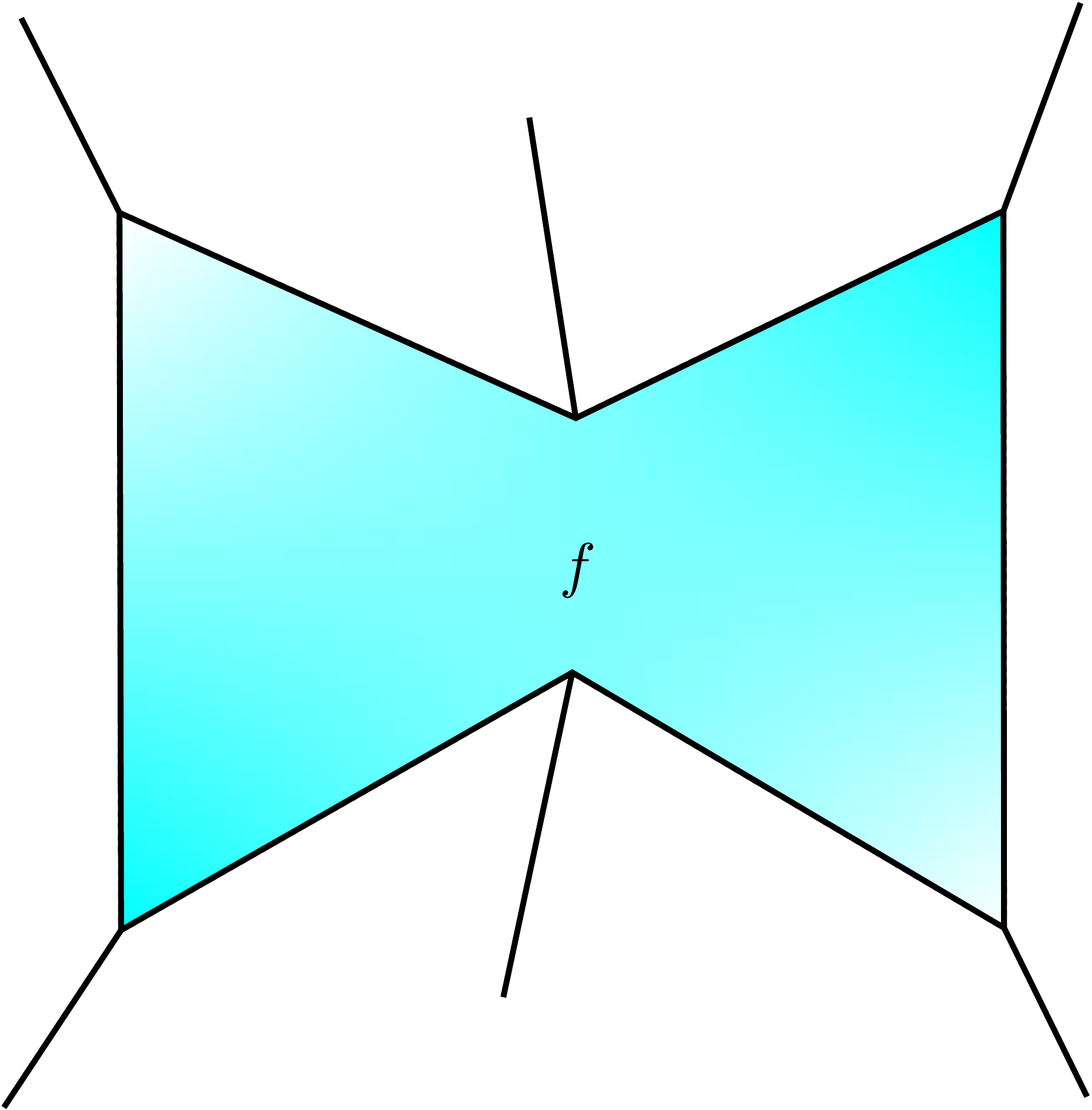}
	\quad\quad
	\includegraphics[scale=0.30]{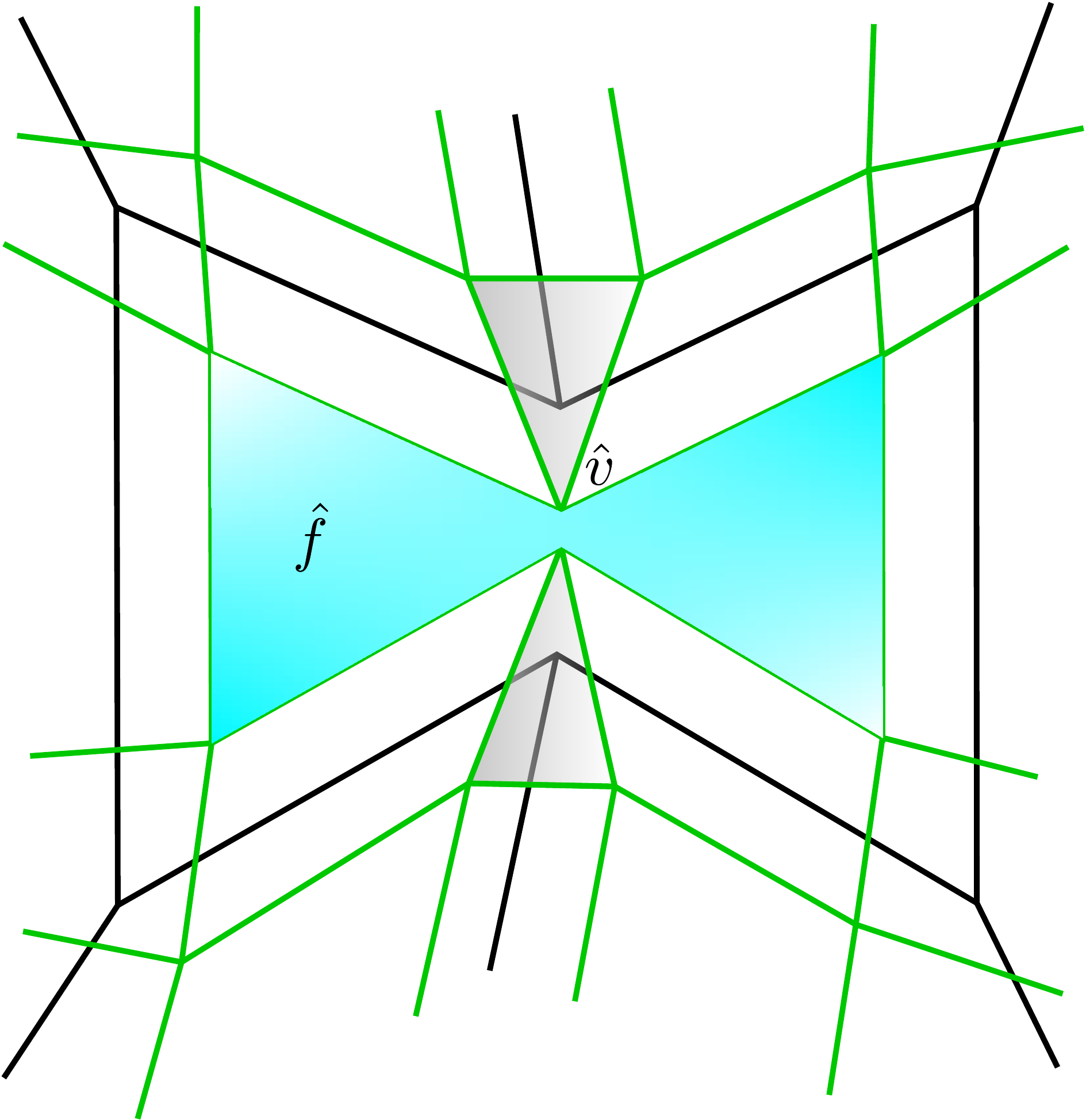}
        \vspace*{0.2in}\\
	\includegraphics[scale=0.30]{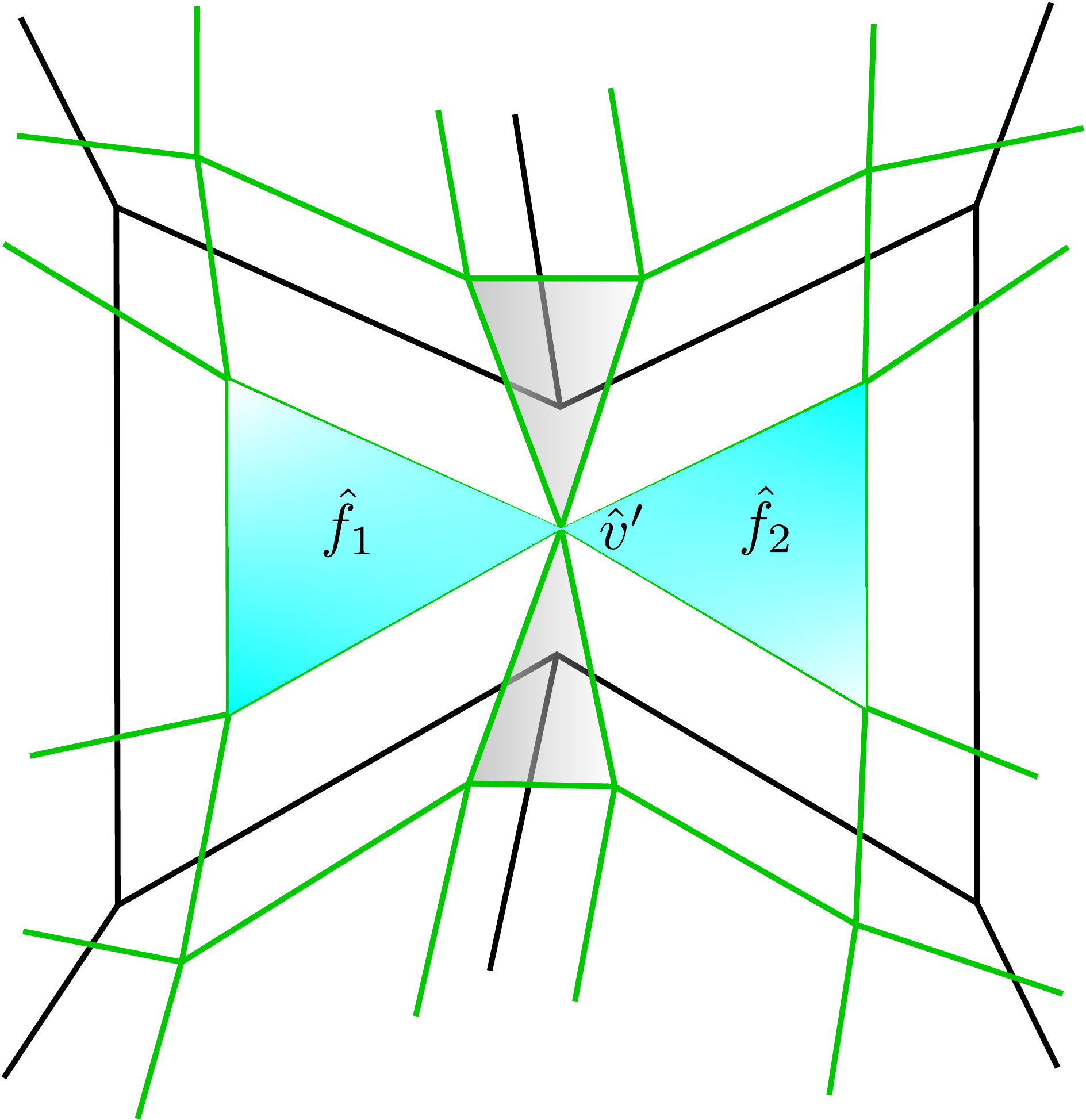}
	\quad\quad
	\includegraphics[scale=0.30]{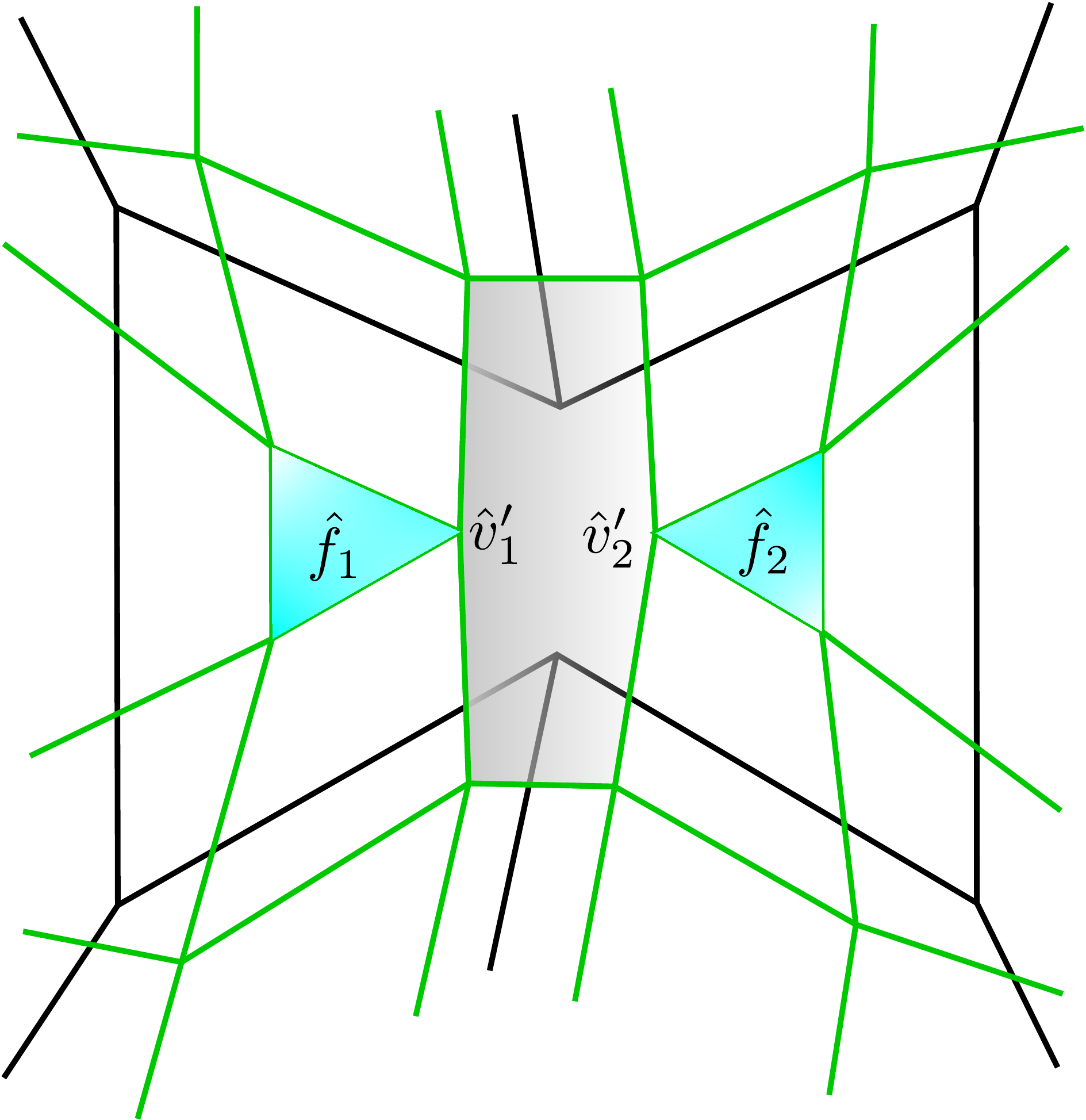}
	\caption{\label{fig:topologychangelemma}
          Top left: Non-convex polygon $f$ (blue) in $K$.
          Top right: Polygon $\hf$ (blue) and Class \ref{2dETcls3} polygons (gray) in $\hK$. 
          Bottom left: $\hf$ split into two polygons $\hf_1, \hf_2$ (blue) by joining two vertices of $\hf$ to $\hat{v}'$.
          Since $\hf_1, \hf_2$ and Class \ref{2dETcls3} polygons (gray) at $\hat{v}'$ are edge-disjoint, $\hat{v}'$ has degree $8$ and $q=2$.
          Bottom right: With higher mitered offset, $\hf_1, \hf_2$ are completely disjoint and $\hat{v}'$ split up into $\hat{v}'_1, \hat{v}'_2$ creating one polygon $F$ (gray), where $\hat{v}'_1, \hat{v}'_2$ has degree $4$.
          }
\end{figure}  

\begin{proof}
  When there are no topological changes, any vertex $\hv$ in polygon $\hf$ is shared by edge-disjoint Class \ref{2dETcls1} and \ref{2dETcls3} polygons in $\hK$.  Let us allow topological changes to $\hf$.
  There are two possible cases.
  In the first case, $\hf$ splits into multiple polygons at vertex $\hat{v}'$ by joining two or more vertices in $\hf$ due to mitered offset.
  Then $\hat{v}'$ is still shared by edge-disjoint Class \ref{2dETcls3} polygons as shown in \cref{fig:topologychangelemma}.
  In the second case, if we further increase the offset distance then $\hv'$ will split into $q$ new vertices ($\hat{v}'_1, \dots, \hat{v}'_q$), where $q$ is the number of Class \ref{2dETcls3} polygons sharing $\hat{v}'$. 
  It will also join $q$ Class \ref{2dETcls3} polygons sharing vertex $\hat{v}'$ into one polygon ($F$) as shown in \cref{fig:topologychangelemma}.
  Then any $\hat{v}_i'$ is shared by edge-disjoint $F$ and Class \ref{2dETcls1} polygons.
  Hence degree of vertices in $\hK$ is even. 	   
  There exists a path from any vertex of the split polygon ($\hf_i$) to any other vertex of $\hf_j$, and hence $\hK$ is connected. 
\end{proof}

We have discussed cases when only combinatorial or only topological changes occur after transformation.
But if \emph{both} combinatorial and topological changes occur, then odd degree vertices may be created due to combinatorial changes. 
In this case, we apply local Euler transformation to some subcomplex of $\hK$, and then by \cref{lem:localeulertransformation} the $1$-skeleton of $\hK$ is again Euler.

\section{Slicing}\label{sec:slicing}

The goal of our 3D printing approach is to have maximum continuous print path and minimum travel path (i.e., non-print path) in each layer.
Further, when printing multiple layers on top of each other, we want to ensure there is no printing in free space.
Ensuring we avoid printing in free space depends crucially on the geometric complexity of the object as well as on the first round of slicing.
We first formalize the condition that the sequence of layers generated by slicing must satisfy in order to prevent printing in free space (Section \ref{subsec:contilayers}).
We assume this condition is satisfied by the layers of the input to our \emph{clipping} procedure, which produces meshes for each polygon in a layer that are guaranteed to be Euler (Section \ref{subsec:clipping}).
 
\subsection{$\epsilon$-Continuous Layers}\label{subsec:contilayers}
Let $\Ps_i = \{P_{ij}\}$ and $\Ps_{i+1} = \{P_{i+1,j}\}$ are sets of the polygons in two consecutive layers created by slicing.
The two layers are said to be $\epsilon$-continuous if for every point $\vx \in P_{i+1,j}$ there exists a point $\vy$ in \emph{some} $P_{ij} \in \Ps_i$ such that $ d(\vx, \vy) \leq \epsilon$ for \emph{all} $P_{i+1,j} \in \Ps_{i+1}$, where  $\epsilon = c r$ with $0 \leq c \leq 1$ and $r$ being the radius of extruder. 
The parameter $c$ determines the maximum \emph{overhang} allowed for the material deposited in a layer over the material in the layer immediately below.
We assume there are sufficient numbers of perimeters in each layer to support the boundary edges in the layer above.
Value of $c$ is chosen based on various design and material considerations.
There are alternative approaches to handle overhangs in specific cases, e.g., using self-supporting rhombic infill structures \cite{WuWaZhWe2016}.
For general applicability of our framework, we assume the output of the slicing step in the design process produces layers that are $\epsilon$-continuous in consecutive pairs.

%
\subsection{Clipping}\label{subsec:clipping}

Suppose $\hK$ is the Euler transformation of $K$, which meshes the union of polygons $\displaystyle \cup_i \Ps_i = \cup_{i,j} P_{ij}$ from all layers, where $P_{ij} \in \Ps_i$ is the $j$th polygon in $i$th layer.
Each polygon $P_{ij}$ has a region $R_{ij}$ to be filled with infill lattice (note that $R_{ij} \subset P_{ij}$ can happen as some polygons may have edges along the boundary of the print domain).
Suppose $\tR_{ij}$ is the inward Minkowski offset with a ball of radius $r$, the extruder radius, of the region $R_{ij}$.
We will use $\tR_{ij}$ instead of $R_{ij}$ to generate the infill lattice for $P_{ij}$.
The reason behind this step is explained in Step \ref{itm:support} (to print the Support Perimeter).
Let polygons $R_{ij}$ and $\tR_{ij}$ be represented by the clockwise-ordered sequence of vertices $\{v_1, .... , v_n\}$ and $\{\tv_1, \dots , \tv_n\}$, respectively.
We define the \emph{clip} operation for intersecting (or clipping) a $2$-complex with a polygon, which may produce a $2$-complex that may have multiple components, and may not be pure.
We also define the \emph{patch} operation that converts the $2$-complex produced by a clip operation back into a connected pure $2$-complex.

\begin{defn} \label{def:clip}
  \emph{({\bfseries Clip})}
  We define how to construct $\tK$, the output of clipping the $2$-complex $\hK$ with polygon $\tR_{ij}$.
  Add  to $\tK$ the polygons, edges, and vertices of $\hK$ contained in $\tR_{ij}$.
  For edges in $\hK$ cut by the boundary of $\tR_{ij}$, add to $\tK$ the portions inside $\tR_{ij}$ as new edges, and their points of intersection on the boundary of $\tR_{ij}$ as new vertices.
\end{defn}
Note that the result of a clip operation may not necessarily be a pure $2$-complex, and can have multiple components (see \cref{fig:multicomponent}).

\begin{defn} \label{def:patch}
  \emph{({\bfseries Patch})}
  Let $\tK$ be the output of a clip operation as specified in \cref{def:clip}.
  Suppose $S= \{\tv_{n+1}, \tv_{n+2}, .... , \tv_{n+m}\}$ is a clockwise ordered sequence of all points of intersection of $\tR_{ij}$ and the $1$-skeleton of $\hK$ with odd degrees in the $1$-skeleton of $\tK$
  (note that $m$ will be even).
  Since $\tR_{ij}$ can intersect edges in $\hK$ between or at their end point(s), vertices in $S$ can be terminal or boundary vertices in the $1$-skeleton of $\tK$.
  Join alternate pair of vertices in $S$ by a clockwise path on $\tR_{ij}$ as shown in Figure \ref{fig:printingprocess}.
  There are two possible choices of joining alternate pairs of vertices ($1$-$2$, $3$-$4,\dots$ or $2$-$3$, $4$-$5,\dots$).
  Pick the option that ensures the end vertices of components represented by subsequences of $S$ are connected by these paths to end vertices of adjacent components.
  Also add new polygons to $\tK$ whose edges include the edges in new paths added as described above, new edges added by the clip operation on $\hK$, and the edges of $\hK$ contained in $\tR_{ij}$.	
\end{defn}

The Patch operation restores the Euler and connected nature of the input complex, except when the Clip step produces isolated simple paths.
In the latter case, the Patch operation still leaves each component Euler.

\begin{lem}\label{lem:singlecomponent}
  Let $\hK$ be a connected pure $2$-complex and its $1$-skeleton is Euler.
  Then $\tK$ produced by the patch operation on $\hK$ is a connected pure $2$-complex and its $1$-skeleton is Euler, assuming none of the components in $\tK$ after the Clip step is a simple path. 
\end{lem}
\begin{proof}
  Clipping $\hK$ with region $\tR_{ij}$ can create multiple components in the infill lattice if $\tR_{ij}$ intersects any polygon in $\hK$ more than two times, or an edge more than once, or all edges connected to a vertex (see \cref{fig:multicomponent}).
  Each component created in this process has an even number of odd degree vertices (by handshake lemma).

  Let $S' =\{\tv_1, \dots, \tv_p\}$ be a clockwise ordered subsequence of vertices in $S$ of some component (with $p < m$).
  As specified in \cref{def:patch}, we join alternate pairs of vertices in $S'$ by a path on $R_{ij}$ (also see the Clipping Step \ref{itm:patch}) such that edge $\{\tv_1, \tv_2\}$ is not included.
  Since $p$ is even, edge $\{\tv_{p-1}, \tv_{p}\}$ is also not included.
  Hence the first and last vertices in $S'$ ($\tv_1$ and $\tv_p$) are left unpaired, but all intermediate vertices are now connected to $\tK$.
  In this case, the patch operations (in the Clipping step) will necessarily pair $\tv_1$ with a similar unpaired end vertex of the previous component, and also pair $\tv_p$ with the unpaired start vertex of the next component.
  Hence the extra edges added by the patch steps ensure that we get a single connected component.
  Also note that each odd degree vertex gets one additional edge, thus making its degree even.
  Hence the $1$-skeleton of $\tK$ is Euler.
  Finally, the new polygons added to $\tK$ (as specified at the end of \cref{def:patch}) ensure that the resulting complex is pure.
\end{proof}

\begin{figure}[htp!]
  \centering
  \begin{subfigure}[t]{1.5in}
    \centering
    \includegraphics[scale=0.25]{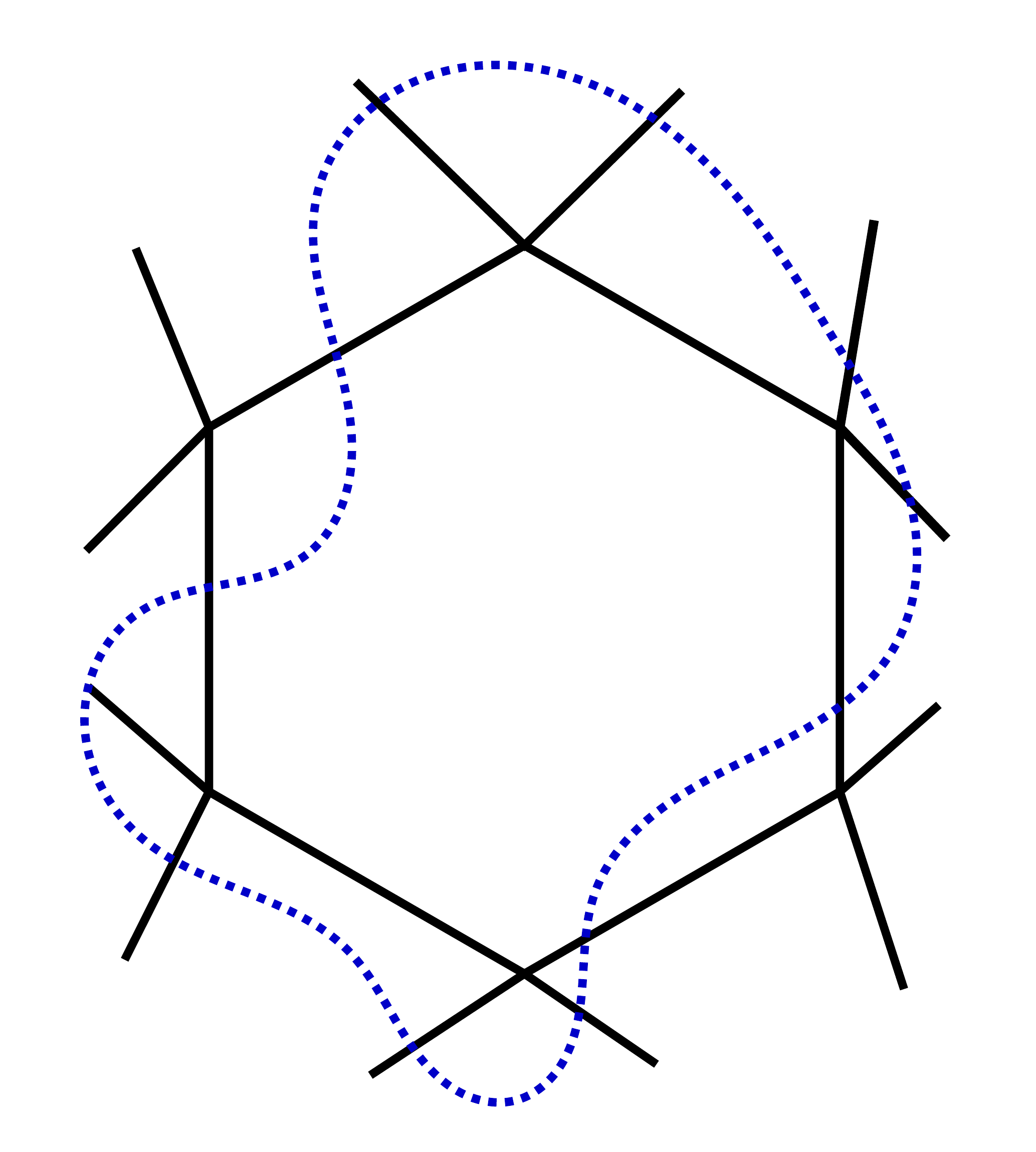}
    \caption{\label{fig:multicomponenta}}
  \end{subfigure}
  \hspace*{1in}
  \begin{subfigure}[t]{1.5in}
    \centering
    \includegraphics[scale=0.25]{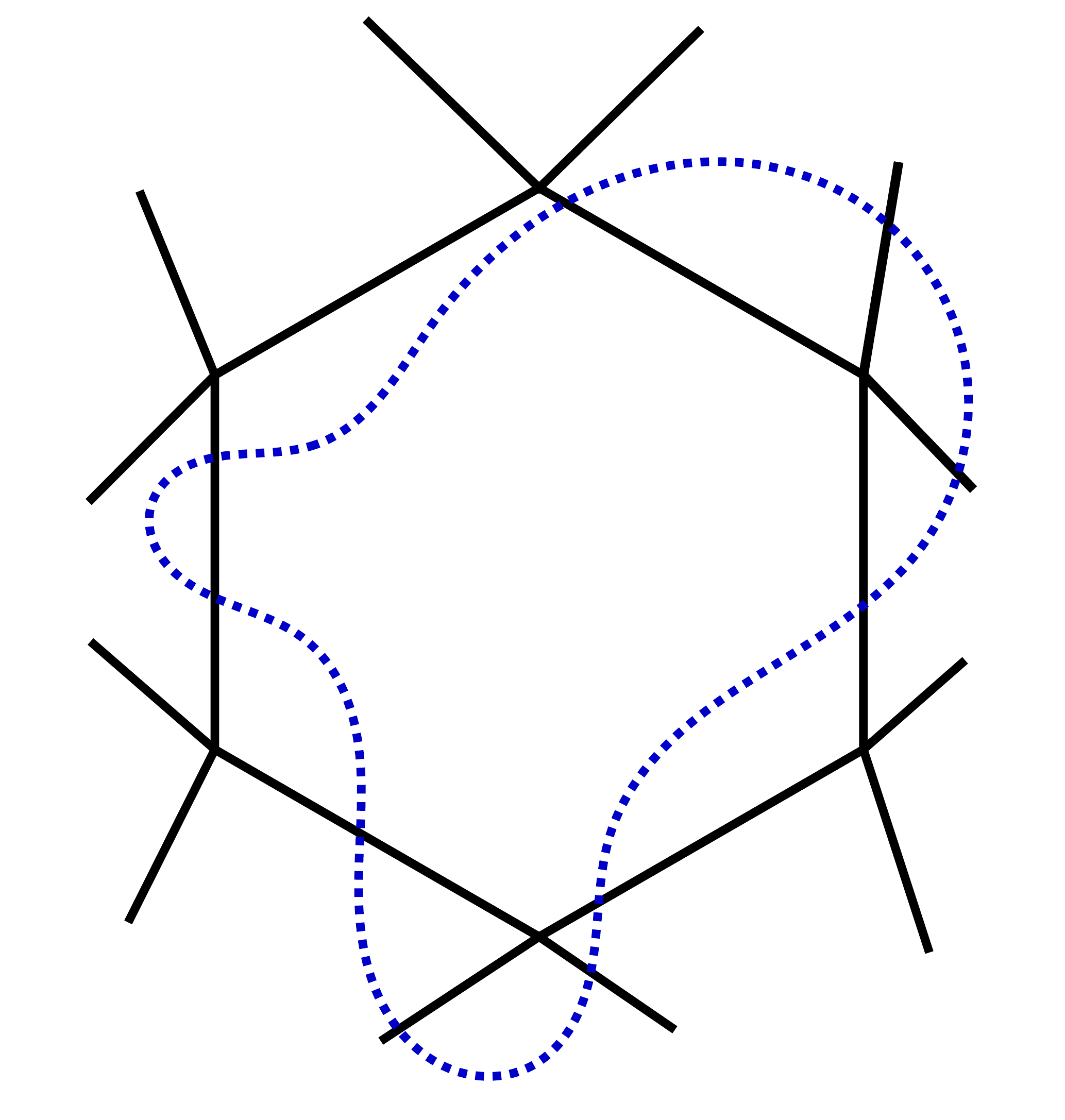}
    \caption{\label{fig:multicomponentd}}
  \end{subfigure}

  \begin{subfigure}[t]{1.5in}
    \centering
    \includegraphics[scale=0.25]{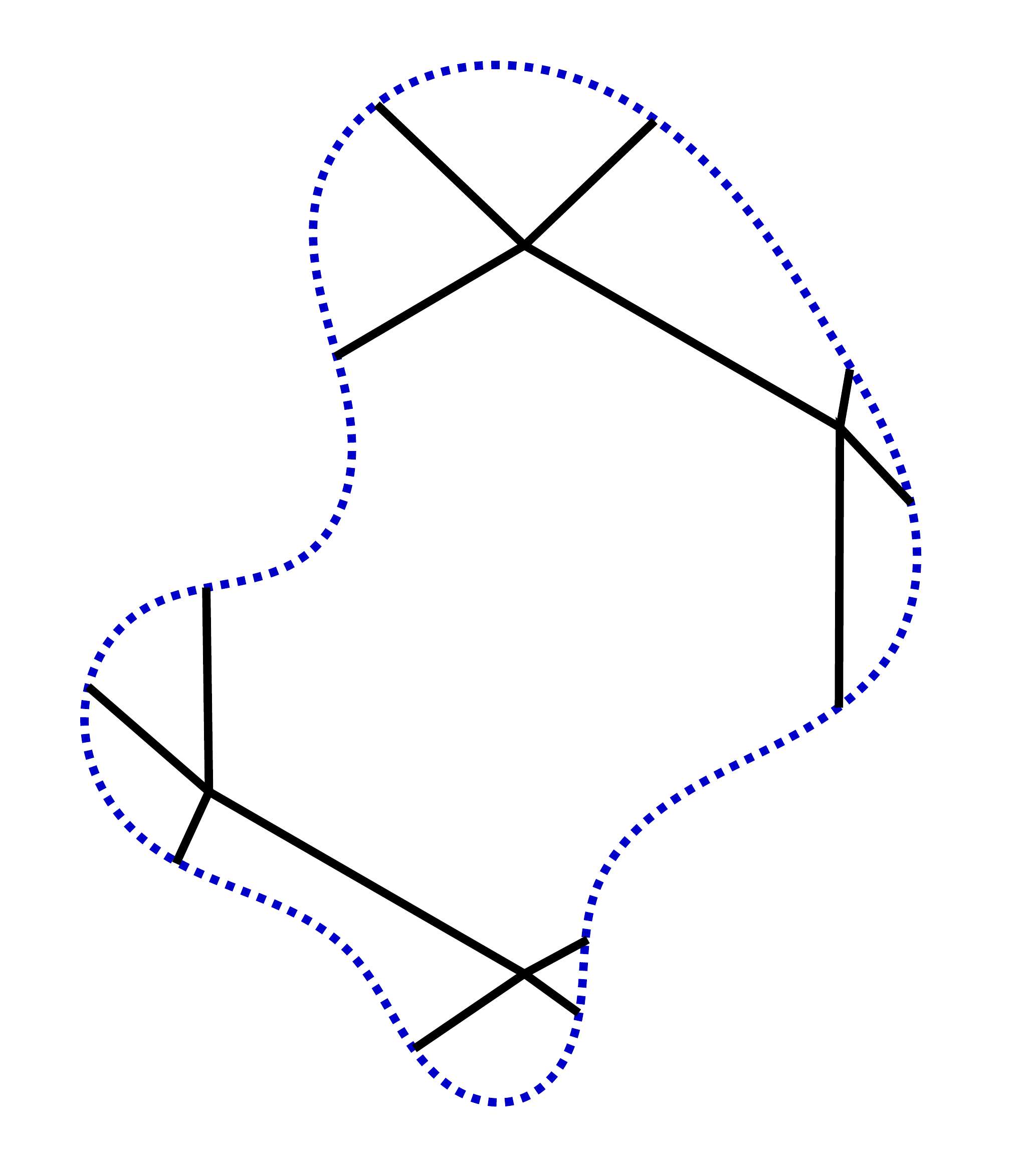}
    \caption{\label{fig:multicomponentb}}
  \end{subfigure}
  \hspace*{1in}
  \begin{subfigure}[t]{1.5in}
    \centering
    \includegraphics[scale=0.25]{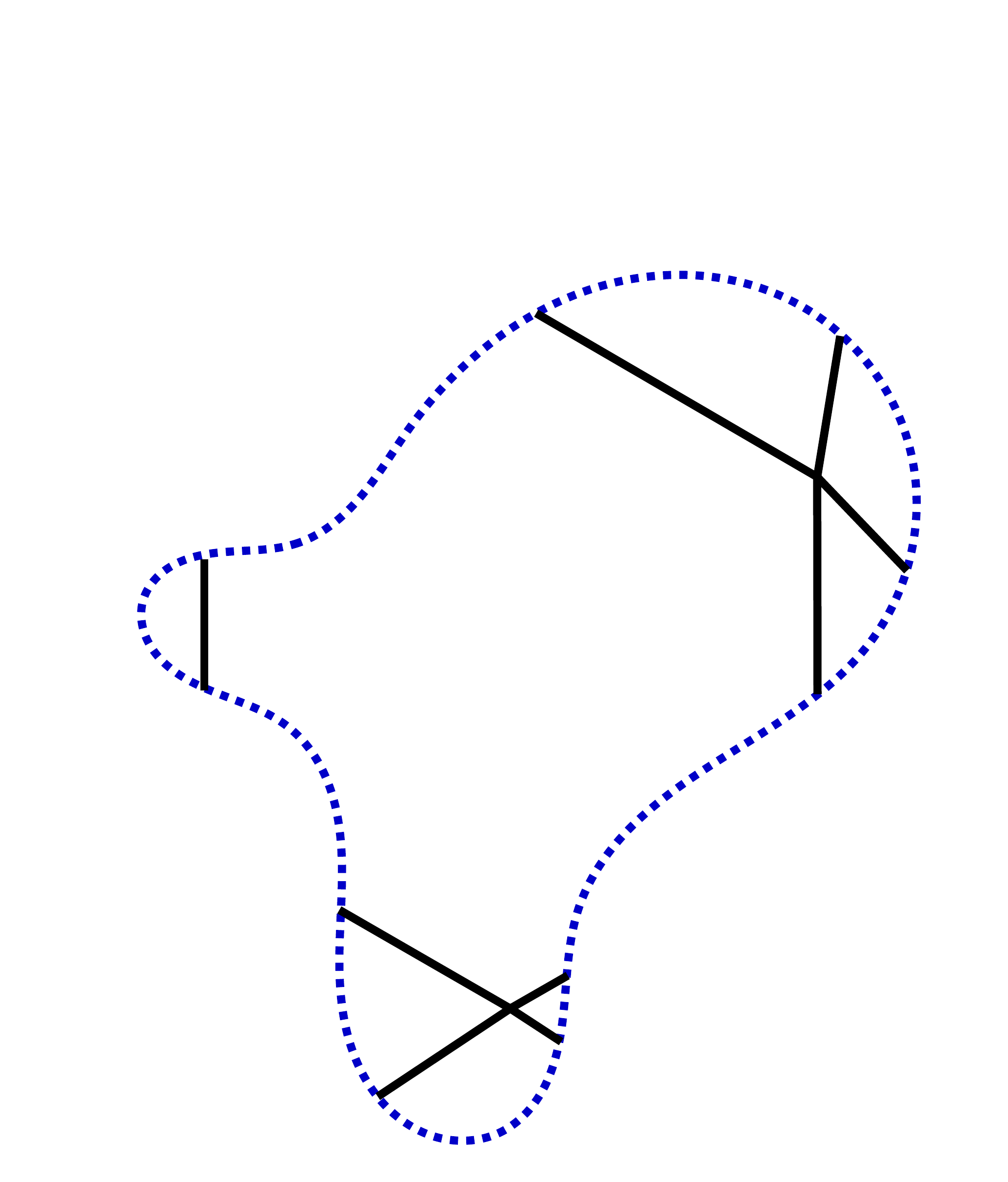}
    \caption{\label{fig:multicomponente}}
  \end{subfigure}

  \begin{subfigure}[t]{1.5in}
    \centering
    \includegraphics[scale=0.25]{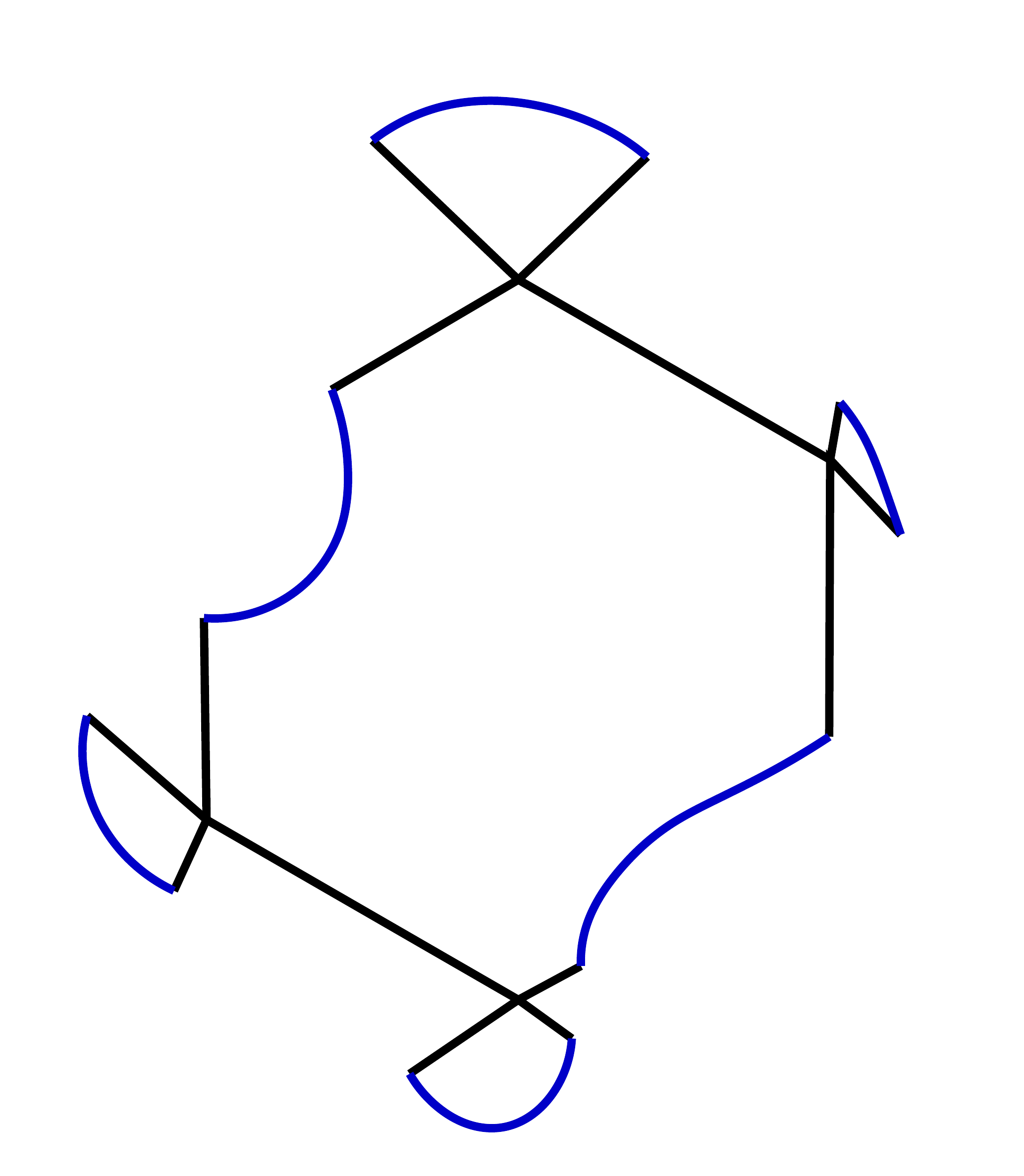}
    \caption{\label{fig:multicomponentc}}
  \end{subfigure}
  \hspace*{1in}
  \begin{subfigure}[t]{1.5in}
    \centering
    \includegraphics[scale=0.25]{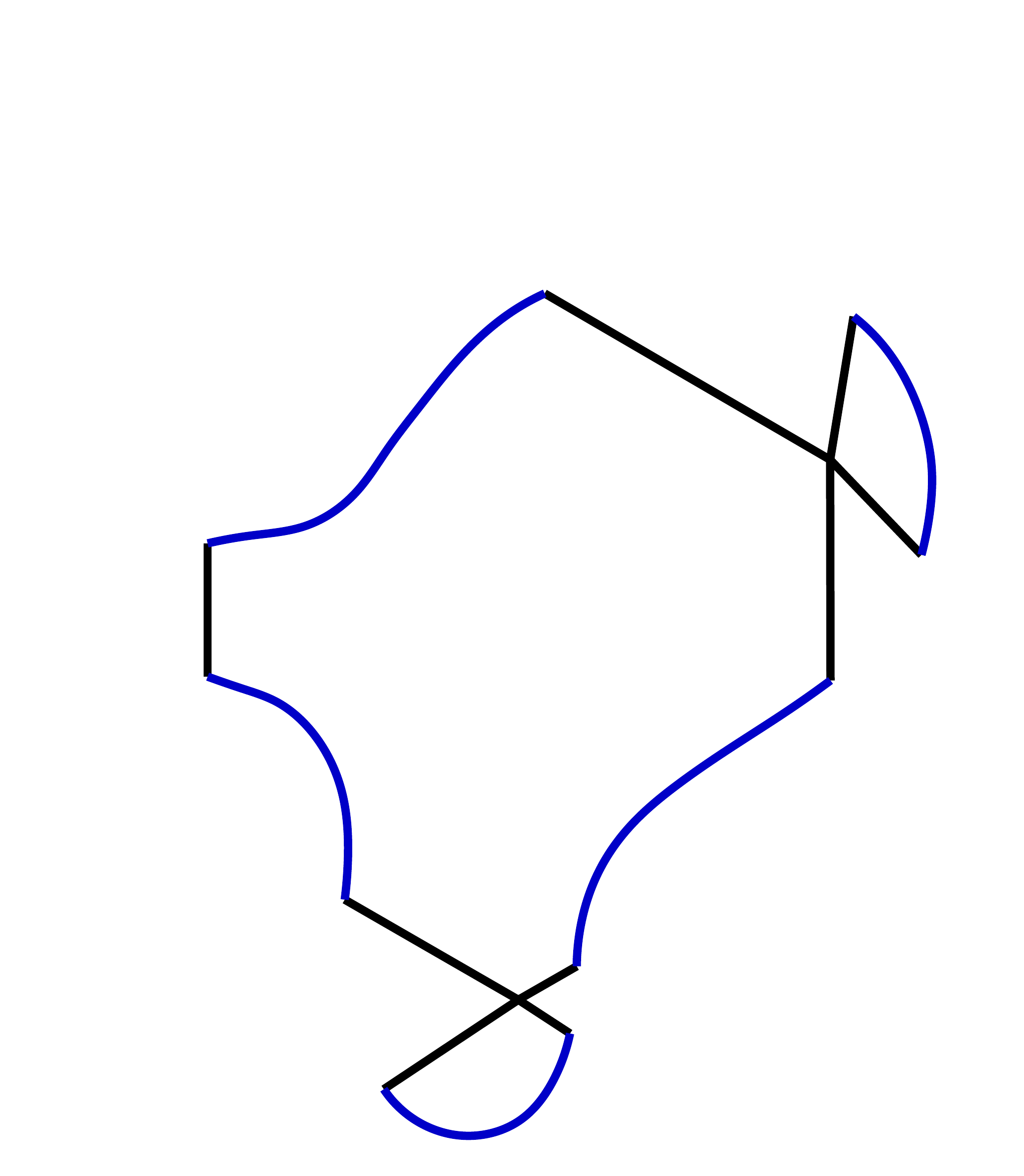}
    \caption{\label{fig:multicomponentf}}
  \end{subfigure}

  \caption{\label{fig:multicomponent}
    Figures (a) and (b) show a $2$-cell (black) of $\hK$.
    Figures (c) and (d) show multiple components after Clip operations on $\hK$ with respective polygons (dotted blue).
    Figures (e) and (f) show subsequent Patch operations connecting the multiple components with solid blue lines.
  }
\end{figure}

\subsection{Continuous Tool Path Planning Framework: Steps} \label{ssec:steps}

Our framework for continuous tool path planning consists of the following steps.
\begin{enumerate}
  \item {\bf{\textit{Slicing:}}}
    Slice an STL file of the design.
    This step creates a sequence of layers, and each layer can have multiple polygons.
    Let $\Ps_i = \{P_{ij}\}$ be the set of all the polygons in layer $i$ with or without holes.
    We assume the layers generated by slicing are $\epsilon$-continuous.
    
  \item{\bf{\textit{Projecting:}}}
    Project all polygons $\{P_{ij}\}$ in each layer $\Ps_i$ on to the horizontal plane.
    Take the {\bfseries union} of all projected polygons (from all layers).
    This union can have an irregular shape depending on the input.
    Let $P$ be the convex hull of the union of projected polygons.
    Note that taking the convex hull will avoid irregularities.
    We are assuming the input design has a single component.
    If not, we can repeat the procedure for each component.

  \item{\bf{\textit{Meshing:}}}
    Mesh $P$ with a pure $2$-complex $K$.
    We assume $K$ satisfies \cref{asmn:Kholesoutside}.
    
  \item \label{itm:eulertransform}{\bf{\textit{Euler Transformation:}}}
    Create $\hK$ by Euler Transformation on $K$.
    We assume that the traversal of edges in $\hK$ will not be affected by extruder size considerations (see \cref{sec:boundaryedges}). 
    
  \item\label{itm:patch}{\bf{\textit{Clipping:}}} 
    $\tK$ is a $2$-complex contained in $\tR_{ij}$.
    It is generated by Clip (\cref{def:clip}) and Patch (\cref{def:patch}) operations on $\hK$ with respect to $\tR_{ij}$ (see \cref{fig:printingprocess}).

    \begin{figure}[htp!] 
      \centering
      \begin{subfigure}[t]{2in}
        \centering
        \vspace*{-2.7in}
        \includegraphics[scale=0.22]{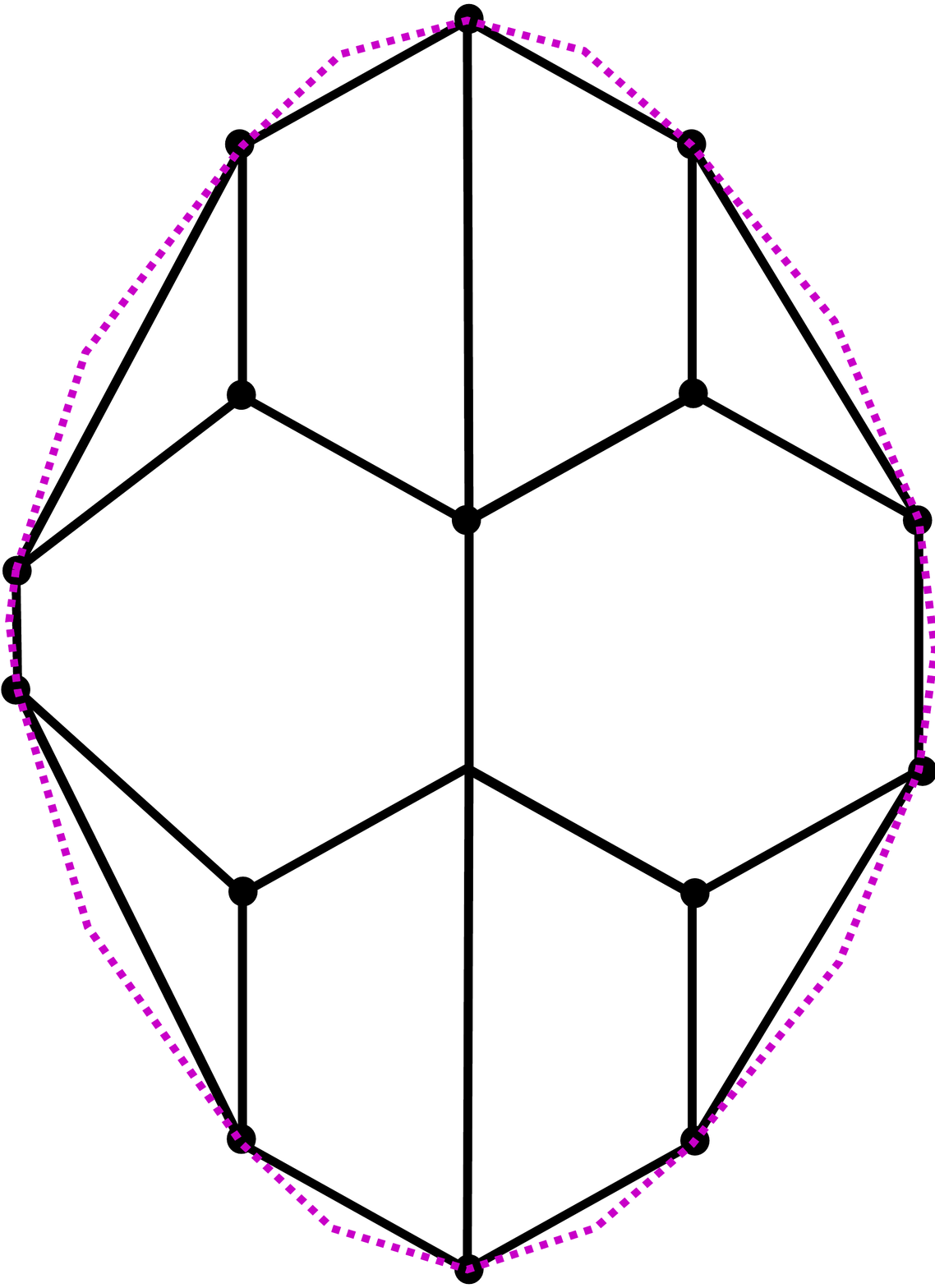}
        \vspace*{0.4in}
        \caption{\label{fig:printingprocessa}}
      \end{subfigure}
      \hspace*{0.2in}
      \begin{subfigure}[t]{2in}
        \centering
        \vspace*{-2.7in}
        \includegraphics[scale=0.22]{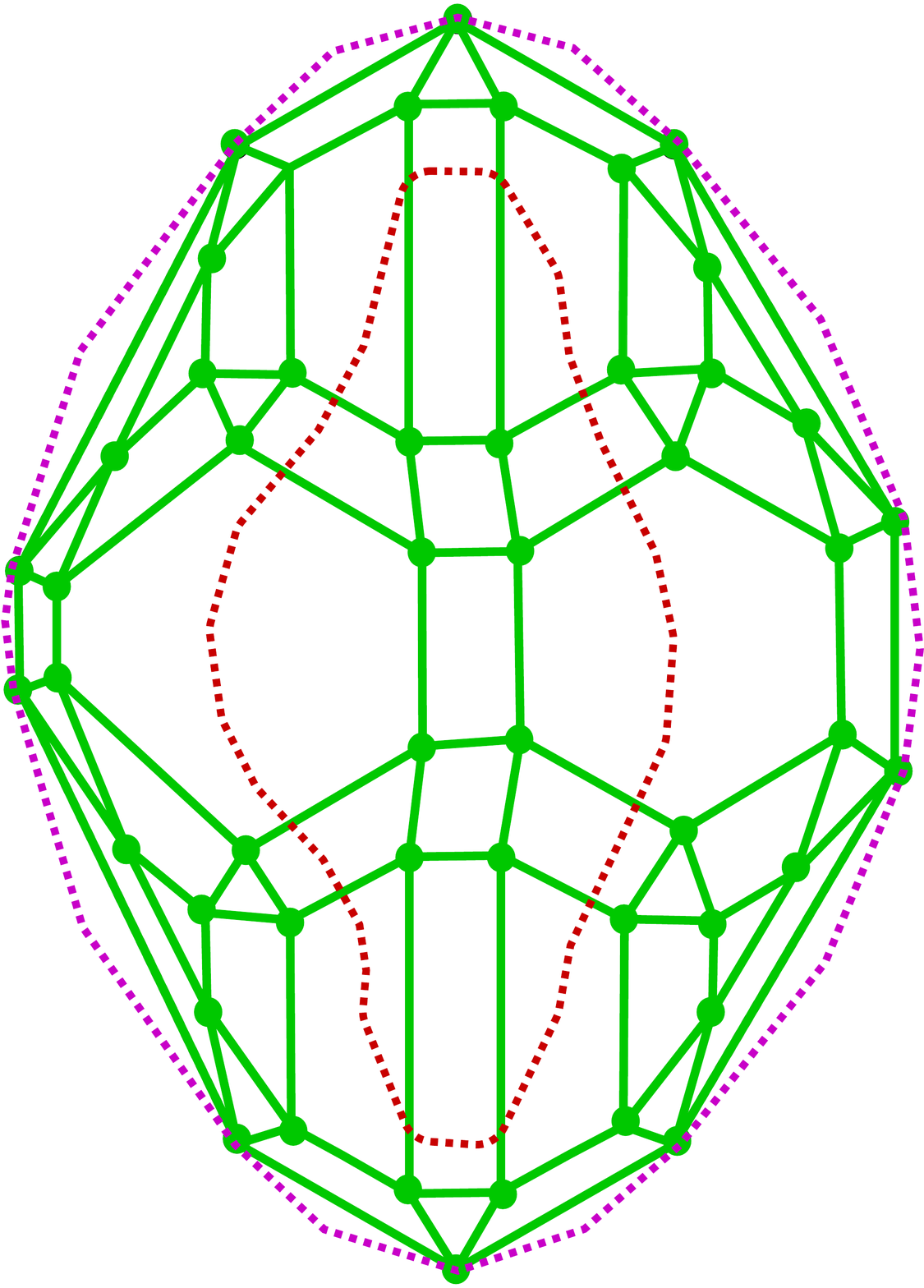}
        \vspace*{0.4in}
        \caption{\label{fig:printingprocessb}}
      \end{subfigure}
      \hspace*{0.2in}
      \begin{subfigure}[t]{1.5in}
        \centering
        \includegraphics[scale=0.4]{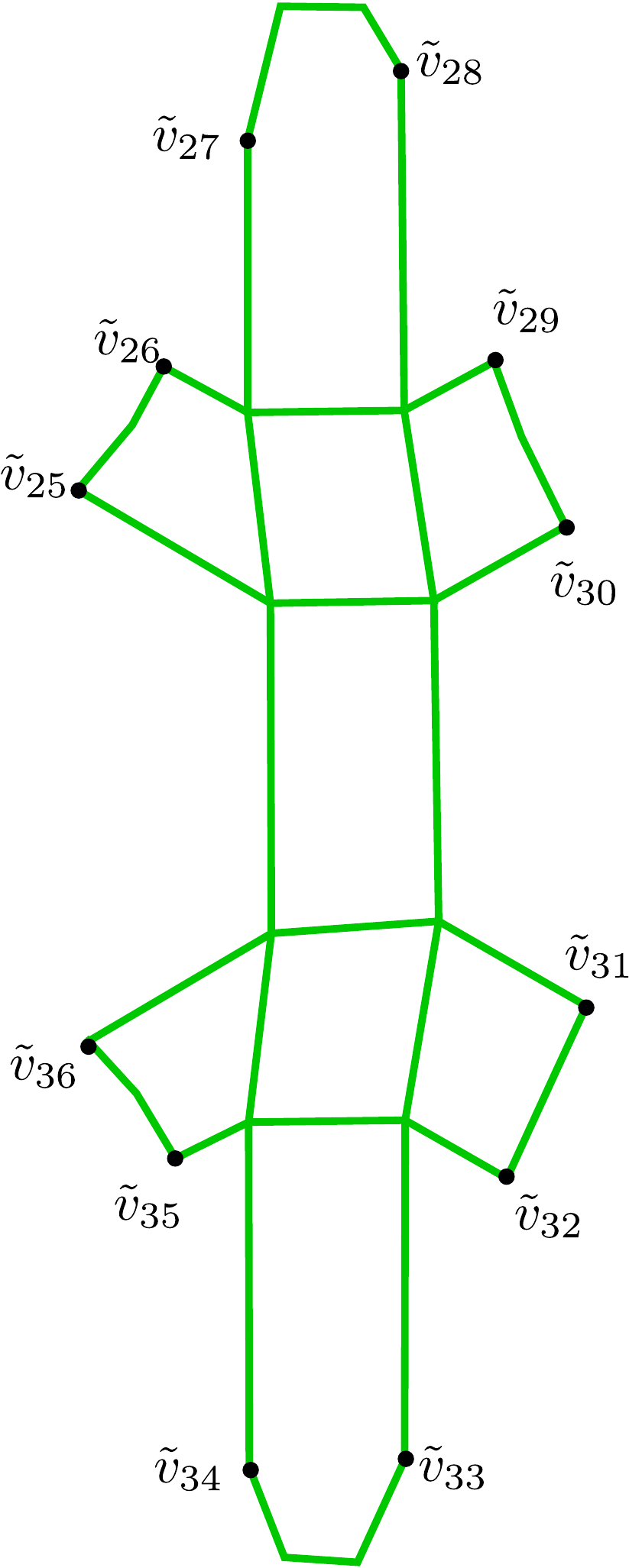}
        \caption{\label{fig:printingprocessc}}
      \end{subfigure}
      \caption{\label{fig:printingprocess}
        (a) Projected polygon $P$ (purple) and initial $2$-complex $K$ (black).
        (b) Euler transformation $\hK$ (green) and Region $\tR_{ij}$ (red dots) of polygon $P_{ij}$.
        (c) $\hK$ is clipped using $\tR_{ij}$ and patched to $2$-complex $\tK$ (green; enlarged for better visibility) after Clipping Step \ref{itm:patch}, where $\tR_{ij}$ has  $\{\tv_1, \dots, \tv_{24}\}$ sequence of vertices and $\{\tv_{25}, ..... , \tv_{36}\}$ are points of intersection of $\tR_{ij}$ and $1$-skeleton of $\hK$.
      }
    \end{figure}

    By \cref{lem:singlecomponent}, we get that $\tK$ is connected and its $1$-skeleton is Euler since $\hK$ has a single component and its $1$-skeleton is Euler. 
    There are two possible choices of pairing alternate vertices for patch operation.
    We can choose either option when $\tK$ is a single component.
    If $\tK$ has multiple components, either option leaves all resulting components Euler.
    
    Printing the infill lattice in each layer amounts to printing edges in the $1$-skeleton of $\tK$.
    We clip the $2$-subcomplex $\tK$ for each layer from $\hK$ to prevent printing in free space.
    Each edge is supported by an edge in the layer below it except for boundary edges in $\tK$.
    We will add a support perimeter to support boundary edges in $\tK$, as discussed in the next step.

  \item\label{itm:support}
    {\bf{\textit{Support Perimeter:}}}
    Let $\tR_{ij} \in P_{ij}$ and $\tR_{i+1,j} \in P_{i+1,j}$ be such that $\tR_{i+1,j}$ is supported by $\tR_{ij}$ through $\epsilon$-continuity as shown in \cref{fig:step_6_reasoning}.
    There are two possible ways we can select alternate pairs of vertices for subcomplex $\tK$ to join in $\tR_{i+1,j}$: $\{\{12, 13\}, \{14, 15\}, \{16, 11\}\}$ or $\{\{11, 12\}, \{13, 14\},$ $\{15, 16\}\}$. 
    Then the edges $\{12, 13\}$ in the first case and $ \{11, 12\}, \{13, 14\}$ in second case are not supported if $\{2, 3\},\{4, 5\}, \{6, 7\}, \{8, 9\},$ and $\{10, 1\}$ are the vertices pairs selected for $\tR_{i+1,j}$.
    To solve this problem we need a way to print all the edges at boundary of the polygons.
    \begin{figure}[htp!] 
  	\centering
  	\includegraphics[scale=0.425]{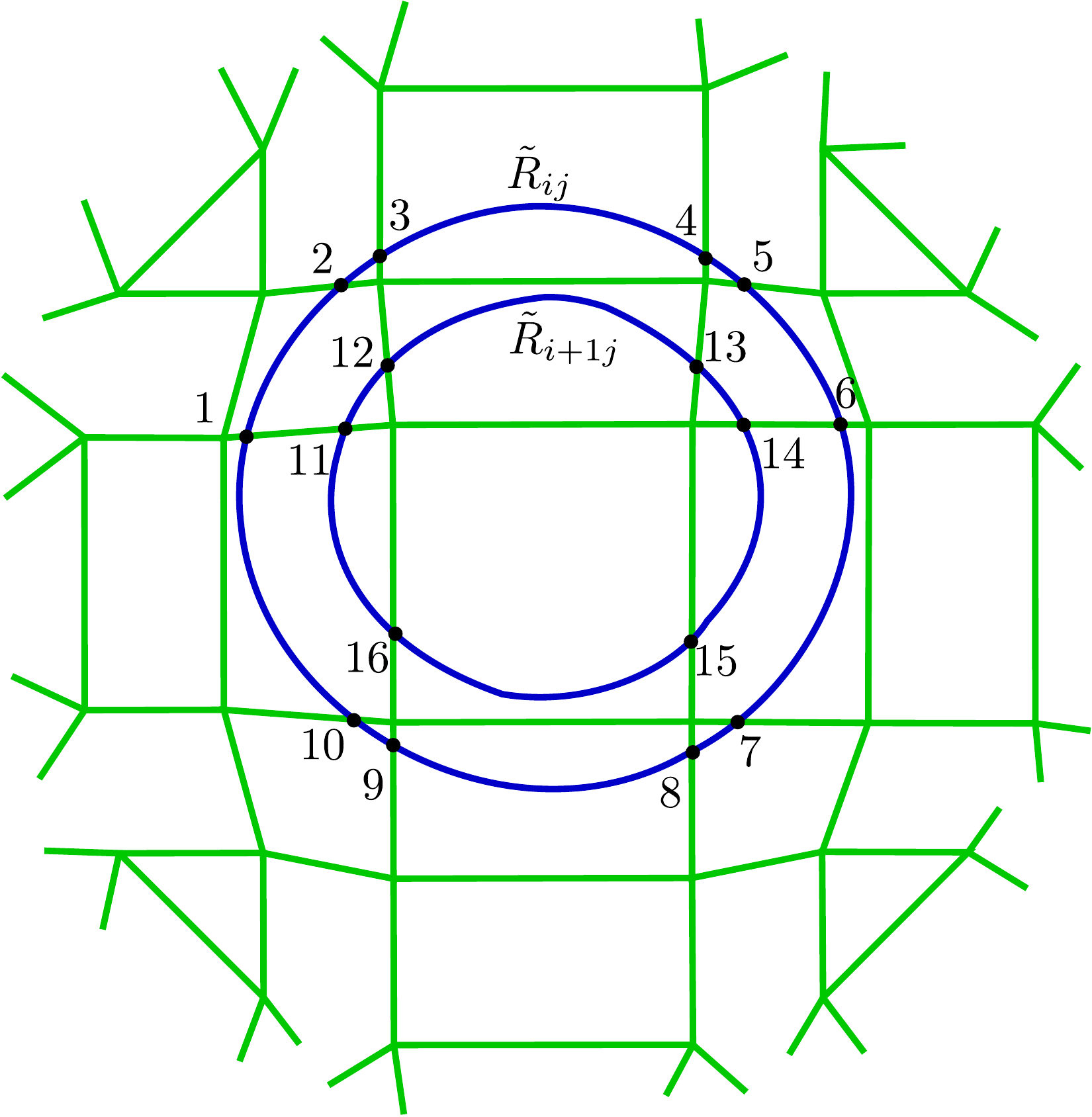}
  	\caption{\label{fig:step_6_reasoning}
          $\tR_{ij}$ (blue) in $P_i$, $\tR_{i+1j}$ (blue) in $P_{i + 1}$ intersect the complex $\hK$ (green) at points $\{1, \dots,10\}$ and $\{11, \dots, 16\}$, respectively.
        }
    \end{figure}  
  
    Since we are not printing all the edges at the boundary of $\tR_{ij}$, we could have some overhanging boundary edges in $\tK$.
    Let $\tP$ be the nonprinted path on $\tR_{ij}$ between $\tv_{n + l}$ and  $\tv_{n+l+1}$, which are vertices of $S$ on $\tR_{ij}$.
    Add circles of radius $r$, the extruder radius, on $\tv_{n+l}$ and $\tv_{n+l+1}$ and on path $\tP$ such that neighboring circles do not intersect.
    We assume the circles only intersect neighboring line segments on the path.
    Suppose $\eta$ is the maximum number of circles of radius $r$ that can be added on path $\tP$, assuming there are $2$ circles of radius $r$ centered at end points of the path.
    Add $\eta$ possible circles on path $\tP$, where the center of the $j^{\mbox{th}}$ circle is $\tv'_j$.
    The total gap we can have between the circles is $2r - \delta$, where $0 < \delta < 2r$ as shown in \cref{fig:supporta}.
    We can uniformly distribute the gap of $(2r - \delta)/(\eta + 1)$ between the circles as shown in \cref{fig:supportb} assuming $\eta$ is at least one.
    Since $\tR_{ij}$ is the inward Minkowski offset of $R_{ij}$ with a $2$-ball of radius $r$, for any $\tv'_j$ there exists a point $a$ such that line segment $\{\tv'_j, a\}$ is perpendicular to $\tR_{ij}$ and $R_{ij}$, and  $d(\tv'_j, a) = 2r$.
    Let $v_{n+k}$ and $v_{n+k+1}$ be points of intersection with $R_{ij}$ of the circle of radius $r$ centered at $a$.
    Create a corner by adding edges $\{v_{n+k}, a\}$, $\{v_{n +k +1}, a\}$ as shown in Figure \ref{fig:supportc}.
    Connect end points of the corners by a path on $R_{ij}$ if corner line segments do not intersect each other.
    Else change end points to the points of intersection so as to form a simple closed polygon as shown in \cref{fig:supportd}.
    We illustrate in \cref{fig:supportperimeter} the support perimeter around $\tK$ shown in \cref{fig:printingprocessc}.
    
 \begin{figure}[htp!]    
   \begin{subfigure}[t]{2.75in}
     \centering
     \includegraphics[scale=0.42]{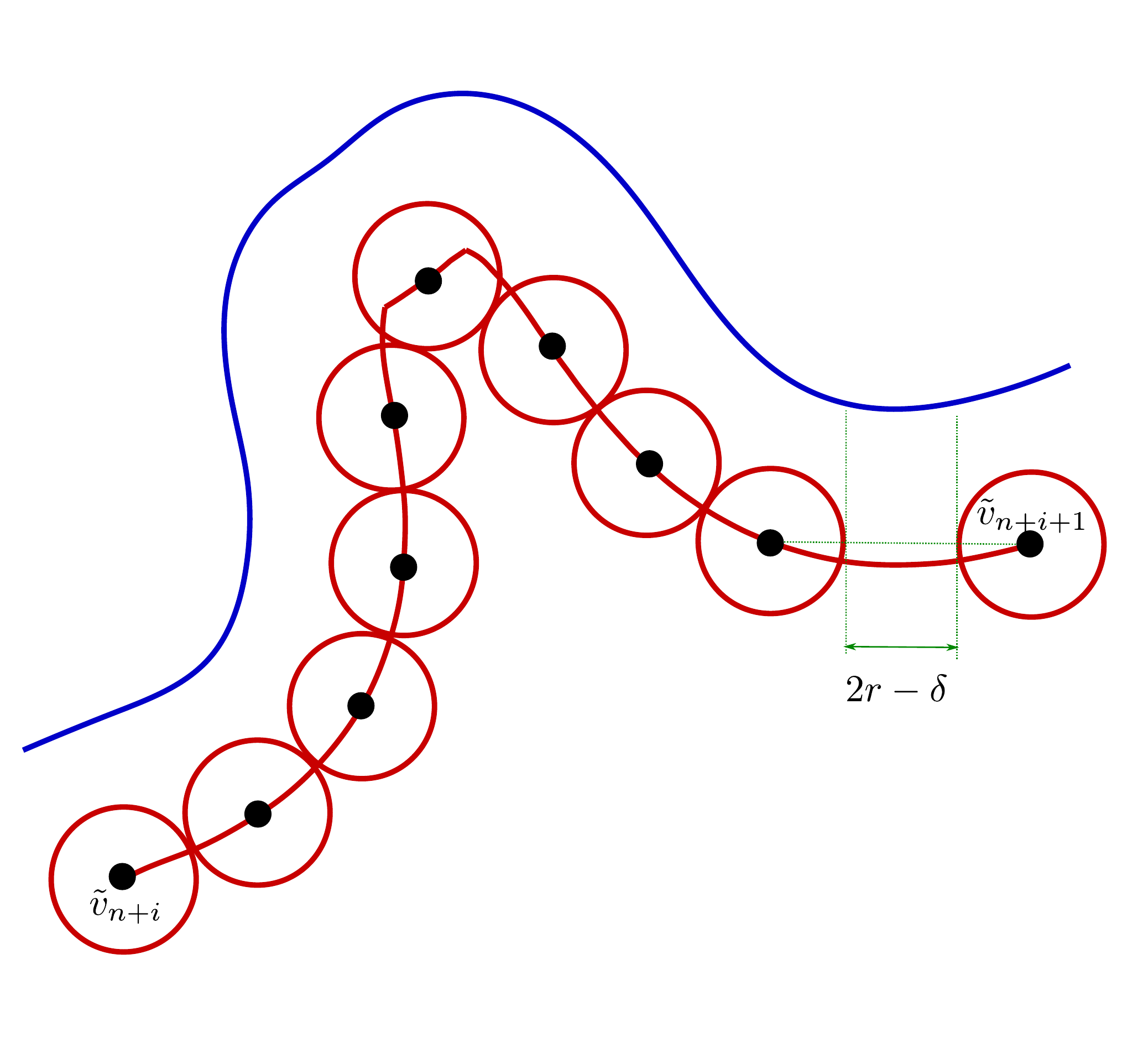}
     \vspace*{-0.1in}
     \caption{\label{fig:supporta}}
   \end{subfigure}
   \hspace*{0.3in}
   \begin{subfigure}[t]{2.75in}
     \centering
     \includegraphics[scale=0.42]{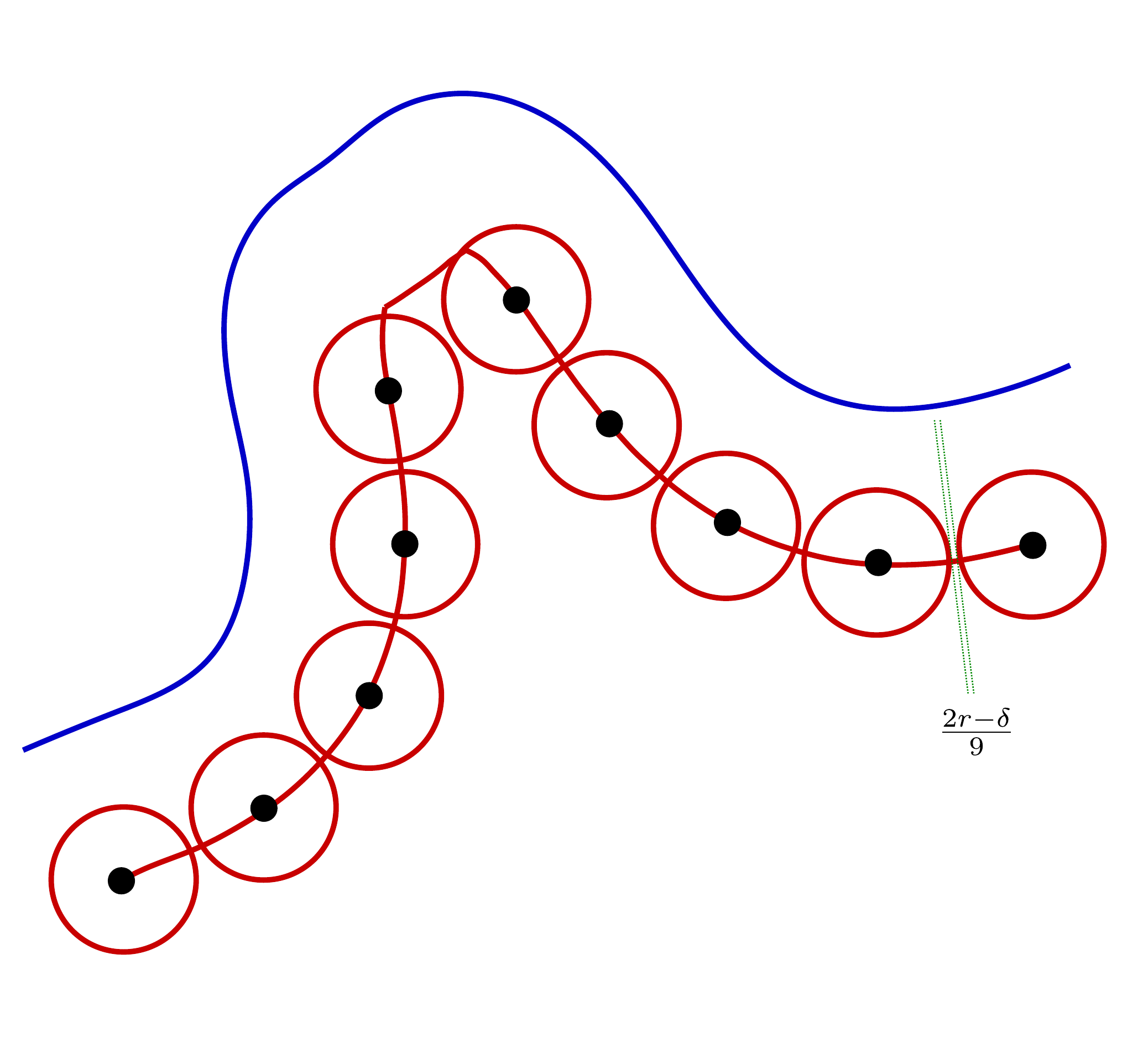}
     \vspace*{-0.1in}
     \caption{\label{fig:supportb}}
   \end{subfigure}
 
   \begin{subfigure}[t]{2.75in}
     \centering
     \includegraphics[scale=0.42]{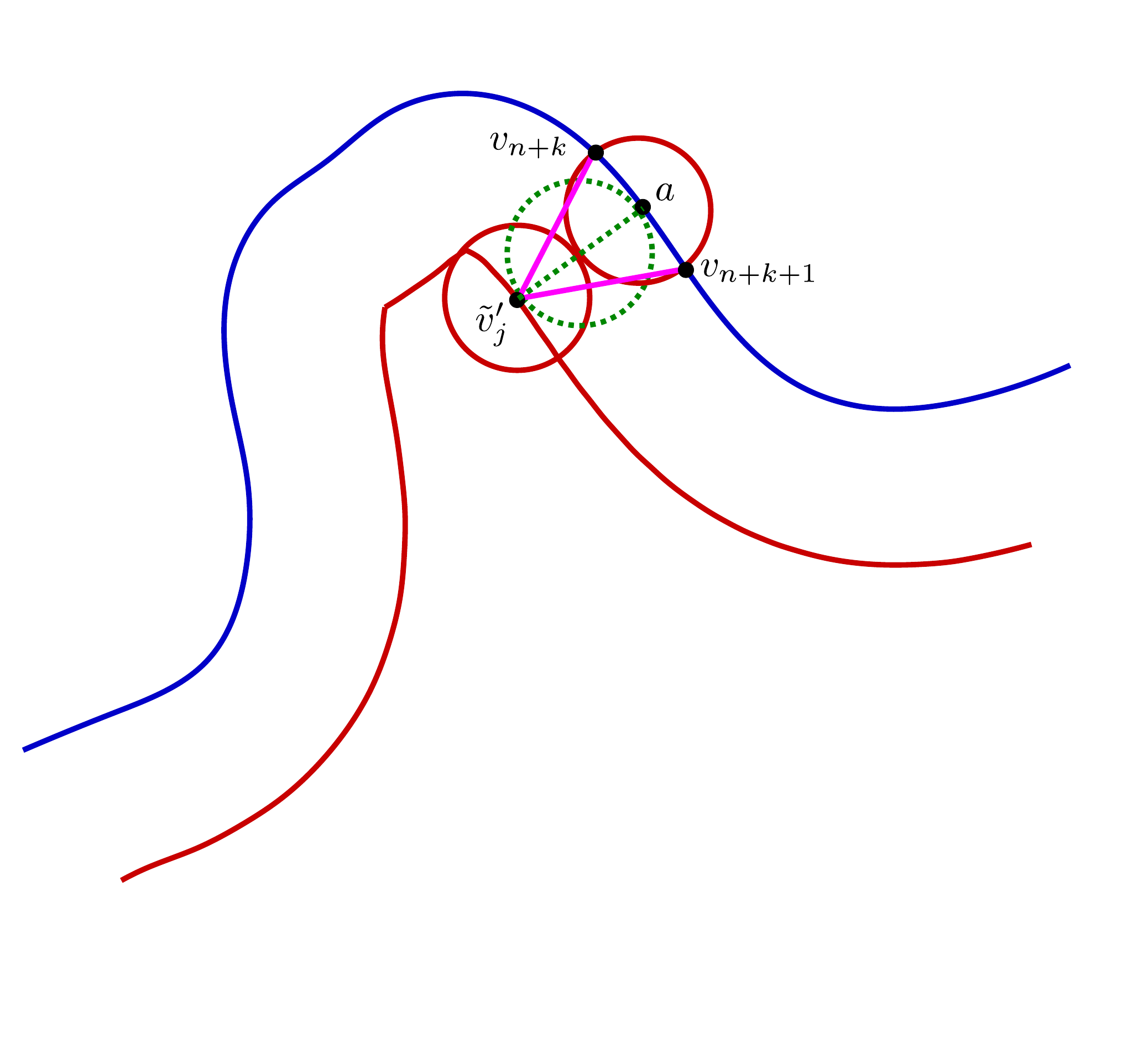}
     \vspace*{-0.1in}
     \caption{\label{fig:supportc}}
   \end{subfigure}
   \hspace*{0.3in}
   \begin{subfigure}[t]{2.75in}
     \centering
     \includegraphics[scale=0.42]{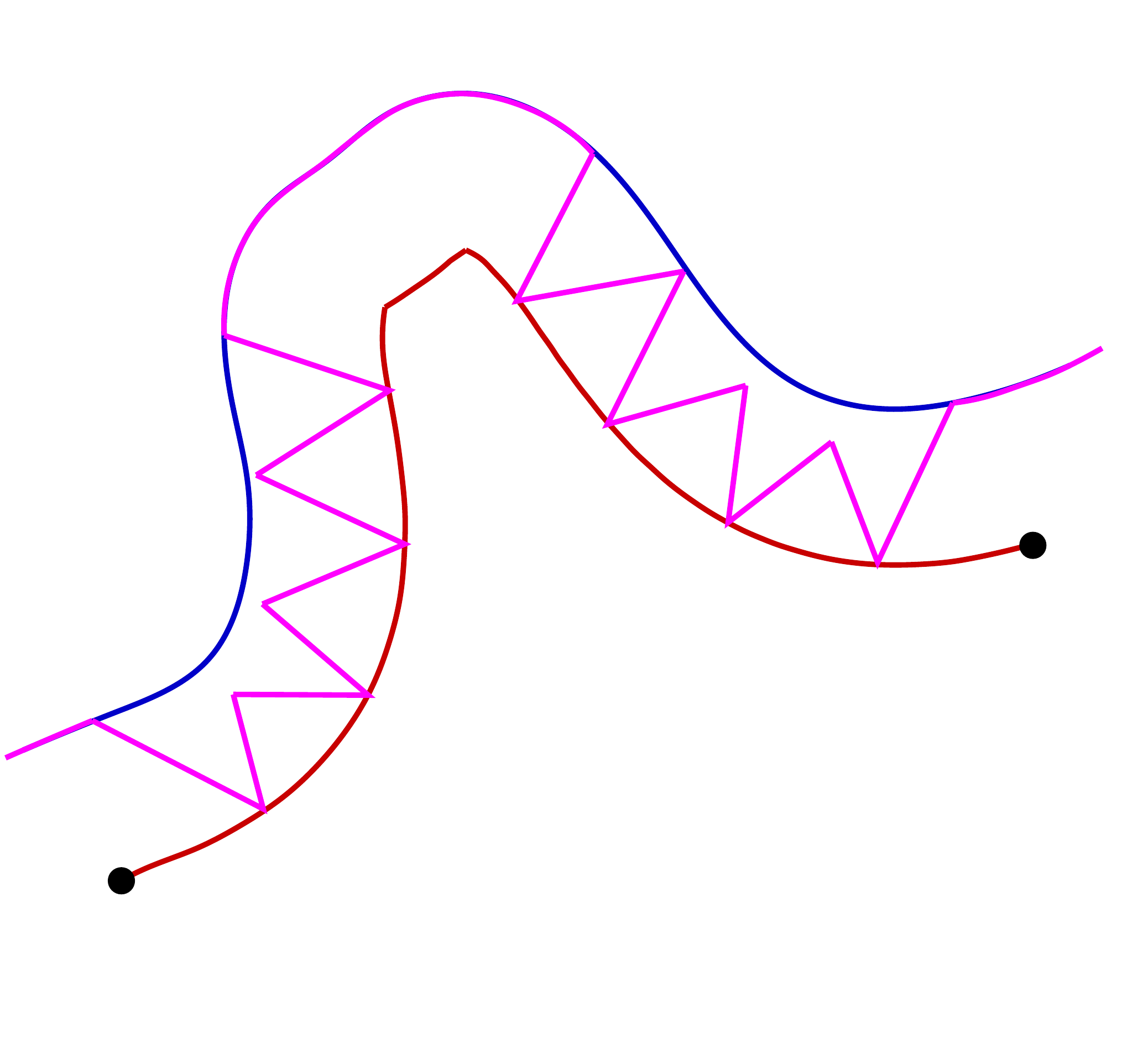}
     \vspace*{-0.1in}
     \caption{\label{fig:supportd}}
   \end{subfigure}

   \caption{\label{fig:support}
     (a) Portion of $R_{ij}$ (blue) and $\tR_{ij}$ (red), $\tP$ (red), with total gap between circles being $2r - \delta$.
     (b) Uniformly distribute the gap into $(2r - \delta)/9$ parts between neighboring circles.
     (c) $\{v'_j, a\}$ is perpendicular to $\tR_{ij}$ and $R_{ij}$, circle centered at $a$ intersects $R_{ij}$ at $v_{n+k}, v_{n+k+1}$, and corner (pink) after adding edges $\{v'_j, v_{n+k}\}, \{v'_j, v_{n + k + 1}\}$.
     (d) Neighboring corners joined (pink) to form a simple closed polygon.
     }
 \end{figure}    

 \begin{figure}[htp!] 
   \centering
   \includegraphics[scale=0.325]{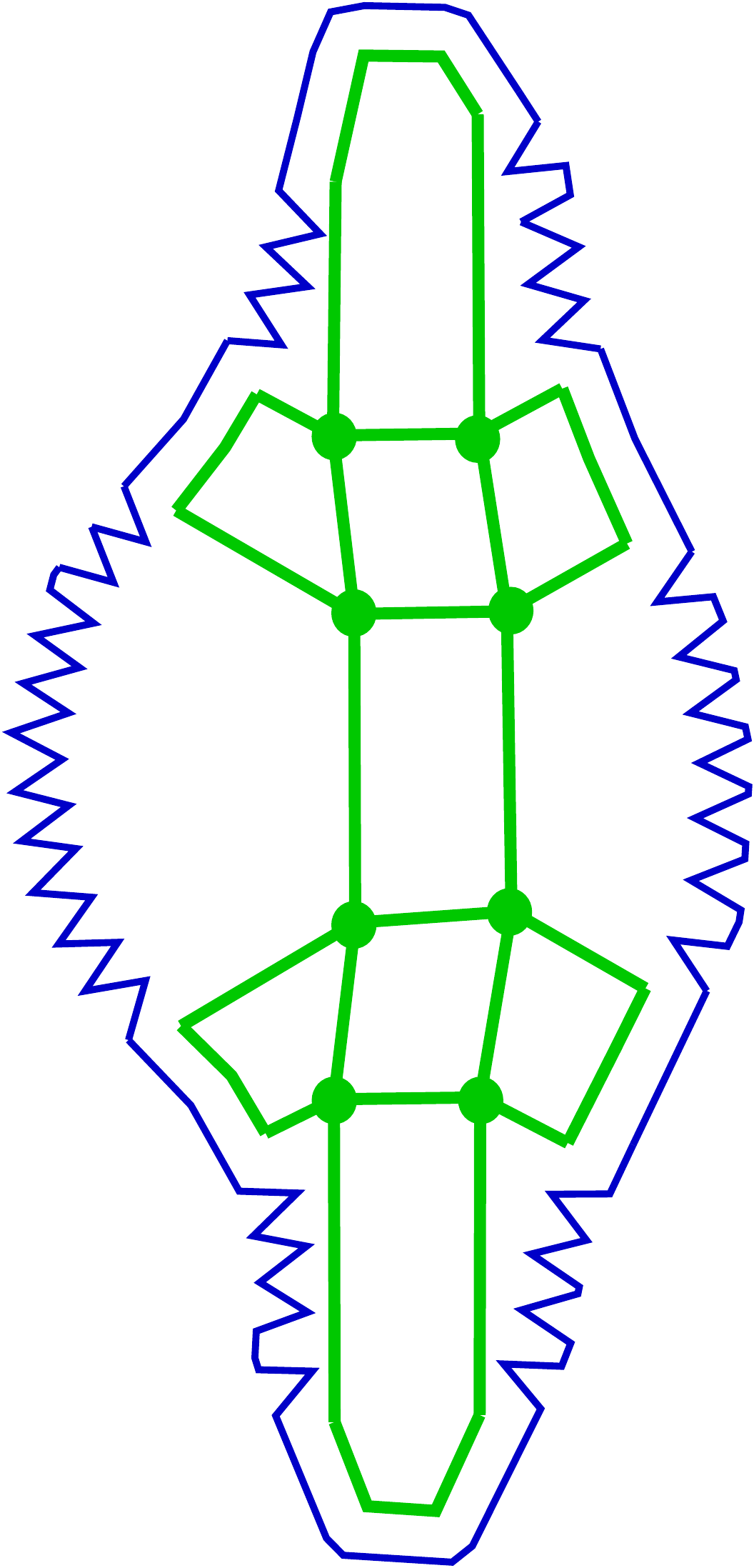}
   \caption{Support perimeter for $\tK$ shown in \cref{fig:printingprocessc} is shown in blue here. }
   \label{fig:supportperimeter}
 \end{figure}   	 
\end{enumerate}


\begin{rem}
  {\rm 
  In Step \ref{itm:support}, we add corners only if we can add circles of radius $r$ on path $\tP$, given there are circles of radius $r$ at the end point of the path.
  It is not guaranteed that line segments of $\tP$ will be covered by the support, since the coverage depends on the curvature of $\tP$ as shown in Figure \ref{fig:supportd}.
  An alternative approach to printing the support is to print the individual non-printed sections in $\tR_{ij}$ while making non-print travel moves in between.
  But this approach will have $m/2$ starts and stops.
  If $\hK$ is a highly dense $2$-complex, then $m$ will be large and we will have a large number of starts and stops in this case.
  }
\end{rem}

\begin{rem}
  {\rm
    A continuous tool path exists only if any polygon in $\tK$ is shrinkable (\cref{def:shrinkable}) with no topological changes.
    But boundary polygons in $\tK$ due to clip and patch (\cref{def:clip,def:patch}) operations can be shrinkable with topological changes or unshrinkable as mentioned in \cref{sec:boundaryedges}.
  }
\end{rem}

\section{Tool Path Algorithm}\label{sec:toolpathalgo}
Since $\tG$, the $1$-skeleton of $\tK$, is Euler, we can construct a tool path that consists only of the print path, i.e., all edges with none of them repeated.
In general, such a tour can cross over itself at a vertex, creating a special case of material collision termed \textit{crossover}.  
We present a tool path algorithm that chooses the subcycles in an Eulerian tour of $\tG$ carefully so as to avoid all crossovers.
First we construct a circuit tree that represents $\tG$, with the vertices of the tree representing edge-disjoint circuits in $\tG$.
Second, we add edge traversal restrictions in order to avoid crossovers in the tool path, which is described by specifying a traversal order of the circuit tree. 
\subsection{Circuit Tree}

\begin{algorithm}[htp!] 
	\caption{CircuitTree}
	\label{alg:circuittree}
	\begin{algorithmic}[1]
		\State Unmark all the edges in $\tK$
		\State CircuitList =  FindBoundaryCircuits$(\tK$, $\emptyset)$
		\Comment{\textit{Initial } CircuitList \textit{contains one circuit} $C_{\rm Init}$}
		\State Mark all the edges in $\tK$ of $C_{\rm Init}$
		\State pred$(C_{\rm Init}) = \emptyset$ \Comment{\textit{Predecessor of} $C_0$ \textit{is empty}}
		\While{CircuitList $\neq \emptyset$}
		\State $C$ = CircuitList.pop$(0)$ 
		\State Circuits = FindBoundaryCircuits$(\tK$, $C)$ \Comment{\textit{Returns list of circuits}}
		\State CircuitList = CircuitList $\cup$ Circuits
		
		\While{Circuits $\neq \emptyset$}
		
		\State $C_s$ = Circuits.pop$(0)$ 
		\State pred$(C_s)$= $C$ \Comment{\textit{predecessor of $C_s$ is $C$}}
		\State Mark all the edges in $\tK$ of $C_s$
		
		\EndWhile
		  	
		\EndWhile
	\end{algorithmic}
\end{algorithm}

\begin{algorithm}[htp!] 
	\caption{FindBoundaryCircuits$(\tK,C_0)$}
	\label{alg:findboundarycircuit}
	\begin{algorithmic}[1]
	\State Find collection of edges $H$ based on $C_0$ \Comment{\textit{edges of cells intersecting $C_0$ not in $C_0$}}
	\State Circuits = \{\}	
     \While{$H \neq \emptyset $}
     
     \State Find circuit $C$ using Modified Hierholzer's algorithm
     \State Remove all edges of $C$ from $H$
     \State Circuits = Circuits $\cup$ \{C\}
     
     \EndWhile			
	
	\State \Return Circuits
	\end{algorithmic}
	
\end{algorithm}

Algorithm \ref{alg:circuittree} constructs the circuit tree.
It works by finding the outermost circuit in $\tK$ first, and continues to find all inner circuits in $\tK$, finishing with the innermost circuits.
The outermost circuit in the $1$-skeleton of $\tK$ consists of boundary edges in $\tK$.
An innermost circuit in the $1$-skeleton of $\tK$ contains no other circuit in its underlying space in $\tK$.  
The outermost circuit corresponds to the root node in the circuit tree, while innermost circuits are leaf nodes.
The union of all circuits in the circuit tree is $\tK$.
Every node in the Circuit tree represents a circuit.
The predecessor of a circuit $C$ (pred$(C)$) is another circuit connected to $C$ at one or more vertices.
All interior vertices in $\tK$ have even degree with at least degree $4$ (some vertices may have even degree at least $6$ due to applications of local Euler Transformation).
Some vertices at the boundary of $\tK$ may have degree $2$ as mentioned in step \ref{itm:patch} of Section \ref{subsec:clipping}.
But there is at least one boundary vertex of $\tK$ with degree at least $4$, if there is more than one node in the circuit tree.
This implies every circuit $C$ can have at most $d$ consecutive descendants in a path in the circuit tree, where $2d$ is the maximum degree of a vertex in $C$.
Based on its construction, any two circuits in the circuit tree are disjoint if they are not on the same path starting at the root node in the circuit tree.
An example circuit tree is shown in Figure \ref{fig:circuittree}.
%

Given a circuit $C_0$, the algorithm finds inner circuits as follows.
Suppose $\tf$ is a $2$-cell in $\tK$ sharing edges $\{\te_i\}$ in $C_0$ and none of its edges are marked in $\tK$ except $\{\te_i\}$.
Then all other edges of $\tf$ except $\{\te_i\}$ are part of successor circuits in the circuit tree.
Let $H$ be the collection of all the edges in all $2$-cells of the form $\tf$, except edges in $C_0$ (see Algorithm \ref{alg:findboundarycircuit}).
Then $H$ represents the next ``onion layer'' of boundary circuits.
If $C_0$ is empty, then $H$ is the  collection of all the boundary edges in $\tK$.

\begin{figure}[hbp!] 
  \centering
  \vspace*{-0.1in}
  \includegraphics[scale=0.35]{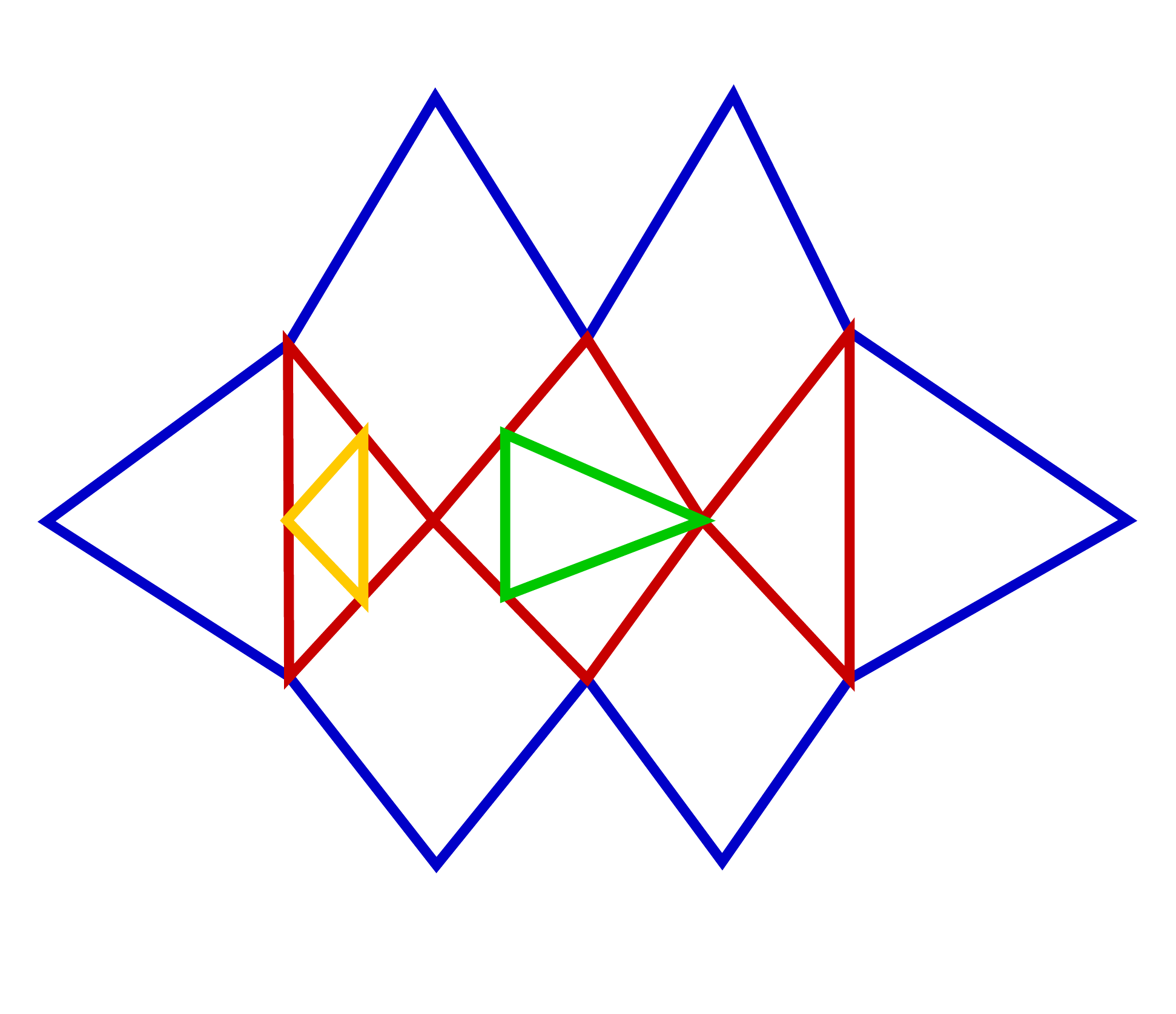}
  \hspace*{0.4in}
  \includegraphics[scale=0.32]{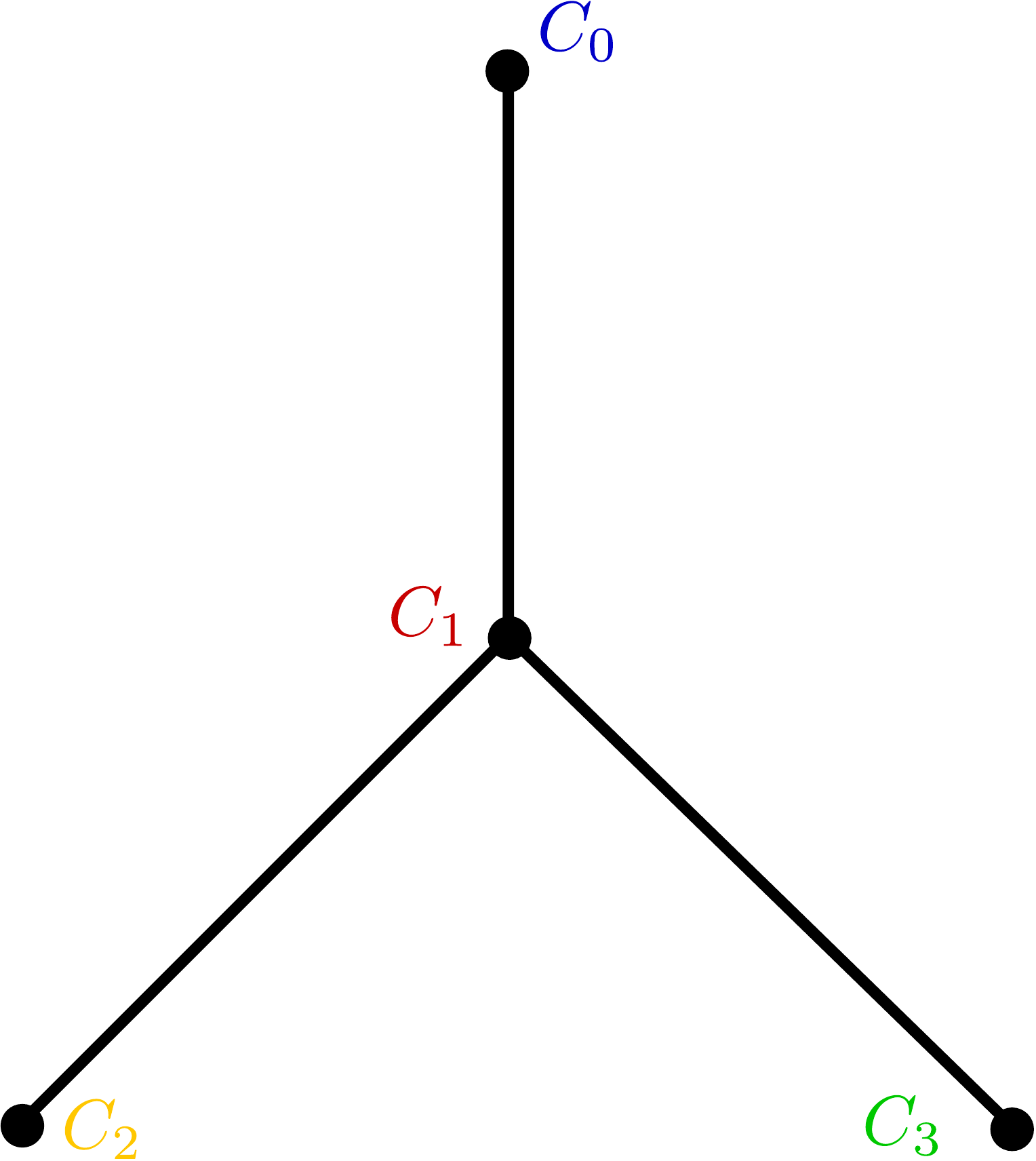}
  \caption{
    $\tK$ (left) has $C_0$ (blue), $C_1$ (red), $C_2$ (yellow), $C_3$ (green) circuits.
    In the circuit tree (right), child circuits $C_2, C_3$  of $C_1$ are disjoint.
  }
  \label{fig:circuittree}
\end{figure}

\noindent \textit{Modified Hierholzer's algorithm}: The original Hierholzer's algorithm \cite{HiWi1873} was designed for a connected graph, but in our case $H$ can have multiple components.
We want to orient the circuits in order to identify the tool path avoiding crossovers.
We assume without loss of generality that each circuit is oriented clockwise as shown in Figure \ref{fig:clockwiseeulercircuit}.
Hence subtours of this circuit are oriented clockwise as well.
Pick a vertex in $H$, find all connected subtours $\{S_j\}$ and join them to obtain a circuit.
Delete all the edges and vertices in this circuit from $H$.
Repeat the process until $H$ is empty.

\begin{figure}[htp!] 
  \centering
  \includegraphics[scale=0.3]{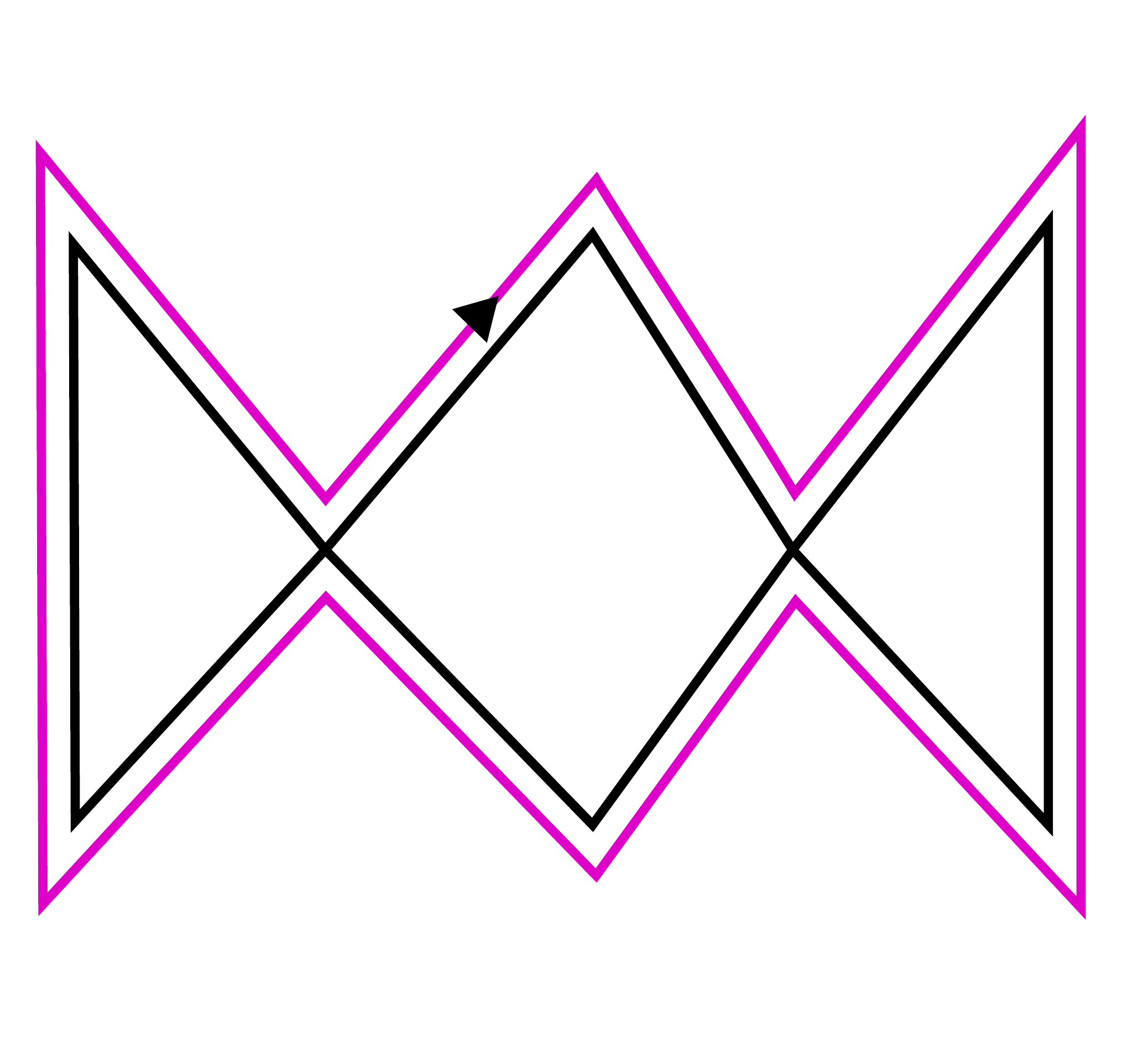}
  \caption{Traversal of $C_1$ (black) from circuit tree in \cref{fig:circuittree} in clockwise orientation with no crossovers is shown in pink.} 
  \label{fig:clockwiseeulercircuit}
\end{figure} 

\emph{Correctness:} Since the $1$-skeleton of $\tK$ is Euler, $H$ consists only of circuits, and hence Algorithm \ref{alg:findboundarycircuit} is guaranteed to terminate.
It runs in $O(|E|)$ time, where $E$ is set of edges in $\tK$.

\medskip
\textit{Complexity:} Identification of $H$, the collection of boundary edges in $\tK$, takes $O(|E|)$ time.
For a given $H$, the modified Hierholzer's algorithm runs in $O(|E|)$ time.
and we can have at most $|E|/3$ iterations of the outermost {\bfseries while} loop in Algorithm \ref{alg:circuittree}.
Hence the circuit tree algorithm runs in $O(|E|^2)$ time.

\medskip
\subsection{Traversal}\label{subsec:traversal}

Recall that we assume all circuits are oriented clockwise.
To identify the tool path traversal that avoids crossovers, we traverse edges along the circuits in the circuit tree such that each parent and child circuit pair is traversed in opposite orientation.
Let $\tv$ be a vertex in the circuit $C_1$ that is shared by two or more circuits in the tree, and has a degree $2d$.
Then there are $q \leq d$ circuits sharing vertex $\tv$.
Further, there exists a path $P = C_1 \to C_2 \to \dots \to C_q$ in the circuit tree sharing vertex $\tv$ where $C_1$ is the ancestor and $C_q$ is the descendant of all circuits in $P$.
Orientations for all other circuits in the tree are uniquely determined, if orientation of root circuit is fixed.
We assume without loss of generality that the traversal of edges in the root circuit in $\tK$ is clockwise.

There are two types of possible subpath crossovers while traversing the circuit tree (our goal is to avoid all crossovers).
The first type of crossover is within a circuit (\textit{type-1}) and the second type of crossover occurs while traversing from a parent to a child circuit (\textit{type-2}).  
If any circuit $C_i$  of the circuit tree is traversed along its orientation (as shown in Figure \ref{fig:clockwiseeulercircuit}, for instance), then it is guaranteed that there is no subpath crossover of (\textit{type-1}).
To prevent subpath crossovers of \textit{type-2}, traverse edges of $C_i$ and pred$(C_i)$ in the circuit tree in \emph{opposite} orientations along with certain edge traversal restrictions.
Thus we traverse edges of $C_q$ clockwise when $q$ is odd, else counterclockwise.
The tool path traversal steps are detailed below.

\begin{enumerate}
  \item\label{itm:traversrestrict}
   Let $(\te_{2j-1}, \te_{2j})$ be a pair of clockwise ordered adjacent edges on a clockwise circuit path of $C_j$, sharing vertex $\tv$.
   Let $\te_i \longrightarrow \te_j$ imply we traverse $\te_j$ immediately after we traverse $\te_i$ in $\tK$.
   Add the following edge traversal restrictions in $\tK$ if $q$ is odd: $\te_1 \longrightarrow \te_{2q}, \te_{2q-1} \longrightarrow \te_{2q -3}, \te_{2q-2} \longrightarrow \te_{2q -4}, \dots,  \te_{5} \longrightarrow \te_{3},\te_{4} \longrightarrow \te_{2}$, where $\te_1 \longrightarrow \te_{2q}$ means we traverse edge $\te_1$ of cycle $C_1$ followed by $\te_{2q}$ of cycle $C_q$.
   Similarly, $\te_{2(q-i)-1} \longrightarrow \te_{2(q-i-1)-1}$ means we traverse edge $\te_{2(q-i)-1}$ of $C_{q-i}$ followed by $\te_{2(q-i-1)-1}$ of $C_{q-i-1}$ and  $\te_{2(q-i)} \longrightarrow \te_{2(q-i-1)}$  means we traverse edge $\te_{2(q-i)}$ of $C_{q-i}$ followed by $\te_{2(q-i-1)}$ of $C_{q-i-1}$.  
   If $q$ is even, we add the restrictions $\te_1 \longrightarrow \te_{2q -1},\te_{2q} \longrightarrow \te_{2q -2}, \te_{2q-3} \longrightarrow \te_{2q -5}, \dots, \te_{5} \longrightarrow \te_{3}, \te_{4} \longrightarrow \te_{2}$.
   Mark all the edges of the circuit tree on path $P$.
   An example with $q=4$ is shown in Figure \ref{fig:edgerestrictq34}.     
   \begin{figure}[hbp!]
     \centering
     \includegraphics[scale=0.35]{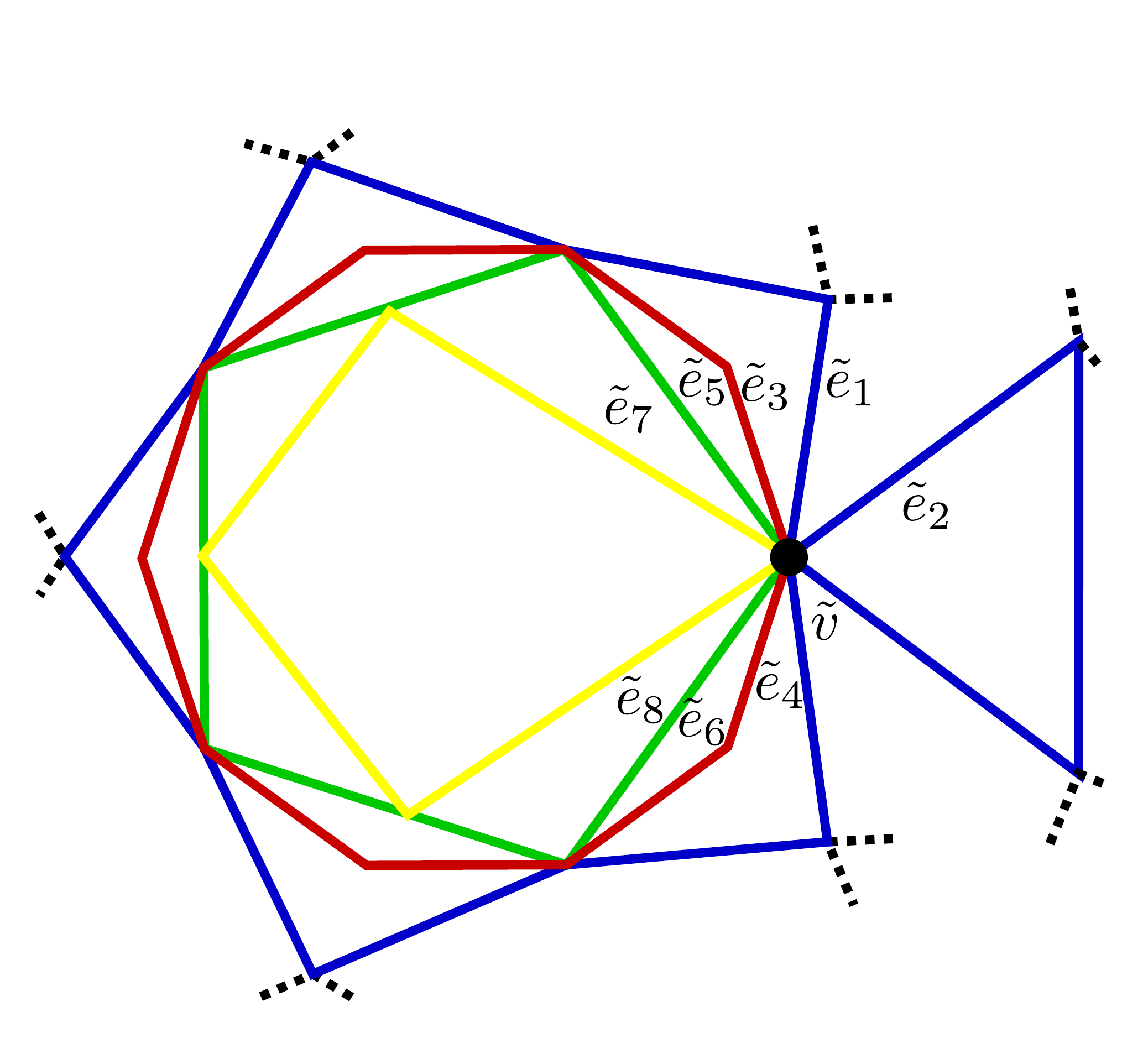}
     \includegraphics[scale=0.35]{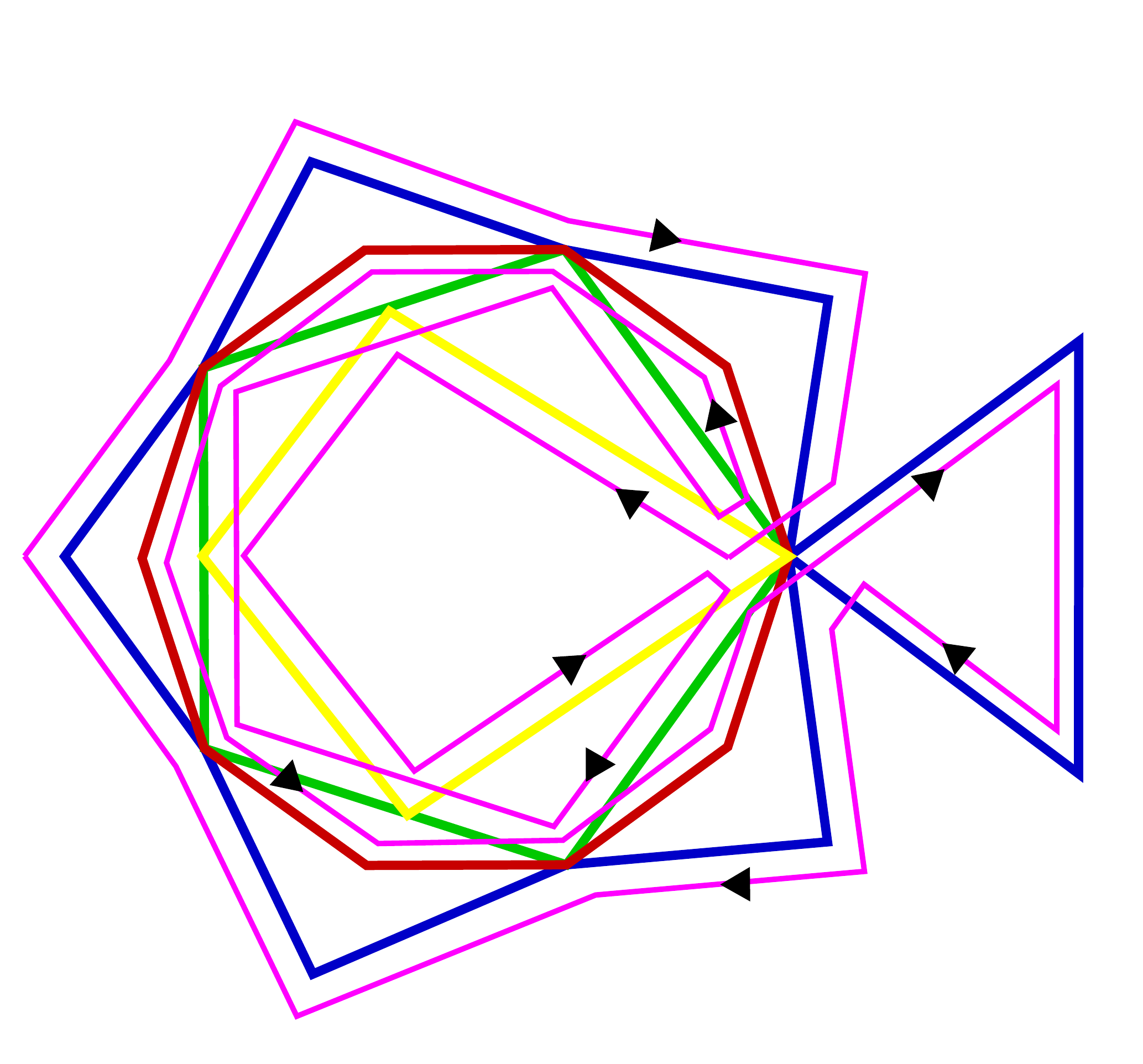}
     \caption{\label{fig:edgerestrictq34}
       Left figure shows a $q=4$ case, with $C_1$ (blue), $C_2$ (red), $C_3$ (green), and $C_4$ (yellow) being the only circuits in a circuit tree sharing vertex $\tv$, where $C_1$ is an ancestor and $C_4$ is a descendant of all circuits in the path $p = \{C_1, C_2, C_3, C_4\}$ on the circuit tree.
       Right figure shows traversal (pink) of edges in $\tK$ of circuits $C_1, C_2, C_3, C_4$ with no crossover where $\te_1 \longrightarrow \te_7, \te_8 \longrightarrow \te_6, \te_5 \longrightarrow \te_3, \te_4 \longrightarrow \te_2$ are edge traversal restrictions.
       Traversal of edges in $C_1$ is clockwise and $C_4$ is counterclockwise for $\tK$.}
   \end{figure}

 \item Repeat Step \ref{itm:traversrestrict} until all edges are marked in the circuit tree.
   
 \item Start by traversing edges of the root circuit in clockwise direction, and follow traversal restrictions.        
\end{enumerate}

\noindent \textit{Complexity:}
We examine $O(|E|)$ edges in $\tK$ to add each edge restriction. 
Since the circuit tree can have at most $(|E|/3) - 1$ edges, the runtime of traversal restriction algorithm is $O(|E|^2)$.

\section{Effect of Extruder Size and Boundary Edges in $\tilde{K}$}\label{sec:boundaryedges}

In our related manuscript \cite{GuKr2018}, we present various measures of geometric quality of the Euler transformed complex $\hK$, and detail how the these measures could be controlled by choosing a small set of user defined parameters.
We present here a couple aspects related to the geometry of $\hK$ that is specific to 3D printing.
In particular, these aspects may affect the generation of continuous print paths.

Recall we denote by $r$ the radius of the extruder.
Depending on the infill lattice and the value of $r$, we could get discontinuous print paths as discussed below.
\begin{itemize}
   \item {\bfseries Edge Covering:}
     An edge in $\hK$ could be covered without traversing it, as shown in Figure \ref{fig:extruder_size_effect}.
     In such cases, we might not be able to generate a continuous print path even when all vertices have even degree.
     In particular, if the mitered offset by at least $r$ of a polygon in undergoes combinatorial changes, then some of its edges can be covered without actually traversing them.
     \begin{figure}[htp!] 
       \centering
       \includegraphics[width=70mm,height=50mm]{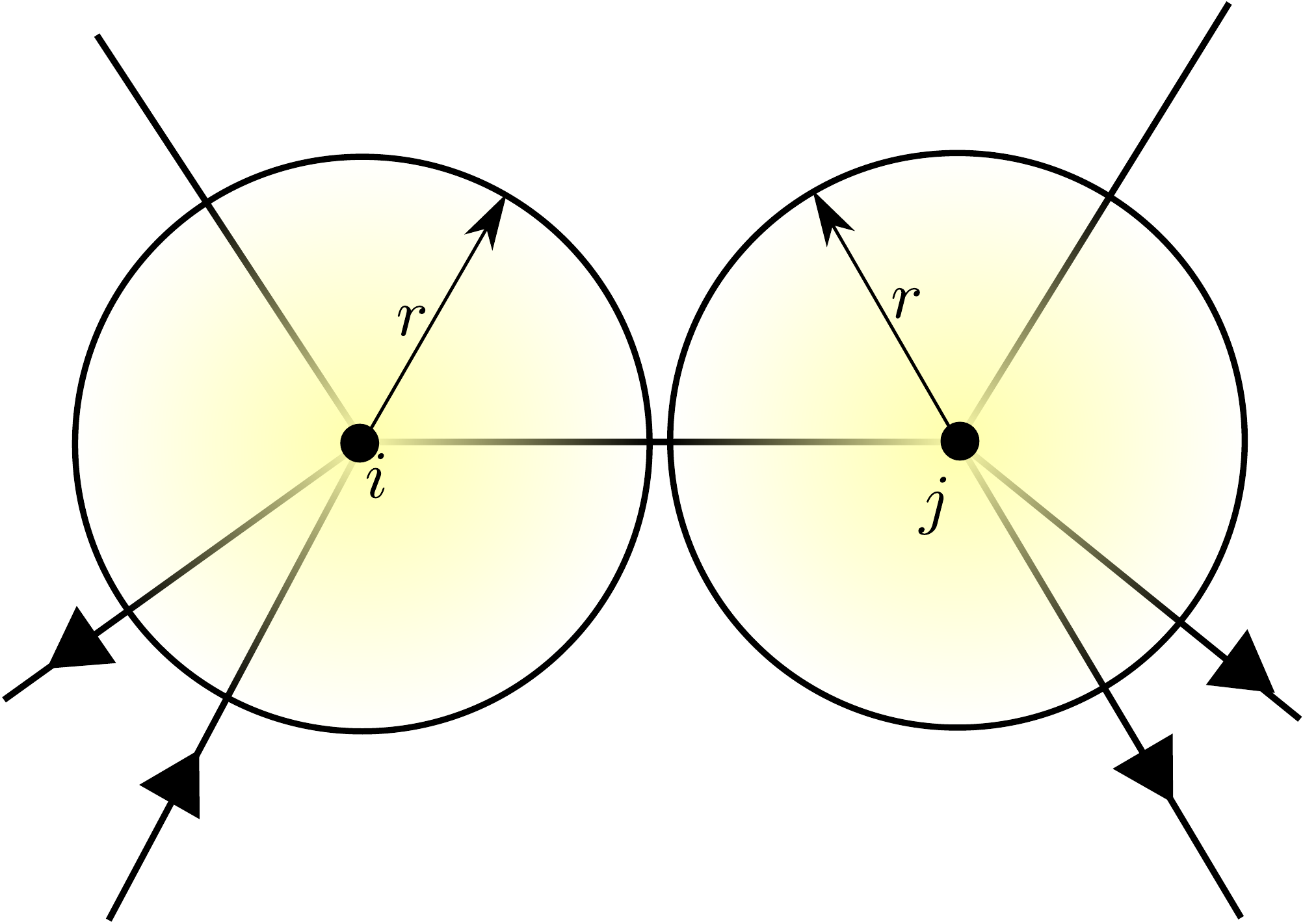}
       \quad\quad
       \includegraphics[width=55mm,height=50mm]{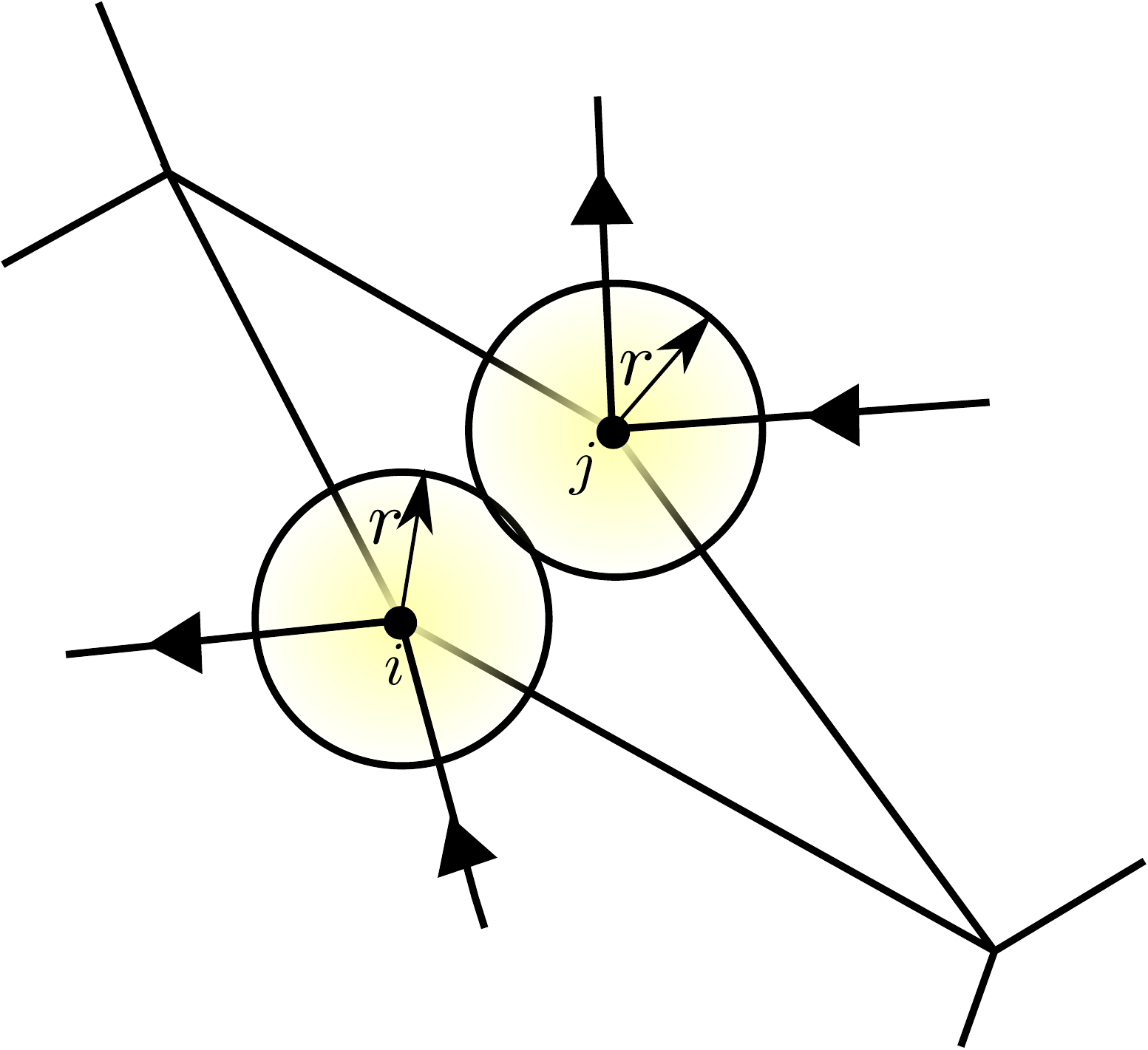}
       \caption{
         Left: Edge $e_{ij}$ is covered by traversing edges other than $e_{ij}$ if $\norm{e_{ij}} \leq 2r + \epsilon$, where $\epsilon$ is zero or too small compared to $r$, the extruder radius.
         Right: We could get material collision if the euclidean distance $d_{ij}$ between points $i$ and $j$ is less than $2r$.}
		
       \label{fig:extruder_size_effect}
     \end{figure}
            
   \item {\bfseries Material Collision:} If the distance between any two points on nonadjacent edges is less than $2r$, we could get material collision as shown in Figure \ref{fig:extruder_size_effect}.
     In particular, if the mitered offset by at least $r$ of polygon undergoes topological changes, then traversing some of its edges can cause material collision. 
\end{itemize} 

We will assume the input complex $K$ is such that $\tK$ is not affected by extruder size in the ways outlined above. 
Still, $\tK$ could have really thin polygons due to clip (\cref{def:clip}) and patch (\cref{def:patch}) operations.
Thin polygons in $\tK$ could make the print path discontinuous in order to prevent material collision.
\begin{defn}{\label{def:shrinkable}}
  \emph{({\bfseries Shrinkable} $2$-cell with offset distance $r$)}
  Let $\tilde{f}$ is a $2$-cell in $\tilde{K}$ and $\bar{f}$ is collection of $2$-cells created after mitered $r$-offset on $\tilde{f}$. $\tilde{f}$ is shrinkable if we can apply mitered $r$-offset such that each vertex and edge in $\bar{f}$ is part of some $2$-cell in $\bar{f}$. If $\bar{f}$ has more than one   $2$-cell then it is shrinkable with topological changes. If a $2$-cell is not shrinkable then it is called unshrinkable.  
\end{defn} 

We assume the traversal of polygons in $\hK$ is not affected by the extruder size ($r$) as outlined above.
This implies interior polygons in $\hat{K}$ are not affected by the extruder size.
But traversal of edges in boundary polygons in $\tilde{K}$ may be affected by extruder size due to clip and patch operations (\cref{def:clip,def:patch}).
\cref{fig:matcola} shows possible material collision in a boundary polygon in $\tilde{K}$, since those $2$-cells become too thin and do not satisfy the assumptions in Step \ref{itm:eulertransform} on Euler Transformation.
We consider the cases where $\tilde{f}$ is shrinkable or unshrinkable based on the extruder size.  

\begin{case}
  $\tilde{f}$ is shrinkable.
\end{case}
 
If $\tilde{f}$ is shrinkable with combinatorial changes then we can print all of its edges.
But if $\tilde{f}$ is shrinkable with topological changes, then we can print only some of the edges.
Let $\tilde{f}$ is shrinkable with topological changes.
This setting also implies $\tilde{f}$ is non-convex.
Let $\tilde{v}_{k}$ be the vertex of edges $\tilde{e}, \tilde{e}'$ in $\tilde{f}$ that is split up into two vertices $\bar{v}_{k}, \bar{v}_{k+1}$ in $\bar{f}$.
Also let $v_{k}', v_{k+1}'$ be the perpendicular points of intersection through $\bar{v}_k, \bar{v}_{k+1}$ on edges $\tilde{e}, \tilde{e}'$.
Then portions of edges between $v_{k}'$ and $v_{k+1}'$ of $\tilde{e}, \tilde{e}'$ have to be set as travel paths in order to avoid material collision.
An example is shown in \cref{fig:matcol}. 

  \begin{figure}[ht!]
	\centering
	\begin{subfigure}[t]{1.4in}
		\centering
		\includegraphics[scale=0.24]{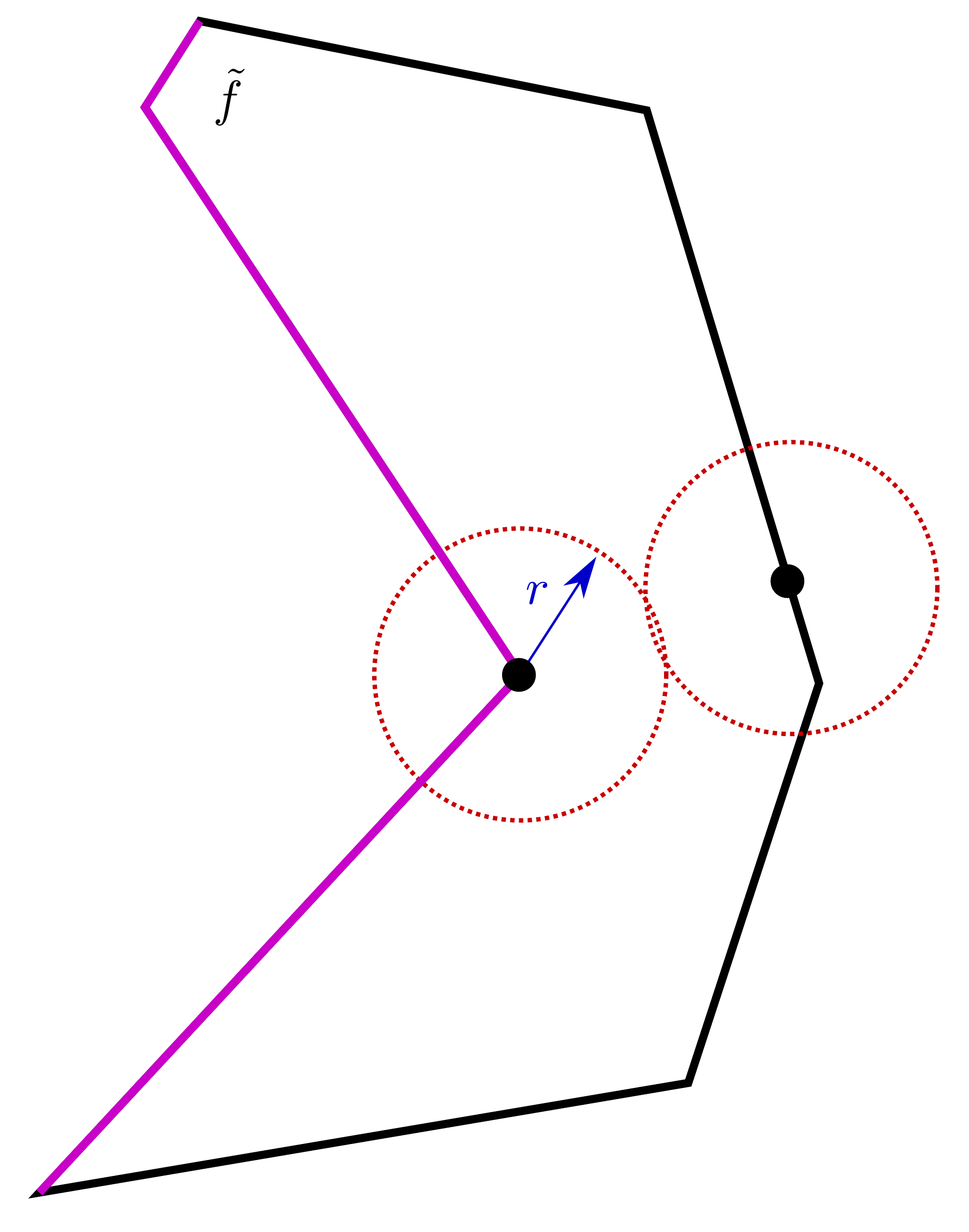}
		\caption{\label{fig:matcola}}
	\end{subfigure}
	\hspace*{.5in}
	\begin{subfigure}[t]{1.4in}
		\centering
		\includegraphics[scale=0.24]{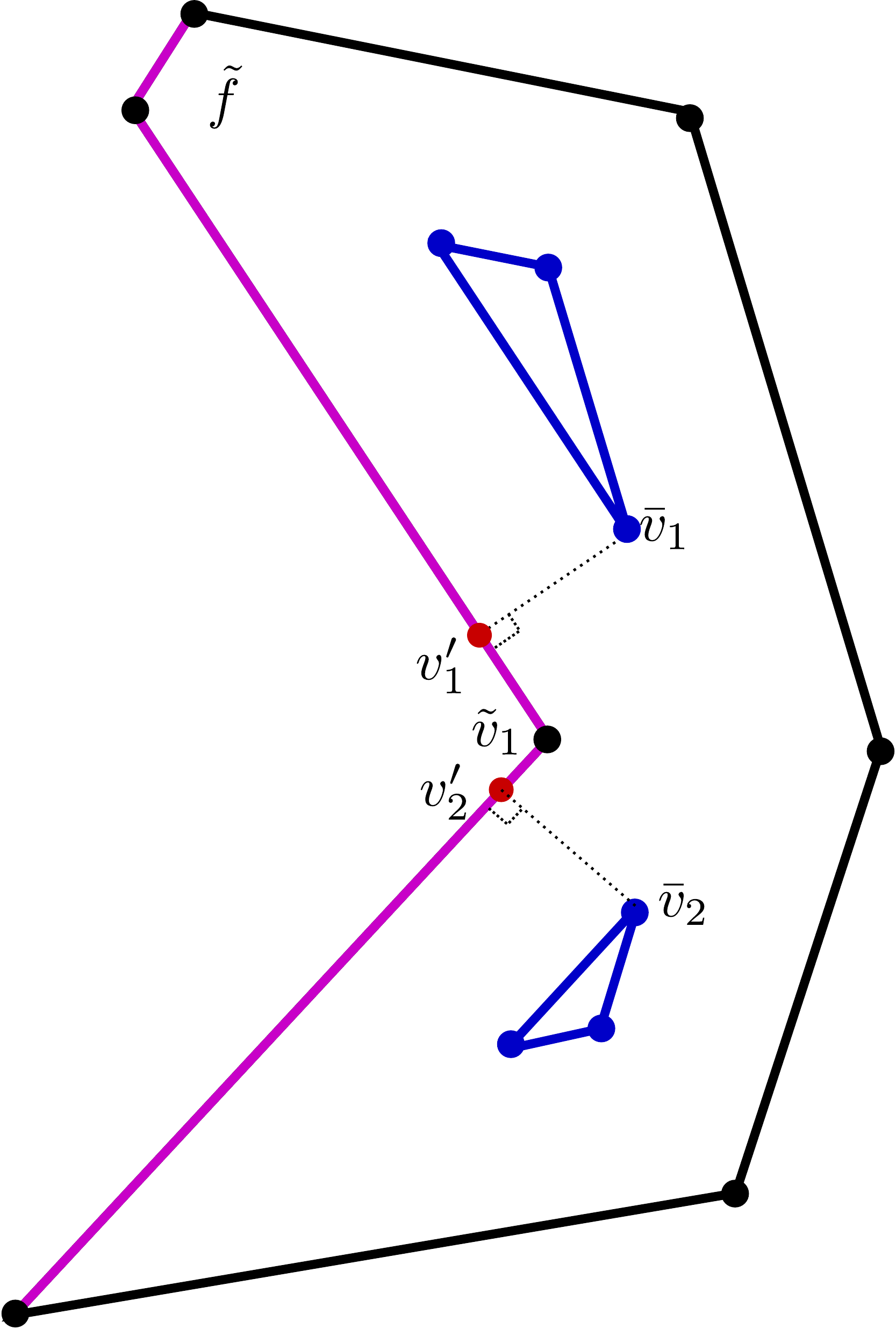}
		\caption{\label{fig:matcolb}}
	\end{subfigure}
	\hspace*{.5in}
	\begin{subfigure}[t]{1.4in}
		\centering
		\includegraphics[scale=0.24]{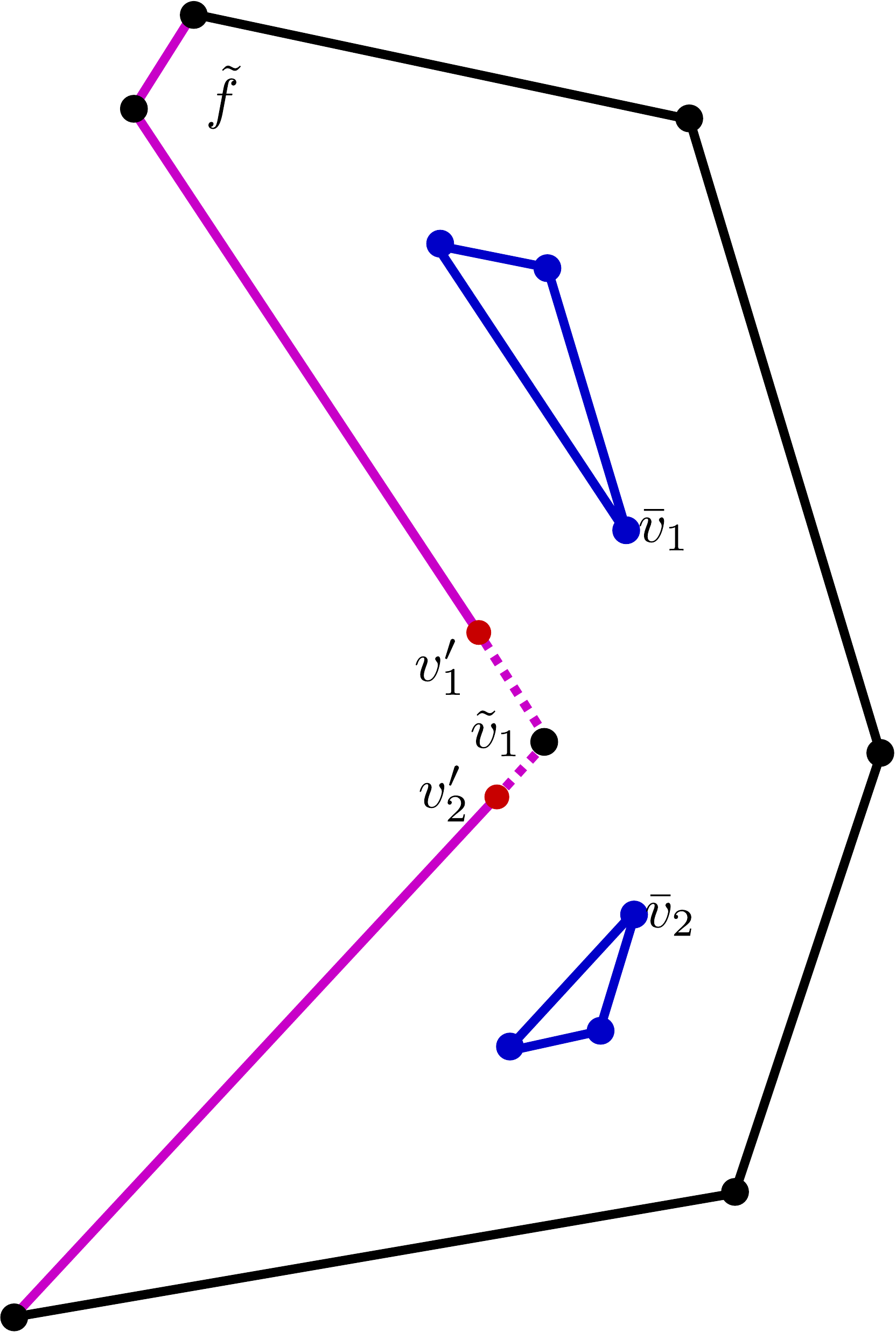}
		\caption{\label{fig:matcolc}}		
	\end{subfigure}
	\caption{\label{fig:matcol}
		Red dotted circle is extruder cross-section of radius $r$ in (a).
		$\tilde{f}$ is $2$-cell of $\tilde{K}$ created in step \ref{itm:patch} shown in (a), (b), (c).
		$\bar{f}$ is a set of $2$-cells created from $\tilde{f}$ (mitered offset), in solid blue in (b), (c).
		The set of $1$-cells created in step \ref{itm:patch} and not portion of any $1$-cell in $\hK$ is in pink.
		$\tilde{v}_1$ is original vertex of $\tilde{f}$ and $v_1', v_2'$ are perpendicular point of intersection through $\bar{v}_1, \bar{v}_2$ to corresponding edges in pink as shown in (b), (c).
		$\{v'_2, \tilde{v}_1\}$, $\{v'_1, \tilde{v}_1\}$ are travel paths in $\tilde{f}$ in (c).
	}
\end{figure}

   \begin{case}
 	$\tilde{f}$ is unshrinkable
 \end{case} 
   Based on our assumptions in Step $\ref{itm:eulertransform}$ on Euler Transformation, printing all the edges or portions of those edges in $\hat{k}$ will not be affected by the extruder size ($r$).
   If $\tilde{f}$ is unshrinkable, then it will be due to boundary edges added in $\tilde{f}$, which are part of the polygon $\tilde{R}_{ij}$ used in the patch operation (\cref{def:patch}) on $\tilde{K}$ (see \cref{fig:unshrinkable}).
   These boundary edges have to be set as travel paths to avoid material collision. 
  
   \begin{figure}[hbp!] 
     \centering
     \vspace*{-0.05in}
     \includegraphics[scale=0.26]{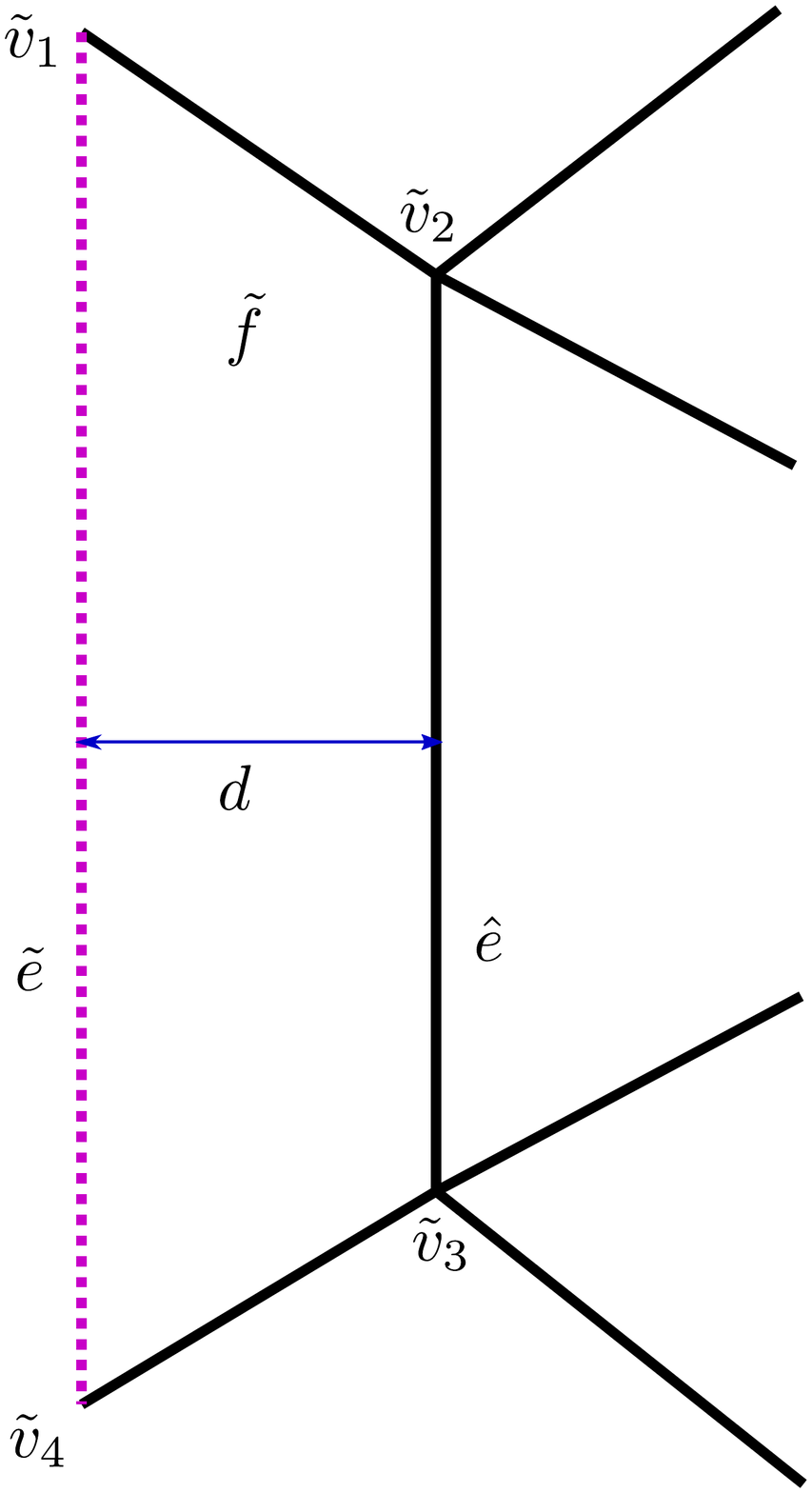}
     \caption{\label{fig:unshrinkable}
       Polygon $\tilde{f}$ is added to $\tilde{K}$ with the creation of edge $\tilde{e}$ in the patch step, where $\tilde{e}, \hat{e}$ are parallel with perpendicular distance $d < 2r$.
       $\tilde{f}$ is unshrinkable with offset distance $r$.
       Hence $\tilde{e}$ has to be set as a travel path.
     }
   \end{figure}

   \noindent We can have at most $m/2$ boundary edges set as travel paths, when there are $m$ boundary edges on $\tilde{R}_{ij}$.

   \begin{rem}
     {\rm It is still an open problem to identify how to continuously print the boundary edges in $\tilde{K}$, if these boundary polygons  in $\tilde{K}$ are shrinkable with topological changes, or are unshrinkable.}
   \end{rem}

\section{Implementation}\label{sec:implem}

We first printed a few proof-of-concept shapes to test our Euler transformation framework.
We printed the shapes shown in \cref{fig:testprts} with 10 layers each.
We then scaled up the jobs to bigger sizes.
The dimensions of the pyramid shown in \cref{printedpyramid} were $609.6\,$mm $\times$ $609.6\,$mm $\times$ $609.6\,$mm, and each layer had a height of $4.26\,$mm, resulting in a total of $143$ layers.

\begin{figure}[htp!]
  \centering
  \includegraphics[height=2.2in]{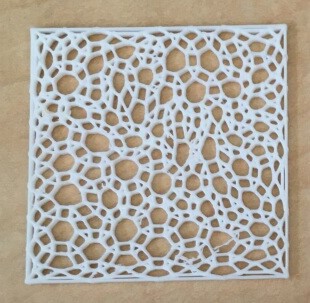}
  \quad
  \includegraphics[height=2.2in]{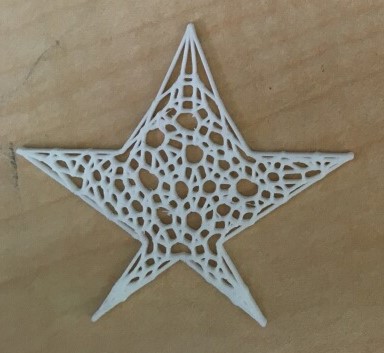}
  \caption{\label{fig:testprts}
    Test prints of a square and a star domain with 10 layers each.
    }
\end{figure}

We illustrate our complete framework on an object with nontrivial geometry and intermediate layer topology---the Stanford bunny.
The height of the bunny is $83$mm. %
We sliced the bunny  into $415$ layers each with height $0.2$mm, and we use an extruder diameter of $0.35$mm.
First we found the minimal square ($86$mm $\times \, 86$mm) containing the union of polygons from every layer.
We triangulated this square using Pymesh.
We then applied Euler transformation on this triangle mesh with a mitered offset of $1$mm.
Finally, we applied clipping and patching operations along with support perimeter, and used the tool path traversal algorithm to generate the infill toolpath for each layer.
We also printed an extra perimeter in each layer to smooth out the surface.
The entire computation ran in 1.5 hrs on a laptop.
See \cref{fig:bunny} for illustrations of salient features of the print object.
%

\begin{figure}[htp!]
  \centering
  \includegraphics[width=3.5in]{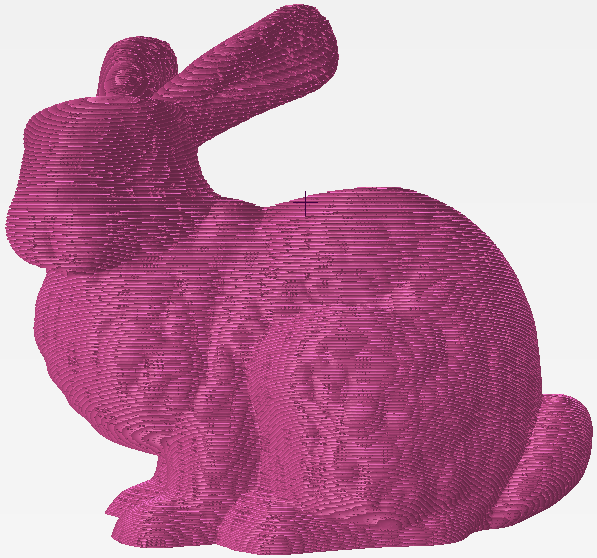}
  \\
  \smallskip
  \includegraphics[width=3.5in]{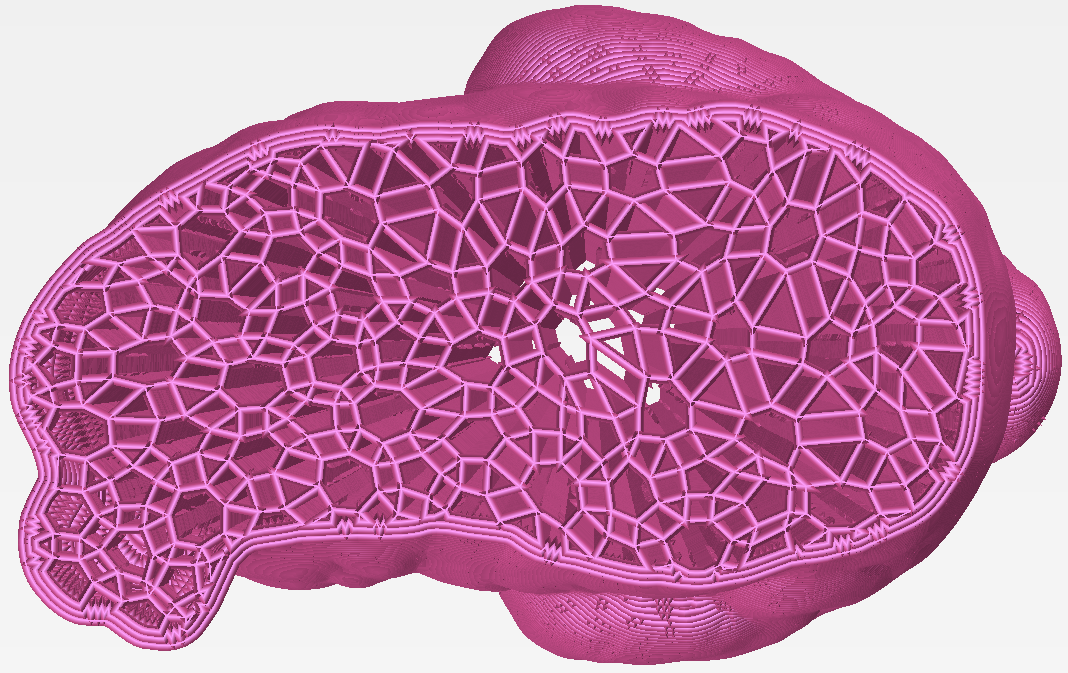}
  \\
  \smallskip
  \includegraphics[width=3.5in]{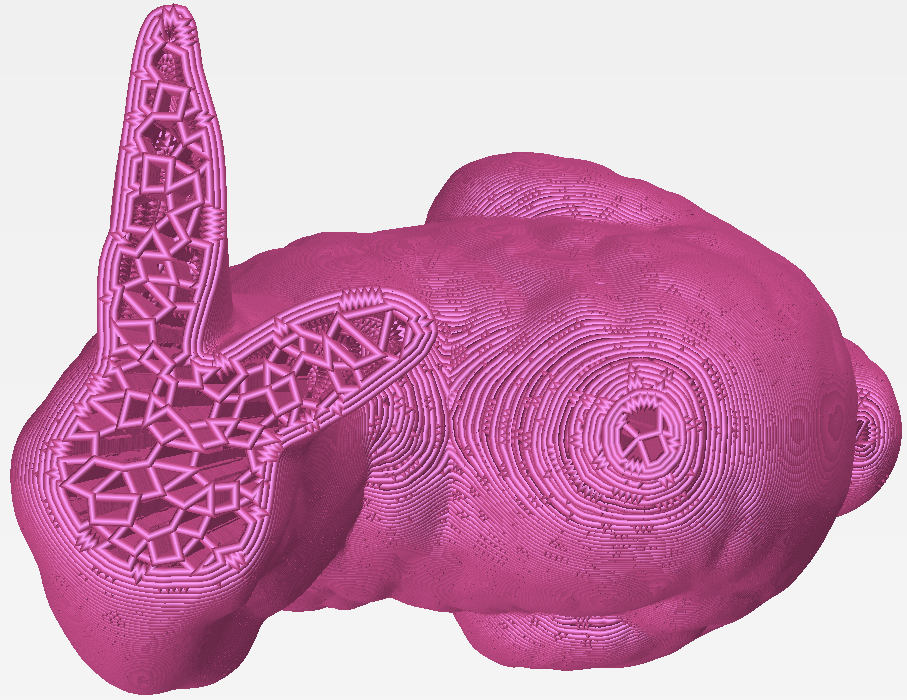}
  \caption{\label{fig:bunny}
    Complete print of the Stanford bunny (top) using 415 layers, and views of layer 212 (middle) and layer 324 (bottom).
    Layer 324 consists of disconnected polygons.
    Support perimeter is evident in the intermediate layers.
  }
\end{figure}

\clearpage
\vspace*{-0.1in}
\section{Discussion}\label{sec:experiment}
\vspace*{-0.1in}

The bottleneck for computational complexity of the Euler transformation is determined by the mitered offsets it creates for each cell in $K$.
The number of cells in $\hK$ are clearly linear in the number of cells in $K$ (\cref{lem:cntshVhEhF2d}).
For $d=2$, if $K$ in $\R^d$ has $m$ $d$-cells, each of which has at most $p$ facets, the time complexity of Euler transformation is $O(m p^d)$ \cite{AuWa2013,AuWa2016}.
Not all cells in the Euler transformation $\hK$ are guaranteed to be convex, even when all cells in $K$ are (see \cref{fig:disjholes}).
We could triangulate the non-convex cells so that all cells in $\hK$ are convex.
But could we do so while maintaining even degrees for all vertices?
A related problem is that of finding a triangulation (rather than a cell complex) of a given domain that minimizes the number of odd-degree vertices.
The total Euclidean length of edges is $\tilde{K}$ is going to be at least double compared to that in the original complex in $K$.
Hence it is better to start with a sparse input complex $K$ (i.e., with a smaller total Euclidean length of edges).
We have described a complete framework for continuous tool path planning in layer-by-layer 3D printing.
The clipping step will be bottlenecked by the computation of intersection of the Euler transformed complex with each polygon in each layer.
We have generalized the Euler transformation defined to allow combinatorial changes when computing mitered offsets of cells.
What about allowing topological changes?
It appears applying the generalized Euler transformation should be able to generate an Euler complex even when topological changes are allowed.
But there might be some new geometric challenges generated in this process, which would have to be taken care of.
We will address this question in future work.
Another promising generalization of our approach would be to \emph{non-planar} 3D printing.
Many of our results should generalize to the non-planar realm as long as underlying support is guaranteed by the design.

\medskip
\noindent {\bfseries Acknowledgments:}
This research was supported in part by an appointment of Gupta to the Oak Ridge National Laboratory (ORNL) ASTRO Program, sponsored by the U.S. Department of Energy (DOE) and administered by the Oak Ridge Institute for Science and Education.
Krishnamoorthy acknowledges partial funding from the National Science Foundation through grants 1661348 and 1819229.
Dreifus acknowledges funding from the Manufacturing Demonstration Facility (MDF) of ORNL and DOE.


\bibliographystyle{plain}

\end{document}